%% file: main.tex
\documentclass[final,letterpaper,11pt]{article}
 
\usepackage{amsmath, amssymb, amsthm, setspace, fancyhdr, lastpage, graphicx, caption, stackengine, longtable ,bbm, indentfirst, enumitem, color,bm, mathtools, relsize}
\usepackage{hyperref}
\hypersetup{colorlinks,citecolor=blue, linkcolor=blue, urlcolor=blue}
\usepackage{psfrag,epsf}
\usepackage{appendix}
\usepackage{minitoc}
\usepackage{makecell}
\usepackage{multirow}
\usepackage{booktabs}
\usepackage{url} 
\usepackage{algorithm}

\newtheorem{thm}{Theorem}[section] 
\newtheorem{lemma}{Lemma}[section] 
\newtheorem{cor}{Corollary}[section]

\newtheorem{remark}{Remark}

\newtheorem{innercustomthm}{Assumption}
\newenvironment{assumption}[1]
{\renewcommand\theinnercustomthm{#1}\innercustomthm}
{\endinnercustomthm}

\newcommand\independent{\protect\mathpalette{\protect\independenT}{\perp}}
\def\independenT#1#2{\mathrel{\rlap{$#1#2$}\mkern2mu{#1#2}}}

\allowdisplaybreaks

\makeatletter
\def\blfootnote{\xdef\@thefnmark{}\@footnotetext}
\makeatother

\newcommand{\rn}[1]{%
	\textup{\expandafter{\romannumeral#1}}%
}
\newcommand{\RN}[1]{%
	\textup{\uppercase\expandafter{\romannumeral#1}}%
}

\newcommand{\E}{\mathbb{E}}
\newcommand{\Pb}{\mathbb{P}}

\newcommand{\var}{\text{Var}}

\newcommand{\Qb}{\mathbb{Q}}
\newcommand{\Pn}{\mathbb{P}_n}

\usepackage[sort,round]{natbib}

\usepackage[normalem]{ulem}
\usepackage{color}
\usepackage{xcolor}

\newcommand{\stout}[1]{\ifmmode\text{\sout{\ensuremath{#1}}}\else\sout{#1}\fi}

\begin{document}
\doparttoc 
\faketableofcontents 

	
	
	\title{\bf {Incremental Intervention Effects \\ in Studies with Dropout \\ and Many Timepoints}}
	\author{Kwangho Kim\thanks{
			\{Department of Statistics \& Data Science, Machine Learning Department\}, Carnegie Mellon University, 5000 Forbes Ave, Pittsburgh, PA 15213. Email: \texttt{kwanghk@cmu.edu}}, \quad
		Edward H. Kennedy\thanks{
			Department of Statistics \& Data Science, Carnegie Mellon University, 5000 Forbes Ave, Pittsburgh, PA 15213. Email: \texttt{edward@stat.cmu.edu}}, \quad
		Ashley I. Naimi\thanks{
			Department of Epidemiology, Rollins School of Public Health, Emory University, Atlanta, GA, USA; E-mail: \texttt{ashley.naimi@emory.edu}}
	}
	\date{}
	\maketitle
	\thispagestyle{empty}

	
	\begin{abstract}
  {
	    Modern longitudinal studies collect feature data at many timepoints, often of the same order of sample size. Such studies are typically affected by {dropout} and positivity violations. We tackle these problems by generalizing effects of recent incremental interventions (which shift propensity scores rather than set treatment values deterministically) to accommodate multiple outcomes and subject dropout. We give an identifying expression for incremental intervention effects when dropout is conditionally ignorable (without requiring treatment positivity), and derive the nonparametric efficiency bound for estimating such effects. Then we present efficient nonparametric estimators, showing that they converge at fast parametric rates and yield uniform inferential guarantees, even when nuisance functions are estimated flexibly at slower rates. We also study the variance ratio of incremental intervention effects relative to more conventional deterministic effects in a novel infinite time horizon setting, where the number of timepoints can grow with sample size, and show that incremental intervention effects yield near-exponential gains in statistical precision in this setup. Finally we conclude with simulations and apply our methods in a study of the effect of low-dose aspirin on pregnancy outcomes. \blfootnote{{Implementation of our method is publicly available at \url{https://github.com/kwangho-joshua-kim/Incremental-dropout}}}
 }
\end{abstract}
	
	\noindent%
	\vfill
	
	\clearpage
    \setcounter{page}{1}
	\newpage
	
	\section{Introduction}
	Causal inference has long been an important scientific pursuit, and understanding causal relationships is essential across many disciplines. However, for practical and ethical reasons, causal questions cannot always be evaluated via experimental methods (i.e., randomized trials), making observational studies the only viable alternative. Further, when individuals can be exposed to varying treatment levels over time, collecting appropriate longitudinal data is important. To that end, recent technological advancements that facilitate data collection are making longitudinal studies with a very large number of time points (sometimes of the same order of sample size) increasingly common \citep[e.g.,][]{kumar2013mobile, eysenbach2011consort, klasnja2015microrandomized}.

    The increase in observational studies with detailed longitudinal data has also introduced numerous statistical challenges that remain unaddressed. For longitudinal causal studies, two analytic frameworks are often invoked: effects of {deterministic fixed interventions} \citep{robins1986, robins2000marginal, hernan2000marginal}, in which all individuals are assigned to a fixed exposure level over all time-points; and effects of {deterministic dynamic interventions} \citep{murphy2001marginal, robins2004optimal} in which, at each time, treatment is assigned according to a fixed rule that depends on past history. In the real world, fixed deterministic interventions might not be of practical interest since the treatment  typically cannot be applied uniformly across a population \citep{Kennedy17}.
    
    Generally, deterministic interventions (fixed or dynamic) rely on a {positivity assumption}, which requires every unit to have a nonzero chance of receiving each of the available treatments at every time point. If the positivity assumption is violated, the causal effect of deterministic (fixed or dynamic) interventions will be no longer identifiable. Even under positivity, longitudinal studies are especially prone to the curse of dimensionality, since exponentially many samples are needed to learn about all treatment trajectories. These issues only worsen when the number of timepoints or covariates increases. Thus, due to a lack of sufficiently flexible analytic methods for longitudinal data, researchers are often forced to either rely on strong parametric assumptions, or forego the estimation of causal effects altogether \citep[e.g.][]{kumar2013mobile}.
    
    {
    One strategy to address such issues in deterministic interventions is to consider stochastic interventions that depend on the observational treatment process and thus are random at each timepoint \citep[e.g.,][]{van2007causal, young2014identification,munoz2012population, haneuse2013estimation, moore2012causal}. Recently, \citet{Kennedy17} proposed novel \textit{incremental intervention effects} which quantify effects of shifting treatment propensities, rather than effects of setting treatment to fixed values. Importantly, incremental effect estimators do not require positivity, and can still achieve $\sqrt{n}$ rates with flexible nonparametric methods. Despite these strengths, the method has not yet been adapted to general longitudinal studies where multiple right-censored outcomes are common (as is common in studies with human subjects). The right-censored outcomes can result in biased estimates of incremental intervention effects unless properly adjusted. This is akin to the well-known concept of confounding bias, and will likely be amplified over time in our case. However, extension to the right censoring setup for incremental intervention effects is not straightforward as, for example, it requires computing new remainder terms to construct the estimators.
    } 
    
    In this paper we propose a more comprehensive form of incremental intervention effects that accommodate not only time-varying treatments, but time-varying outcomes subject to right censoring (i.e., dropout). We provide an identifying expression for incremental intervention effects when dropout is conditionally ignorable, still without requiring (treatment) positivity, and derive the nonparametric efficiency bound for estimating such effects. We go on to present efficient nonparametric estimators, showing that they converge at fast rates and give uniform inferential guarantees, even when nuisance functions are estimated at much slower rates with flexible machine learning tools. Importantly, we study the variance ratio of incremental effects to more conventional deterministic effects in a novel infinite time horizon setting, where the number of timepoints can grow with sample size to infinity. We specifically show that incremental intervention effects can reduce the variance near exponentially, thus yielding extraordinary gains in statistical precision in this setup. Finally, we conduct a simulation study and show that our proposed methods can successfully adjust for subject dropout in incremental intervention effects, and apply our methods to a longitudinal study of the effect of low-dose aspirin on pregnancy outcomes.
	
	\section{Setup}
	
	 We consider a study where for each subject we observe covariates $X_t \in \mathbb{R}^d$, treatment $A_t \in \mathbb{R}$, and outcome $Y_t \in \mathbb{R}$, with all variables allowed to vary over time $t$, but where subjects can drop out or be lost to follow-up. In particular, we consider the case where we want to observe a sample of i.i.d observations $(Z_1, ... , Z_n)$ from a probability distribution $\Pb$ with, for those subjects who remain in the study up to the final timepoint $t=T$, 
    \[ Z = (X_1, A_1, Y_1, X_2, A_2, Y_2, ..., X_T, A_T, Y_T). \]
    But in general we only get to observe
    \begin{align} \label{setup:causal-process}
    Z = \left(X_1, A_1, R_2, R_2(Y_1, X_2, A_2), ... , R_T, R_T(Y_{T-1}, X_T, A_T), R_{T+1}, R_{T+1}Y_{T} \right)
    \end{align}
    where 
    $R_t = \mathbbm{1} \text{\{ still in the study at time t\}}$ is 
     an indicator for whether the subject contributes data at time $t$. We write $R_t(Y_{t-1}, X_t, A_t)$ as a shorthand for $(R_tY_{t-1}, R_tX_t, R_tA_t)$, so in the missingness process that we consider, subjects can drop out at each time after the measurement of covariates/treatment. This is motivated by the fact that this is likely the most common type of dropout, since outcomes $Y_t$ at time $t$ are often measured together with or just prior to covariates $X_{t+1}$ at time $t+1$. As we consider a monotone dropout (i.e., right-censoring) process, $R_t$ is non-increasing in time $t$, i.e.,
    \begin{align*}
    & 
    \begin{cases}
    R_t = 1 \ \Rightarrow & {(R_1,...,R_{t-1}) = \bm{1} } \\
    R_t = 0 \ \Rightarrow & {(R_{t+1},...,R_T) = \bm{0} }, 
    \end{cases}
    \end{align*}
    where $\bm{0}, \bm{1}$ are vectors of zeros and ones respectively. Thus our data structure $Z$ is a chain with $t$-th component
    $$
    \left\{R_t, R_t(Y_{t-1}, X_t, A_t)\right\}
    $$
    for $t = 1,...,T+1$, where $R_1=1$ and we do not use $Y_{0}$ or $X_{T+1}, A_{T+1}$. Although we suppose each subject's dropout will occur before the $t$-th stage, our data structure also covers the case when the dropout will occur after the $t$-th stage because in that case we can write 
    $$
    \left\{R_t(Y_{t-1}, X_t, A_t), R_{t+1} \right\}
    $$
    as the $t$-th component of our chain. 
    
    For simplicity, we consider binary treatment in this paper, so that the support of each $A_t$ is $\mathcal{A} = \{0,1\}$. We use overbars and underbars to denote all the past history and future event of a variable respectively, so that $\overline{X}_t = (X_1, ... , X_t)$ and $\underline{A}_t = (A_t, ... , A_T)$ for example. We also write $H_t = ( \overline{X}_t, \overline{A}_{t-1}, \overline{Y}_{t-1})$ to denote all the observed past history just prior to receiving treatment at time $t$, with support $\mathcal{H}_t$. Finally, we use lower-case letters $a_t, h_t, x_t$ to represent realized values for $A_t, H_t, X_t$, unless stated otherwise.
    
    Now that we have defined our data structure we turn to our estimation goal, i.e., which treatment effects we aim to estimate. We use  $Y_t^{\overline{a}_t}$ to denote the potential (counterfactual) outcome at time $t$ that would have been observed under a treatment sequence $\overline{a}_t=(a_1,...,a_t)$ (note that we have $Y_t^{\overline{a}_T}=Y_t^{\overline{a}_t}$ as long as the future cannot cause the past). In longitudinal causal problems it is common to pursue quantities such as $\E(Y_t^{\overline{a}_{t}})$, i.e., the mean outcome at a given time under particular treatment sequences $\overline{a}_t$; for example one might compare the mean outcome under $\overline{a}_t=\bm{1}$ versus $\overline{a}_t=\bm{0}$, which represents how outcomes would change if all versus none were treated at all times. However identifying these effects requires strong positivity assumptions (i.e., that all have some chance at receiving every treatment at every time), and estimating these effects often requires untenable parametric assumptions especially when $t \gg 1$. 
    
    Following \citet{Kennedy17} we instead consider incremental intervention effects, which represent how mean outcomes would change if the odds of treatment at each time were multiplied by a factor $\delta$ (e.g., $\delta=2$ means odds of treatment are doubled). Incremental interventions shift propensity scores rather than impose treatments themselves; they represent what would happen if treatment were gradually more or less likely to be assigned, relative to the natural/observational treatment, in the population. Since they are `population-level' effects, they are useful for giving an interpretable picture to understand the overall societal effects, but will likely be less useful than classical deterministic effects for making specific recommendations about optimal treatment. Nonetheless, there are a number of benefits of studying incremental intervention effects: for example, positivity assumptions can be entirely and naturally avoided; complex effects under a wide range of intensities can be summarized with a single curve in $\delta$, no matter how many timepoints $T$ there are; and they more closely align with actual intervention effects than their fixed treatment regime counterparts. We refer to \citet{Kennedy17} for more discussion and details on the tradeoff between deterministic and incremental intervention effects.

    Formally, incremental interventions are dynamic stochastic interventions where treatment is assigned based on new interventional propensity scores defined by
    \begin{align}  \label{eqn:incr-intv-ps}
    q_t(h_t; \delta,\pi_t) = \frac{\delta\pi_{t}(h_t)}{\delta\pi_{t}(h_t) + 1 - \pi_{t}(h_t)},
    \end{align}    
    not the observational propensity scores $\pi_t(h_t) = \Pb(A_t=1 \mid H_t=h_t)$. In other words, $q_t$ is a shifted version of $\pi_t$ obtained by multiplying the odds of receiving treatment by $\delta$.
    We denote potential outcomes under the above intervention as $Y_{t}^{\overline{Q}_t(\delta)}$ where $\overline{Q}_t(\delta) = \{Q_1(\delta), ... , Q_{t}(\delta)\}$ represents a sequence of draws from the conditional distributions $Q_s(\delta) \mid H_s=h_s \sim \text{Bernoulli}\{q_s( h_s; \delta, \pi_s)\}$, $s=1,...,t$. We often drop $\delta$ and write $Q_t = Q_t(\delta)$ when the dependence is clear from the context. Note here we use capital letters for the intervention indices since they are random, as opposed to $Y_t^{\overline{a}_t}$ where the intervention is deterministic. Therefore in this paper, we aim to estimate the mean counterfactual outcome 
    $$ \psi_t(\delta) = \E\left(Y_t^{\overline{Q}_t(\delta) } \right) $$
    for any $t\leq T$. In the next section we describe the necessary conditions for identifying $\psi_t(\delta)$ in the presence of dropout.
	
	{
	\begin{remark}
	 To be precise, the incremental effect $\psi_t(\delta)$ is the compounding effect by the two different changes. Consider only the first two timepoints. In this case the propensity score under the incremental intervention at the later timepoint will be different from its observational value for two reasons: 1) $\delta$ multiplied to the propensity scores, and 2) covariates at the earlier timepoint that have been changed by the resultant (incremental) intervention. With many timepoints, in a long term, these effects are compounded over time and just manifested as a single number of the incremental effect. This nuance stems from the nature of incremental interventions, i.e., the way they depend on the observational treatment process through $q_t$. 
	\end{remark}
	}

	\section{Identification}
	In this section, we will give assumptions under which the entire marginal distribution of the  counterfactual outcome $Y_t^{\overline{Q}_t(\delta)}$ is identified. Specifically, we require the following assumptions for all $t \leq T$.
	
	\begin{assumption}{A1} \label{assumption:A1} 
	$Y_t = Y_t^{\overline{a}_t}$ if $\overline{A}_t = \overline{a}_t$
	\end{assumption}
	\begin{assumption}{A2-E} \label{assumption:A2-E} 
	$A_{t^\prime} \independent Y_{t}^{\overline{a}_{t}} \mid H_{t^\prime}$, $\forall t^\prime \leq t$
	\end{assumption}
	\begin{assumption}{A2-M} \label{assumption:A2-M} 
	$R_t \independent (\underline{X}_t, \underline{A}_t, \underline{Y}_{t-1}) \mid H_{t-1}, A_{t-1}, R_{t-1}=1$
	\end{assumption}
	\begin{assumption}{A3} \label{assumption:A3}
	$\Pb(R_{t} = 1 \mid H_{t-1}, A_{t-1}, R_{t-1}=1) > \epsilon_\omega$ for some $\epsilon_\omega > 0$ a.s.
	\end{assumption}

	Assumptions (\ref{assumption:A1}) and (\ref{assumption:A2-E}) correspond to consistency and exchangeability (or sequential ignorability) respectively, which are commonly adopted in the literature. Consistency means that the observed outcomes are equal to the corresponding potential outcomes under the observed treatment sequence, and would be violated in settings with interference for example. Exchangeability means that the treatment and counterfactual outcome are independent, conditional on the observed past (if there were no dropout), i.e., that treatment is as good as randomized at each time conditional on the past. Experiments ensure that exchangeability holds by construction. 
	
	In our work, we additionally require assumptions (\ref{assumption:A2-M}) and (\ref{assumption:A3}) because of the missingness/dropout. (\ref{assumption:A2-M}) is the standard time-varying missing-at-random (MAR) assumption for monotone missingness, ensuring that dropout is independent of the future conditioned on the observed history up to the current time point \citep[e.g.,][]{national2010prevention, robins1995analysis, van2003unified}. One may think of this type of MAR assumption as a sequentially random dropout process, where the decision to drop out at time $t$ is like the flip of a coin, with probability of ‘heads’ (dropout) depending only on the measurements recorded through time $t-1$ \citep[][Chapter 4]{national2010prevention}. This would be a reasonable assumption if we can collect enough data to explain the dropout process, so we can ensure that those who dropout look like those who do not, given all past observed data. (\ref{assumption:A3}) is a positivity assumption for missingness, meaning that each subject in the study has some non-zero chance at staying in the study at the next timepoint. This would be expected to hold in many studies, but may not if some subjects are `doomed' to drop out based on their specific measured characteristics. 
	
	
	Importantly, here we do not require positivity conditions on the propensity scores as we are targeting the effects $\psi_t(\delta)$ with the incremental intervention $q_t$ defined in (\ref{eqn:incr-intv-ps}), not deterministic effects. The next result gives an identifying expression for $\psi_t(\delta)$ under the above assumptions.
	
	\begin{thm} \label{thm:ident-exp}
		Suppose identification assumptions (\ref{assumption:A1}) - (\ref{assumption:A3}) hold. Then for all $t \leq T$, the incremental effect on outcome $Y_t$ with given value of $\delta \in [\delta_l, \delta_u]$,   $0 < \delta_l \leq \delta_u < \infty$, equals
		\begin{equation} \label{eqn:ident-exp}
			\begin{aligned}
				& \psi_t(\delta) 
				& = \underset{\overline{\mathcal{X}}_{t}\times \overline{\mathcal{A}}_{t}}{\int} \mu(h_{t},a_{t}, R_{t+1}=1)  \prod_{s=1}^{t} q_s(a_s \mid h_s, R_s=1) d\nu(a_s) \ d\Pb(y_{s-1}, x_s \mid h_{s-1},a_{s-1}, R_s=1), 
			\end{aligned}
		\end{equation}
		where $\overline{\mathcal{X}}_{t} = \mathcal{X}_1 \times \cdots \times \mathcal{X}_t$, $\overline{\mathcal{A}}_{t} = \mathcal{A}_1 \times \cdots \times \mathcal{A}_t$, \\
		$\mu(h_{t},a_{t}, R_{t+1}=1) = \E(Y_{t} \mid H_{t} = h_{t}, A_{t} = a_{t}, R_{t+1}=1)$, and
		\begin{equation} \label{eqn:incremental-density}
			q_s(a_s \mid h_s, R_s=1)=\frac{a_s\delta\pi_{s}(h_s, R_s=1) + (1-a_s) \{ 1 - \pi_s(h_s, R_s=1) \} }{\delta\pi_{s}(h_s, R_s=1) + 1 - \pi_{s}(h_s, R_s=1)}.
		\end{equation}
		Here, $\pi_{s}(h_s, R_s=1)=\Pb(A_s=1 \mid H_s=h_s,R_s=1)$ and $\nu$ is some dominating measure for the distribution of $A_s$.
	\end{thm}
	
	When we derive the identification result in Theorem \ref{thm:ident-exp}, as in \citet{Kennedy17} we use the g-formula \citep{robins1986} where we put in the incremental intervention for the treatment distribution and a point mass for the right-censoring indicator of $1$, followed by applying the identification lemma (Lemma \ref{lem:identification} in the appendix) under the additional assumptions (\ref{assumption:A2-M}) and (\ref{assumption:A3}). The next corollary illustrates what this identification result gives in the simple point-exposure study. 
	
	\begin{cor} \label{cor:ident-exp-pt-trt}
		When $T=1$, the data structure reduces to
		\begin{align*}
			Z=(X, A, R, RY),
		\end{align*}    
		thus in this case $R=1$ means the outcome is not missing. Then the identifying expression simplifies to
		\begin{equation*} 
			\begin{aligned}
				&\psi(\delta) = \E \left[ \frac{\delta\pi(X)\mu(X,1,1) + \{1 - \pi(X)\}\mu(X,0,1) }{\delta\pi(X) + \{1 - \pi(X)\}} \right]
			\end{aligned}
		\end{equation*}
		where $\pi(X) = \Pb(A=1 \mid X)$ and $\mu(x,a,1) = \E(Y \mid X = x, A = a, R=1)$.
	\end{cor}
	
	Therefore when $T=1$, the effect $\psi(\delta)$ is simply a weighted average of the two regression functions $\mu(X,1,1)$, $\mu(X,0,1)$ among those with observed outcomes, with weights depending on the  propensity scores and $\delta$.

	\section{Efficiency Theory} \label{sec:efficiency-theory}
	
	In the previous section, we showed that the incremental intervention effect adjusted for subject dropout can be identified without requiring any positivity conditions on the treatment process. Our main goal in this section is to develop a nonparametric efficiency theory based on the efficient influence function for $\psi_t(\delta)$.
	
	The efficient influence function plays a crucial role in non/semiparametric efficiency theory because 1) its variance gives an asymptotic efficiency bound, and 2) its form indicates how to do an appropriate bias correction in order to construct estimators that attain such efficiency bound. Mathematically, given a target parameter $\psi$ an \textit{influence function} $\phi$ acts as the derivative term in a distributional analog of a Taylor expansion, which can be seen to imply
	\begin{equation} \label{def:influence-function-1}
    \frac{\partial \psi(\Pb_\epsilon)}{\partial\epsilon} \Big\vert_{\epsilon=0} = \int \phi(z;\Pb) \left( \frac{\partial \log d\Pb_\epsilon(z)}{\partial\epsilon} \right)\Big\vert_{\epsilon=0} \ d\Pb(z)
	\end{equation}
	for all smooth parametric submodels $\Pb_\epsilon$ containing the true distribution at $\epsilon=0$, i.e., $\Pb_{\epsilon=0}=\Pb$. Of all the influence functions, the \textit{efficient influence function} is defined as the one which gives the greatest lower bound of all parametric submodel $\Pb_\epsilon$, so giving the efficiency bound for estimating $\psi$. For more details we refer to Section \ref{sec:ifdetails} in the appendix and references therein \citep{Bickel98,vaart98, van2003unified, Tsiatis06, Kennedy16}.
	
	The next theorem gives an expression for the efficient influence function for our incremental effect $\psi_t(\delta)$, under a nonparametric model. 
	
	\begin{thm} \label{thm:eif}
		The (uncentered) efficient influence function for the intervention effect $\psi_t(\delta)$, $\forall t \leq T$, is given by
		\begin{equation*} 
			\begin{aligned}	
				& \sum_{s=1}^{t} \left( \frac{1}{\delta A_s + 1-A_s}  \right) \Bigg[ \frac{\left\{m_{s}(H_{s}, 1) -m_{s}(H_{s}, 0)\right\}\delta(A_s- \pi_{s}(H_{s}))\omega_s( H_s, A_s)}{\delta\pi_{s}(H_{s}) + 1 - \pi_{s}(H_{s})} \\ 
				& \quad + \begin{pmatrix}
                \delta m_{s}(H_{s}, 1)\left\{\pi_{s}(H_{s})\omega_s( H_s, A_s)-A_sR_{s+1}\right\}  \\
                + m_{s}(H_{s}, 0)\left\{(1-\pi_{s}(H_{s}))\omega_s( H_s, A_s) - (1-A_s)R_{s+1} \right\}
                \end{pmatrix} \Bigg] \\
                & \quad \times \prod_{k=1}^{s} \left\{ \frac{\delta A_k + 1-A_k}{\delta\pi_k(H_k) + 1-\pi_k(H_k)} \cdot\frac{R_{k}}{\omega_{k}( H_{k}, A_{k})} \right\} 
				+ \prod_{s=1}^{t} \left\{ \frac{\delta A_s + 1-A_s}{\delta\pi_s(H_s) + 1-\pi_s(H_s)}\cdot\frac{R_{s}}{\omega_s( H_s, A_s)} \right\}Y_t R_{t+1},
			\end{aligned}
		\end{equation*}  
		where $\pi_s(h_s) = \Pb(A_s=1 \mid H_s=h_s, R_s=1)$, $\omega_s( H_s, A_s) = d\Pb(R_{s+1}=1 \mid H_s, A_s,R_s=1)$, and
		\begin{align*}
			m_s&(h_s,a_s, R_{s+1}=1) \\
			& = \int_{ \mathcal{R}_s} \mu(h_{t},a_{t}, R_{t+1}=1)  \prod_{k=s+1}^{{t}} q_k(a_k \mid h_k, R_k=1) d\nu(a_k) d\Pb(y_{k-1},x_k|h_{k-1},a_{k-1}, R_k=1) 
		\end{align*} 
		for $\forall s \leq t$. Here $\mathcal{R}_s = (\overline{\mathcal{X}}_{t}\times \overline{\mathcal{A}}_{t}) \setminus (\overline{\mathcal{X}}_{s}\times \overline{\mathcal{A}}_{s})$, $\mu(h_{t},a_{t}, R_{t+1}=1) = \E(Y_t \mid H_{t} = h_{t}, A_{t} = a_{t}, R_{t+1}=1)$, and $\nu$ is a dominating measure for the distribution of $A_k$.
	\end{thm}
	
	The proof is given in Appendix \ref{proof:thm-eif}. This result will be used to construct an efficient, model-free estimator for our new incremental intervention effects in the next section. In Theorem \ref{thm:eif} all terms are to be estimated via regression tools or simply obtained from the observed data. Note that we have new weighting terms such as $\frac{\mathbbm{1}\left(R_{s}=1\right)}{\omega_s( H_s, A_s)}$ that are used to adjust for dropout effects at each stage $s \leq t$. As one may expect, should all the data be fully observed (i.e., $\Pb[R_t=0]=1$ a.e $[\Pb]$ for all $t\leq T$), both the identifying expression and efficient influence function will reduce to the formulas presented in \citet{Kennedy17}.  
	
	\begin{remark}
	     Although we derived the above efficient influence function from first principles, based on the  pathwise differentiability in \eqref{def:influence-function-1}, it could equivalently be derived 
	     using results on mapping complete- to observed-data influence functions under general coarsening at random \citep[e.g.,][]{robins1994estimation, robins1995analysis, van2003unified, Tsiatis06}. However, in either case, computing error bounds requires the derivation of the second-order remainder terms in  von Mises expansion, which is new in our work and not immediate from the earlier results.
	\end{remark}

	The above efficient influence function involves three types of nuisance functions: the treatment propensity scores $\pi_s(H_s)$, the missingness/dropout propensity scores $\omega_s( H_s, A_s)$, and the psuedo outcome regression functions $m_s(H_s,A_s,R_{s+1}=1)$, $\forall s\leq t$. As in \citet{Kennedy17}, each $m_s$ can be estimated through sequential regressions without resorting to complicated conditional density estimation, since they are marginalized versions of the full regression functions $\mu(h_{s},a_{s}, R_{s+1}=1)$ that condition on all in the past. We give the sequential regression formulation for $m_s$ in Appendix \ref{sequential-regression-formulation}.
	
	The efficient influence function corresponding to $T=1$ follows a relatively simple and intuitive form, equaling a weighted average of the efficient influence functions for $\E(Y^1)$ and $\E(Y^0)$ plus contributions from the propensity scores $\omega_s, \pi_s$. We give this result in Appendix \ref{eif-for-T=1} as well.

	\section{Estimation and Inference}
	
	\subsection{Proposed Estimator} 
	\label{subsec:proposed-estimator}
	In this section we develop an estimator that can attain fast $\sqrt[]{n}$ rates, even when other nuisance functions are modeled nonparametrically and estimated at slower rates. 
	
	To begin, let $\varphi(Z;\bm{\eta},\delta, t)$ denote the uncentered efficient influence function from Theorem \ref{thm:eif}, which is a function of $Z$, indexed by a set of nuisance functions
	$$
	\bm{\eta} = (\bm{\pi, m, \omega}) = \left(\pi_1,...,\pi_{t}, m_1,...,m_{t}, \omega_1,...,\omega_{t} \right),
	$$
	$\delta$, and $t \leq T$, where $\pi_t, m_t, \omega_t$ are the same nuisance functions defined in Theorem \ref{thm:eif}. 
	
	Since $\E[\varphi(Z;\bm{\eta},\delta, t)] = \psi_t (\delta)$, a natural estimator would be the naive plug-in $Z$-estimator
	$$
	\hat{\psi}_{inc.pi}(t; \delta) = \Pb_n \{ \varphi(Z;\hat{\bm{\eta}},\delta, t) \} 
	$$
	where $\hat{\bm{\eta}}$ represents a set of nuisance function estimates and $\Pb_n$ denotes the empirical
	measure so that sample averages can be written by $\frac{1}{n}\sum_i f(Z_i) = \Pb_n\{f(Z)\} = \int f(z) d\Pb_n(z)$.
	
	If we assume $\pi_{t}$ and $\omega_{t}$ were correctly parametrically modeled, then one could use the following simple inverse-probability-weighted (IPW) estimator
	$$
	\hat{\psi}_{inc.ipw}(t; \delta) = \Pb_n \left\{\prod_{s=1}^{t} \left( \frac{\delta A_s + 1-A_s }{\delta\hat{\pi}_s(H_s) + 1-\hat{\pi}_s(H_s)} \cdot \frac{\mathbbm{1}\left( R_{s+1}=1 \right)}{\hat{\omega}_{s}(H_s, A_s)} \right)Y_t \right\}.
	$$
	Note that this IPW estimator is a special case of $\hat{\psi}_{inc.pi}$ where $\hat{m}_s$ is set to zero for all $s \leq t$.
	
	However, the above inverse-weighted or plug-in $Z$-estimators typically require both strong parametric assumptions and empirical process conditions (e.g., Donsker-type or low entropy conditions) that restrict the flexibility of the nuisance estimators. Especially, the latter is due to using the data twice (once for estimating the nuisance functions, again for estimating the bias, i.e., the average of the uncentered influence function), thus can cause overfitting. To avoid this downside and make our estimator more practically useful, here we use sample splitting \citep{zheng10,chernozhukov16double,Kennedy17,robins2008estimation}. As will be seen shortly, sample splitting allows us to avoid complex empirical process conditions even when all the nuisance functions $\bm{\eta}$ are arbitrarily flexibly estimated. Further, bias-corrected influence function-based estimators allow us to withstand slower rates for nuisance estimation while attaining faster rates for estimation of the parameter of interest. 
	
	Now we give an algorithm allowing slower than $\sqrt{n}$ rates and non-Donsker complex nuisance estimation as follows. First, we randomly split the observations $(Z_1, ..., Z_n)$ into $K$ disjoint groups, using a random variable $S_i$, $i=1,...,n$, drawn independently of the data, where each $S_i \in \{1,...,K\}$ denotes the group membership for unit $i$. Then our proposed estimator is given by
	\begin{equation} \label{eqn::proposed-estimator}
		\widehat{\psi}_t (\delta) = \Pb_n\left\{\varphi(Z;\hat{\bm{\eta}}_{-S},\delta, t)\right\}  \equiv \frac{1}{K}\sum_{k=1}^{K} \Pb_n^{(k)} \{ \varphi(Z;\hat{\bm{\eta}}_{-k},\delta, t) \}
	\end{equation}
	where we let $\Pb_n^{(k)}$ denote sample averages only over a group $k$, i.e., $\{i : S_i = k\}$, and let $\hat{\bm{\eta}}_{-k}$ denote the nuisance estimator constructed excluding the group $k$. We detail exactly how to compute the proposed estimator $\widehat{\psi}_t(\delta)$ in Appendix \ref{sec:algorithm}. 
	
	Our methods effectively utilize all the observed samples available at each time, without any need for discarding a subset of observed sample in advance. It is also worth noting that our algorithm is amenable to parallelization due to the sample splitting.

	\subsection{Asymptotic Theory}

	This subsection is devoted to characterizing an asymptotic behavior of our proposed estimator, that $\widehat{\psi}_t(\delta)$ is
	$\sqrt[]{n}$-consistent and asymptotically normal even when the nuisance functions are estimated nonparametrically at much slower than $\sqrt{n}$ rates. 
	
	In what follows we denote the $L_2(\Pb)$ norm of function $f$ by $\Vert f \Vert = \left(\int f(z)^2 d\Pb(z) \right)^{1/2}$, to distinguish it from the ordinary $L_2$ norm $\Vert \cdot \Vert_2$ for a fixed vector. Also note that although we used $m_s$ to denote the pseudo-regression function defined in Theorem \ref{thm:eif}, in principle they are indexed by both the time $s$ and increment parameter $\delta$ as in $m_{s,\delta}$. The next theorem shows uniform convergence of $\hat{\psi}_t (\delta)$, which lays the foundation for subsequent statistical inferential and testing procedures.
	
	\begin{thm} \label{thm:convergence}
		Define the variance function as
		$\sigma^2(\delta, t)=\E\left[\left(\varphi(Z;\bm{\eta},\delta, t) - \psi_t(\delta) \right)^2 \right]$ and let $\hat{\sigma}^2(\delta, t)=\Pb_n\left[\left(\varphi(Z;\hat{\bm{\eta}}_{-S},\delta, t) - \hat{\psi}_t(\delta) \right)^2 \right]$ denote its estimator. Assume:
		\begin{itemize}
			\item [1)] The set $\mathcal{D}=[\delta_l, \delta_u]$ is bounded with $0 < \delta_l \leq \delta_u < \infty$.
			\item [2)] $\Pb\left[ \mid m_s(H_s, A_s, R_{s+1}=1) \mid \leq C \right]= \Pb\left[ \mid \hat{m}_s(H_s, A_s, R_{s+1}=1) \mid \leq C \right] = 1$, $\forall s \leq t$, for some constant $C<\infty$.
			\item [3)] $\sup_{\delta \in \mathcal{D}} \big| \frac{\hat{\sigma}^2(\delta, t)}{\sigma^2(\delta, t)} -1 \big| = o_\Pb(1)$, and $\| \sup_{\delta \in \mathcal{D}} \mid \varphi(Z;\bm{\eta},\delta, t) - \varphi(Z;\hat{\bm{\eta}}_{-S},\delta, t) | \|= o_\Pb(1)$.
			\item [4)] $
			\left( \underset{\delta\in \mathcal{D}}{sup}\| m_{\delta,s} -  \widehat{m}_{\delta,s} \| + \| \pi_s -  \widehat{\pi}_{s} \| \right) \Big(  \| \widehat\pi_r - {\pi}_r \| + \| \widehat\omega_r - {\omega}_r \| \Big)  = o_\Pb\left(\frac{1}{\sqrt{n}}\right)
			$, $\forall r \leq s \leq t$.
		\end{itemize}
		Then we have
		$$
		\frac{\hat{\psi}_t (\delta) - \psi_t (\delta)}{\hat{\sigma}(t, \delta)/\sqrt[]{n}} \leadsto \mathbb{G}(\delta, t)
		$$
		in $l^{\infty}(\mathcal{D})$, where $\mathbb{G}$ is a mean-zero Gaussian process with covariance\\ $\E[\mathbb{G}(\delta_1, t_1)\mathbb{G}(\delta_2, t_2)]=\E\left[\widetilde{\varphi}(Z;\bm{\eta},\delta_1, t_1) \widetilde{\varphi}(Z;\bm{\eta},\delta_2, t_2)\right]$ and $\widetilde{\varphi}(Z;\bm{\eta},\delta, t) = \frac{\varphi(Z;\bm{\eta},\delta, t) - \psi_t(\delta)}{\sigma(\delta, t)}$.
	\end{thm}
	
	A proof of the above theorem can be found in Appendix \ref{proof:thm-6-1}. We also analyze the second order remainders for the efficient influence function, and keep the intervention distribution completely general (see Lemma \ref{lem:eif}, \ref{lem:remainder_1}, \ref{lem:remainder_2} in the appendix). Therefore, one may apply our results to studies of other stochastic interventions under missingness/dropout as well. 
	
	Assumptions 1), 2) and 3) in Theorem \ref{thm:convergence} are all quite weak. Assumptions 1) and 2) are mild boundedness conditions, where assumption 2) could be further relaxed at the expense of a less simple proof, for example using bounds on $L_p$ norms. Assumption 3) is also a mild consistency assumption, with no requirement on rate of convergence. The main substantive assumption is Assumption 4), which requires that the product of nuisance estimation errors must vanish at fast enough rates. One sufficient condition for this is that all the nuisance functions are consistently estimated at a rate of $n^{1/4}$ or faster. 

	Lowering the bar from $\sqrt[]{n}$ to $n^{1/4}$ indeed allows us to employ a richer set of modern machine learning tools, since such rates are attainable under diverse structural constraints \citep[e.g.,][]{yang2015minimax, raskutti2012minimax, gyorfi2006distribution}. In this paper, however, we are agnostic about how such rates should be attained. In practice, we may want to consider using different estimation techniques for each of $\bm{\pi, m, \omega}$ based on our prior knowledge and descriptive information, or use ensemble learners.
	
	Based on the result in Theorem \ref{thm:convergence}, we can construct pointwise $1-\alpha$ confidence intervals for $\psi_t (\delta)$ as
	$$
	\widehat{\psi}_t (\delta) \pm z_{1-\alpha/2}\frac{\hat{\sigma}^2(\delta, t)}{\sqrt[]{n}}
	$$
	where $\hat{\sigma}^2(\delta, t)$ is the variance estimator defined in Theorem \ref{thm:convergence}. Following \citet{Kennedy17}, one may use the multiplier bootstrap for uniform inference, by replacing the $z_{1-\alpha/2}$ critical value with $c_\alpha$ satisfying
	$$
	\Pb \left( \underset{\delta\in \mathcal{D}, 1\leq s \leq t}{\sup} \left| \frac{\widehat{\psi}_s (\delta)-\psi_s (\delta)}{\widehat{\sigma}(\delta, s)/\sqrt[]{n}} \right| \leq c_\alpha \right) = 1 - \alpha + o(1) .
	$$
	
	We refer to \citet{Kennedy17} for details on how to construct $c_\alpha$ via the multiplier bootstrap.

	\section{Infinite Time Horizon Analysis}
	\label{sec:Inf-time-horizon}
	
	The great majority of causal inference literature considers a finite time horizon where the number of timepoints $T$ is small and fixed, or even just equal to one, a priori ruling out much significant (if any) longitudinal structure. However, in practice more and more studies accumulate data across very many timepoints, due to ever increasing advances in data collection technology. In fact, in many applications $T$ can even be comparable to or larger than sample size $n$. This renders most of the classical methods based on finite time horizons obsolete, as their theoretical results/analysis have not been validated in such time horizon where $T$ can grow to infinity. For example, \citet{kumar2013mobile} describe how new mobile and wearable sensing technologies have revolutionized randomized trials and other health-care studies by providing data at very high sampling rates (10-500 times per second). \citet{klasnja2015microrandomized, qian2020micro} use $210$ timepoints in their study of micro-randomized trials for evaluating just-in-time adaptive interventions via mobile applications. As we collect more granular and fine-grained data, some recent studies explore efficient off-policy estimation techniques in infinite-time horizon settings (e.g., \citet{liu2018breaking} in reinforcement learning). Interestingly, though, there has been no formal analysis for general longitudinal studies.

    Therefore in this section we analyze the behavior of the IPW version of our proposed estimator (relative to the standard IPW estimator in classical deterministic settings), in a more realistic regime where $T$ can scale with sample size. To the best of our knowledge, this is one of the first such infinite-horizon analyses in causal inference, outside of some recent similarly specialized examples involving dynamic treatment regimes \citep{laber2018optimal,ertefaie2018constructing}. Specifically, we study the variance ratio bound, and show how deterministic effects are afflicted by an inflated variance relative to incremental intervention effects as $T$ grows.
	
    We proceed with comparing the variances of estimators of the deterministic effect for the always-treated (receiving treatment at every timepoint) versus the incremental effect for $\delta>1$. For simplicity and concreteness, in what follows we consider a simple setup where the propensity scores are all equal to $p$ (i.e., $\pi_t(H_t)=p$ for all $t$) and there is no dropout (i.e. $d\Pb\{R_{t+1}=1\}=1$ a.e. $[\Pb]$ for all $t=1,...,T$). This makes the pseudo-regression functions $m_s$ in Theorem \ref{thm:eif} equal to zero. In this setup we have unbiased estimators of the always-treated effect $\psi_{at}=\E(Y^{\overline{\bm{1}}_T})$ and the incremental effect $\psi_{inc} = \E(Y^{\overline{Q}_T(\delta)})$ given by
	$$
	\widehat{\psi}_{at} = \prod_{t=1}^{T} \left( \frac{A_t}{ p} \right)Y
	$$
	and
	$$
	\widehat{\psi}_{inc} = \prod_{t=1}^{T} \left( \frac{\delta A_t + 1-A_t }{\delta p + 1-p}  \right)Y 
	$$
	respectively, where $Y=Y_T$. In the next theorem, we analyze the variance ratio of the two estimators and show that one can achieve near-exponential precision gains by targeting $\psi_{inc}$. 
	
	\begin{thm} \label{thm:inf-time-horizon}
		Consider the estimators and conditions defined above. Further assume that $\left\vert  Y \right\vert \leq {b_u}$ for some constant $ b_u>0$ and $\E\left[\left(Y^{\overline{\bm{1}}_T} \right)^2 \right] > 0$. Then for any $T \geq 1$,
			\begin{align*}
				C_{T}\left[ \left\{ \frac{\delta^2p^2 + p(1-p)}{(\delta p + 1 - p)^2} \right\}^{T} - p^{T} \right]  \leq 
				\frac{\var(\widehat{\psi}_{inc})} {\var(\widehat{\psi}_{at})}
				\leq C_{T}\zeta(T;p)\left\{ \frac{\delta^2p^2 + p(1-p)}{(\delta p + 1 - p)^2} \right\}^{T}
			\end{align*}
			where $C_{T} = \frac{b_u^2}{\E\left[\left(Y^{\overline{\bm{1}}_T} \right)^2 \right]}$ and $\zeta(T;p) = \left( 1+ \frac{c \left(\E\left[Y^{\overline{\bm{1}}_T}\right]\right)^2}{\left( 1/p \right)^{T} \E\left[\left(Y^{\overline{\bm{1}}_T} \right)^2 \right]} \right)$ for any fixed value of $c$ such that $\frac{1}{1-p^T{\left(\E\left[Y^{\overline{\bm{1}}_T}\right]\right)^2}\big/{\E\left[\left(Y^2 \right)^{\overline{\bm{1}}_T} \right]}} \leq {c}$.
	\end{thm}
	The proof of the above theorem is given in Appendix \ref{proof:thm-inf-time} and is based on the similar logic used in deriving the g-formula \citep{robins1986}. Note that we only require two very mild assumptions in the above theorem: the boundedness assumption on $Y$, and $\E[(Y^{\overline{\bm{1}}_T} )^2 ] > 0$, which is equivalent to saying $Y^{\overline{\bm{1}}_T}$ is a non-degenerate random variable. In the proof, we give a more general result for any sequence $\overline{a}_{T} \in \overline{\mathcal{A}}_{T}$ as well. 
	
	Theorem \ref{thm:inf-time-horizon} allows us to precisely quantify the relative statistical certainty in estimating the two effects. Specifically, since $\frac{\delta^2p^2 + p(1-p)}{(\delta p + 1 - p)^2} < 1$ for $\delta > 1$ and $\zeta(T;p)$ is bounded (and converging to one monotonically), the variance ratio decays exponentially in $T$. This implies that we may reap extraordinary gains in statistical precision from targeting ${\psi}_{inc}$ instead of ${\psi}_{at}$, when we intend on incorporating substantial number of timepoints in the study. The same goes for effects for the never-treated versus the incremental interventions with $\delta<1$ (see Appendix \ref{proof:thm-inf-time}).

    \begin{remark} 
    The variance ratio we study in this section is somewhat distinct from usual relative efficiency, since here we are considering two different (but closely related) target parameters. However, when we are indifferent about the inferential target, the variance ratio still can serve as a useful guidance in selecting an estimator. As $\delta \rightarrow \infty$, the gap between two effects monotonically shrinks to zero and the two target parameters $\psi_{at}$ and $\psi_{inc}$ become eventually identical, so the variance ratio goes to $1$. On the other hand, for finite $\delta$, the two effects are not quite the same, but how much one versus the other is of more interest is debatable ($\psi_{inc}$ could indeed be more preferable if we aim to describe how outcomes would vary with more practical gradual changes in treatment intensity). If we do not have a strong reason to prefer one effect over the other, we could choose the one with smaller variance in favor of improved statistical precision. This issue also arises for local effects under positivity violations, instrumental variables, etc. \citep[e.g.,][]{imbens2014instrumental, aronow2016local, crump2009dealing}, where an estimand is adaptively chosen on the basis of smaller variance.
    \end{remark}

	In what follows we refine Theorem \ref{thm:inf-time-horizon} so that one can characterize the minimum number of timepoints to guarantee a smaller variance for ${\psi}_{inc}$.
	
	\begin{cor} \label{cor:inf-time-horizon}
		There exists a finite number $T_{min}$ such that
		\begin{align*}
		\var(\widehat{\psi}_{inc}) < \var(\widehat{\psi}_{at})
		\end{align*}
		for every $T > T_{min}$, where $T_{min}$ is never greater than
		\begin{align*}
		\min \left\{T: \left[\frac{\delta^2p+1-p}{(\delta p + 1 - p)^2}\right]^T - \frac{c_{\bm{1}}}{p^T} + 2 < 0\right\}
		\quad 
		\text{where} \quad c_{\bm{1}}=\frac{\E\left[\left(Y^{\overline{\bm{1}}_T} \right)^2 \right]}{b^2_u}.
		\end{align*} 
	\end{cor}

	The proof is given in Appendix \ref{proof:cor-inf-time}. The proof of the above corollary relies upon the fact that $\var(\widehat{\psi}_{inc})$ can be represented as a variance of the weighted sum of the IPW estimators for $Y^{\overline{a}_{T}}$, $\forall \overline{a}_{T} \in \overline{\mathcal{A}}_{T}$ (see Lemma \ref{lem:inf-time-decomp} in the appendix). 

	\begin{remark}
	 It may be possible to further tighten the upper bound for $T_{min}$, but considering the focus of our paper this would not be very illuminating and practically meaningful, since the value of $T_{min}$ in the above corollary is already quite small in general. To illustrate, consider $Y \in [0,1]$, $\delta = 2.5, p = 0.5$, and an extreme case of $c_{\bm{1}}=0.05$ (i.e., $Y^{\overline{\bm{1}}_T}$ is mostly concentrated around $0$). Then $T_{min}=6$. If we use $\delta = 5, p = 0.5$, then $T_{min}=9$. 
	\end{remark}
		
	Theorem \ref{thm:inf-time-horizon} and Corollary \ref{cor:inf-time-horizon} can be generalized to the case of observational studies where the nuisance functions need to be estimated, but our view is that the simple case captures the main ideas and the general case would only add complexity. 
	
	To empirically assess the validity of Theorem \ref{thm:inf-time-horizon}, we conduct two simple simulation studies as below. 

    \textit{Simulation 1 (Randomized Trial).}
    We set $p=0.5$ and let $Y \ \mid \ \overline{A}_{t}  \sim N\left( 10 + {A}_{t}, 1 \right)$ truncated at $\pm$ two standard deviations. Based on this data generation process, given a value of $\delta$, we generate 100 different datasets for $t=1,...,50$, $n=500$, where we make sure the positivity assumption is valid in our simulation \footnote{This is done in a similar spirit to Laplace smoothing in Naive Bayes.}. Then we compute the sample variance of each estimator and their ratio correspondingly, and present them in Figure \ref{fig:sim1-inf-time}.

	\begin{figure}[!t] 
		\begin{minipage}{0.45\textwidth}
			\centering
			\includegraphics[width=.9\linewidth]{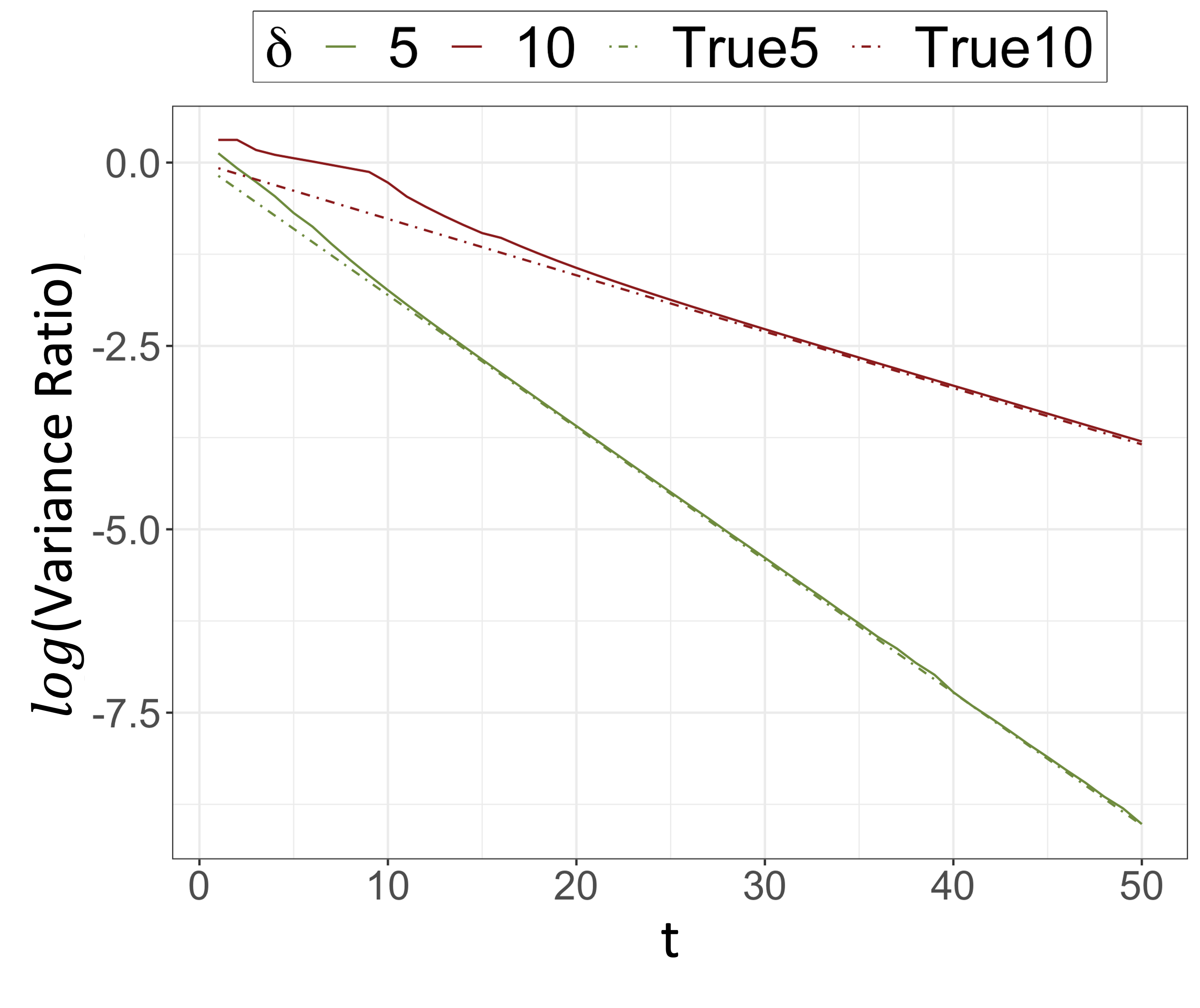}
		\end{minipage}\hfill
		\begin{minipage}{0.45\textwidth}
			\centering
			\includegraphics[width=.9\linewidth]{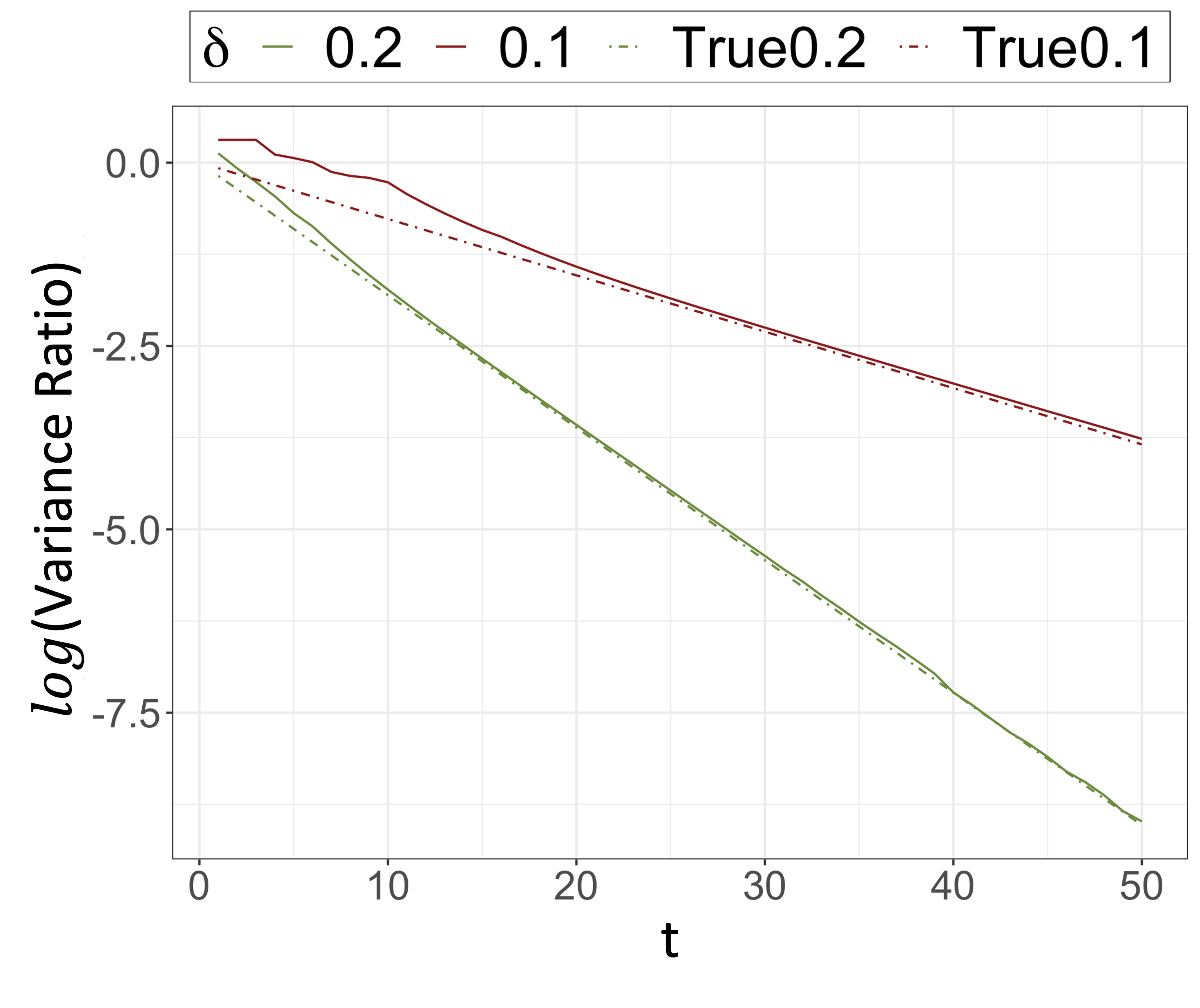}
		\end{minipage}
		\caption{Variance ratio in log-scale over $t$ for the always-treated with $\delta=5,10$ (Left), and for the never-treated with $\delta=0.2, 0.1$ (Right). The true lower bounds for each $\delta$ are represented by dotted line.}
		\label{fig:sim1-inf-time}
	\end{figure}
 
    \textit{Simulation 2 (Observational Study).}
    Although not directly addressed in Theorem \ref{thm:inf-time-horizon}, here we also consider the setting for  observational studies. Specifically, we consider a model
    $$
    X_t=(X_{1,t}, X_{2,t}) \sim N(0,\textbf{I}) 
    $$
    $$
    \pi_t(H_t) = expit\Big( 2\sum_{s=t-2}^{t-1} \left(A_s-1/2\right) \Big)
    $$
    $$
    \left(Y \big\vert \overline{X}_{t}, \overline{A}_{t} \right) \sim N\big(\mu(\overline{X}_{t}, \overline{A}_{t}),1 \big)
    $$
    for all $t \leq T$, where we let $\mu(\overline{X}_{t}, \overline{A}_{t})= 10 + A_{t} + A_{t-1} +\vert((\bm{1}^\top X_{t}+\bm{1}^\top X_{t-1}) \mid$, $\bm{1}=[1,1]^\top$ and let $expit$ denote the inverse logit function. In this simulation, we assume that it is more (less) likely to receive a treatment if a subject has (not) received treatments recently. The rest of the specification remains the same as \textit{Simulation 1}. The result is presented in Figure \ref{fig:sim2-inf-time}.

	\begin{figure}[!t] 
		\begin{minipage}{0.45\textwidth}
			\centering
			\includegraphics[width=.9\linewidth]{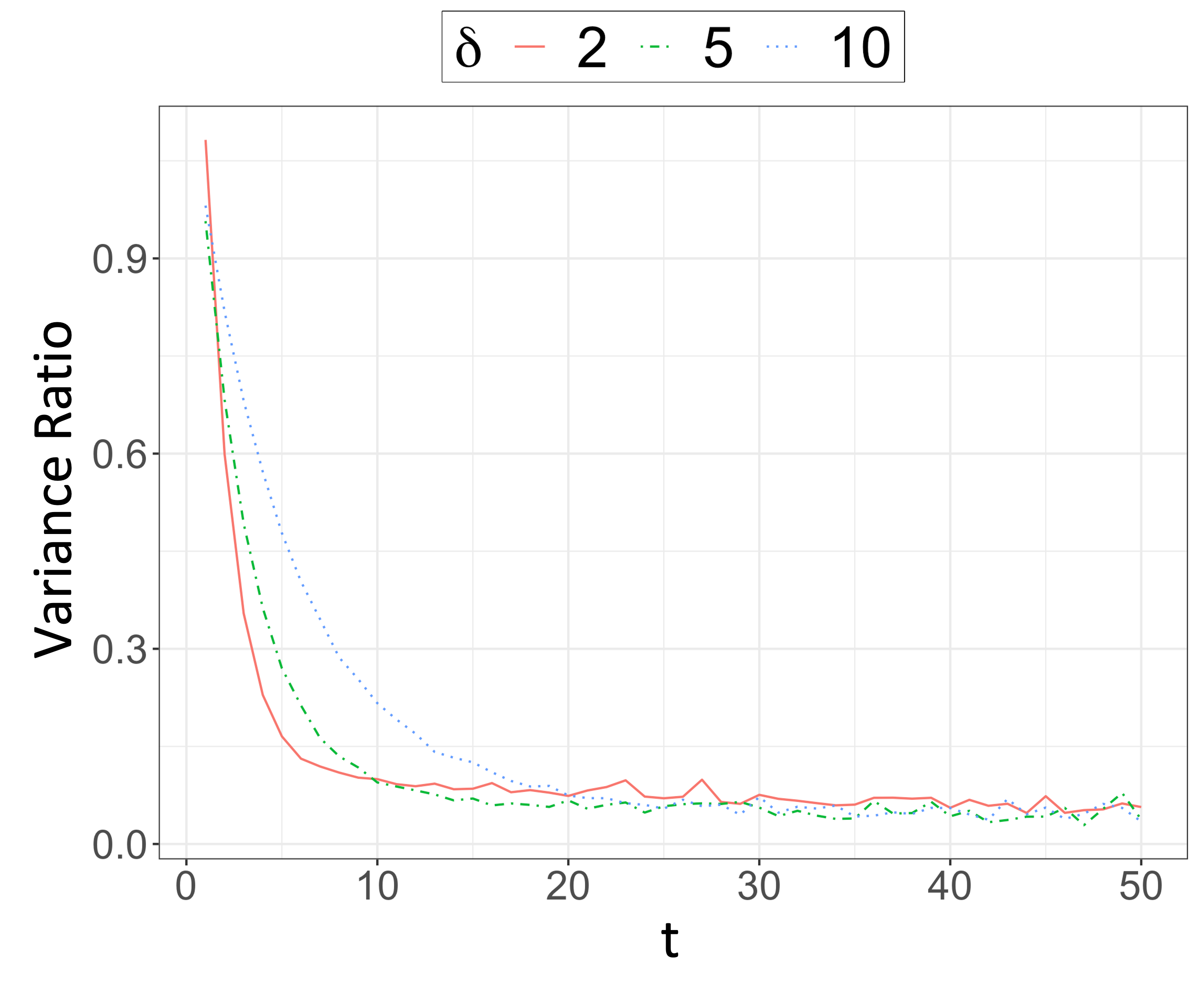}
		\end{minipage}\hfill
		\begin{minipage}{0.45\textwidth}
			\centering
			\includegraphics[width=.9\linewidth]{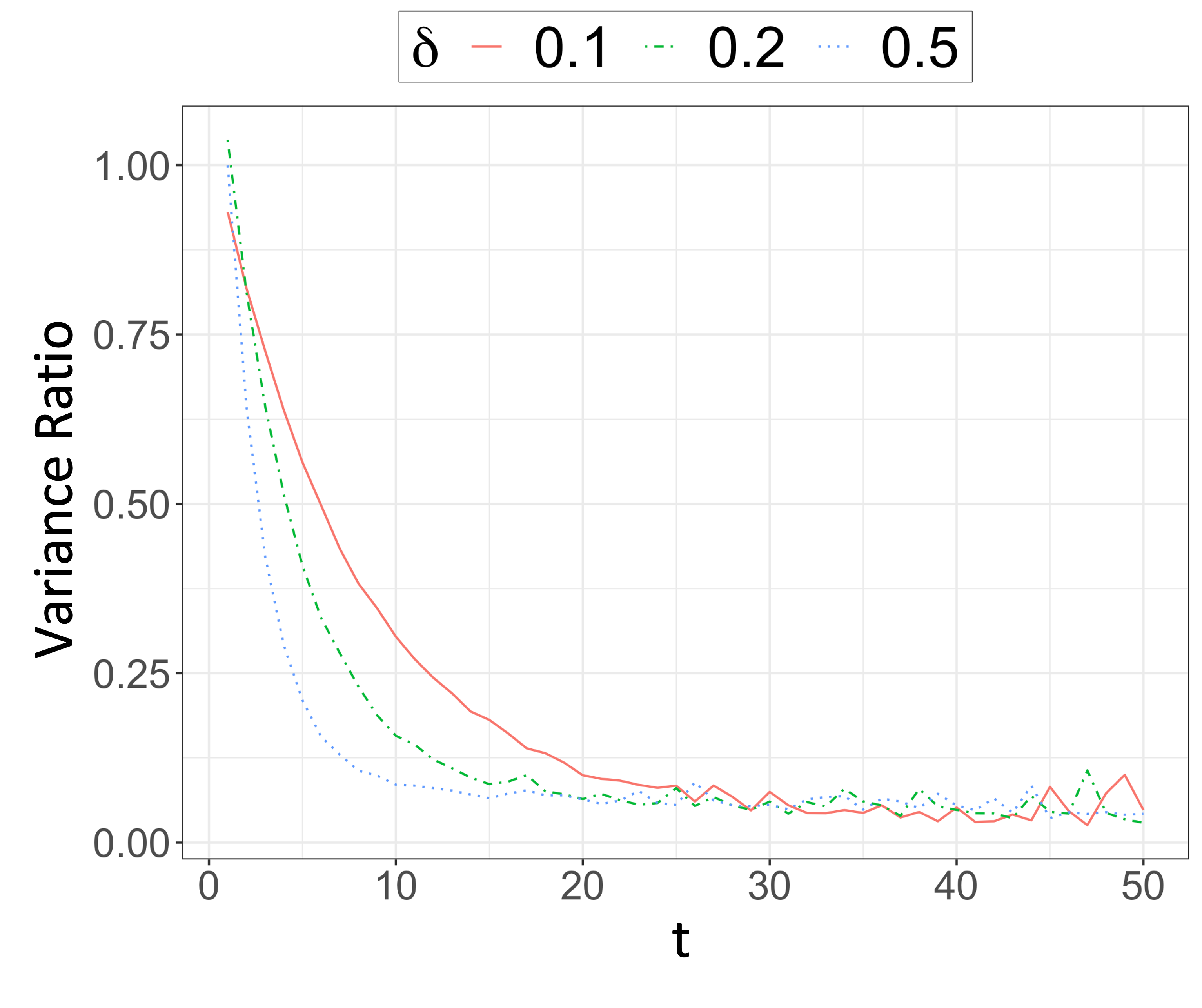}
		\end{minipage}
		\caption{Variance ratio over $t$ for the always-treated with $\delta=2,5,10$ (Left), and for the never-treated with $\delta=0.5, 0.2, 0.1$ (Right).}
		\label{fig:sim2-inf-time}
	\end{figure}
	
    The simulation results support our theoretical results. 
	Overall, the result in this section provides crucial insight into the longitudinal study with many timepoints, suggesting that massive gains in statistical certainty are possible by studying incremental rather than classical deterministic effects.

	\section{Experiments}
	
	\subsection{Simulation Study} \label{sec:sim-study}
	
	In this section we explore finite-sample performance of the proposed estimator $\hat{\psi}_t(\delta)$ via synthetic simulation. We consider the following data generation model 
	{
	$$
	X_s=[X_{1,s}, X_{2,s}, X_{3,s}]^\top,
	$$
	\[
	    \text{where } [X_{1,s}, X_{2,s}]^\top \sim N([X_{1,s-1}, X_{2,s-1} A_{s-1}]^\top,\textbf{I}), \, X_{1,0}=X_{2,0}=0,
	\]
	\[
	    \text{and }  X_{3,s}= 
    	\begin{cases}
        2.5\{2\times {Bernoulli}(1/2)-1\}, \ \text{if } \ s=1 \\
        X_{3,s-1} + N(0, 0.5^2), \ \text{elsewhere} , 
        \end{cases}
	\]
	$$
	A_s \sim {Bernoulli}\left(1/2 + \xi(X_{s-1},A_{s-1})/2 \right), 
	$$
	$$
	\text{where } \xi(X_{s-1},A_{s-1})= 
	\begin{cases}
    (2A_{s-1}-1) \vert X_{1,s-1} + X_{2,s-1} \vert/16, \ \text{if } \ \vert X_{1,s-1} + X_{2,s-1} \vert < 1 \\
    (2A_{s-1}-1)/16, \ \text{elsewhere } , 
    \end{cases}
	$$
	$$
	\Pb(R_{s+1}=1 \mid H_{s},A_{s}, R_{s}=1) = expit\left(r_d+5(2\times\mathbbm{1}\{X_{3,s} > 0\}-1) + (A_s + A_{s-1})/4 + (X_{1,s}+X_{2,s})/16\right), 
	$$
	$$
	Y_s \big\vert {H}_{s}, {A}_{s}  \sim N\big((X_{1,s} + X_{1,s-1})/4 + A_sX_{3,s} + 1, 1 \big),
	$$
	for $s=1 ,...,t$. $r_d >0$ is a constant used to control an average amount of dropout units. In this setup we assume that the more likely subjects have been treated, the more likely they will receive the treatment in the next timepoint in general. Moreover, the dropout probability at each $s \leq t$ is largely driven by the sign of $X_{3,s}$: the dropout probability will be low (high) if $X_{3,s}>0$ $(<0)$. Therefore, although each $X_{3,s}$ is designed to have a symmetric, bimodal distribution with a mean of $0$, the value of $X_{3,s}$ for surviving subjects will tend to be much greater than $0$. Due to the way $X_{3,s}$ also interact with the outcome in the above model, discarding all the subjects that have dropped out should lead to an upward-biased estimate of the incremental intervention effects. In Appendix \ref{appendix:aux-figures-simulation}, we provide auxiliary figures for the sake of better understanding of our simulation. Variables akin to $X_3$ that considerably affect both outcome and dropout are commonly found in practice (e.g., side effects).
    }
    
	We estimate the incremental effect at $t=4$. We compare our proposed estimator (${\hat{\psi}_{ours}}$) with three baseline methods: the naive Z-estimator ($\hat{\psi}_{inc.pi}$) and the IPW estimator ($\hat{\psi}_{inc.ipw}$), both of which are defined in Section \ref{subsec:proposed-estimator}, and the original incremental-effect estimator ($\hat{\psi}_{inc.nc}$) proposed by \citet{Kennedy17}. Note that for using $\hat{\psi}_{inc.nc}$ we have to discard samples that have ever dropped out, whereas in other estimators surviving subjects are properly re-weighted for dropout adjustment at each timepoint. Since finite-sample properties of $\hat{\psi}_{inc.nc}$ were already extensively explored in \citet{Kennedy17}, here we primarily focus on the effect of dropout in the longitudinal setting.
	
	To estimate nuisance parameters, following \citet{Kennedy17} we form an ensemble of some widely-used nonparametric models. Specifically, we use the cross-validation-based superleaner ensemble algorithm \citep{van2007super} via the \texttt{SuperLearner} package in R to combine support vector machine, random forest, k-nearest neighbor regression. For ${\hat{\psi}_{ours}}$ and $\hat{\psi}_{inc.nc}$, we use $K = 2$-fold sample splitting. 
	
	We repeat simulation $S=250$ times in which we draw $n$ samples each simulation. We use $D=30$ values of $\delta$ equally spaced on the log-scale within $[0.2, 3]$. As in \citet{Kennedy17}, performance of each estimator is assessed via integrated bias and root-mean-squared error (RMSE) defined by
    \[
    \widehat{bias} = \frac{1}{D} \sum_{d=1}^{D} \left\vert \frac{1}{S} \sum_{s=1}^{S}  \hat{\psi}_s(t;\delta_d) - {\psi}(t;\delta_d) \right\vert , \quad  
    \widehat{RMSE} = \frac{\sqrt{n}}{D} \sum_{d=1}^{D} \left[ \frac{1}{S} \sum_{s=1}^{S} \left\{ \hat{\psi}_s(t;\delta_d) - {\psi}(t;\delta_d) \right\}^2 \right]^{1/2}  
    \]
    where $\hat{\psi}_s(t;\delta_d)$ and ${\psi}(t;\delta_d)$ are the estimate and true value of the target parameter respectively, for $s$-th simulation and $\delta_d$. We present the results in Table \ref{tbl:synthetic-sim1}.

\begin{table}
	\begin{center}
		\begin{tabular}{l l l l l l l l l c}
			\toprule
			\multirow{2}{*}{\, n} &\multicolumn{4}{l}{\qquad \qquad \qquad {$\bm{\widehat{bias}} (\times 10^{-3})$}} &\multicolumn{4}{l}{\qquad \qquad \qquad {$\bm{\widehat{RMSE}}$}} &\multirow{2}{*}{ \thead{Average \\ Dropouts (\%)} } \\ 
			&$\hat{\psi}_{inc.pi}$ & $\hat{\psi}_{inc.ipw}$ & $\hat{\psi}_{inc.nc}$  & ${\hat{\psi}_{ours}}$ 
			&$\hat{\psi}_{inc.pi}$ & $\hat{\psi}_{inc.ipw}$ & $\hat{\psi}_{inc.nc}$  & ${\hat{\psi}_{ours}}$ &\\
			\midrule
			  & 14.5 & 24.1 &  30.1 & \textbf{9.8} &  1.59 & 2.78 &  2.96 & \textbf{1.37} & 50.5 \\
			1000 & 12.5 & 14.7 &  19.8 & \textbf{8.3} &  1.31 & 1.84 &  2.01 & \textbf{1.14} & 28.0 \\
			  &  10.7 & 11.3 &  9.5 & \textbf{7.2} &  1.17 & 1.35 &  1.13 & \textbf{0.99} & 8.9 \\
			  & 10.3 & 12.0 &  23.8 & \textbf{7.1} &  1.15 & 1.34 &  1.29 & \textbf{1.05} & 49.6 \\
			2500 & 10.2 & 10.9 &  14.1 & \textbf{6.2} &  1.06 & 1.19 &  1.06 & \textbf{0.95} & 27.5 \\
			  & 7.8 & 7.5 &  5.3 & \textbf{4.5} &  0.94 & 1.03 &  {0.93} & \textbf{0.91} & 9.1 \\
			\bottomrule
	\end{tabular}\end{center}
	\caption{{Integrated bias and RMSE across different baselines and simulation settings.}}
	\label{tbl:synthetic-sim1}
\end{table}

As shown in Table \ref{tbl:synthetic-sim1}, when there is a substantial amount of subject dropout, $\hat{\psi}_{inc.nc}$ shows much worse performance than all the other dropout-adjusted estimators, which is expected by the design of our data generation model. However, this gap shrinks as dropout rates decrease. Also in each setting, the proposed estimator ${\hat{\psi}_{ours}}$ appears to perform better and more markedly improve with sample size than $\hat{\psi}_{inc.pi}$ and $\hat{\psi}_{inc.ipw}$. This behavior is indicative of the validity of our theory that ${\hat{\psi}_{ours}}$ is not only able to adjust for the dropout process, but more efficient than $\hat{\psi}_{inc.pi}$ and $\hat{\psi}_{inc.ipw}$. 


	\subsection{Application} \label{sec:application}
	
	Here we illustrate the proposed methods in analyzing the Effects of Aspirin on Gestation and Reproduction (EAGeR) data, which evaluates the effect of daily low-dose aspirin on pregnancy outcomes and complications. The EAGeR trial was the first randomized trial to evaluate the effect of pre-conception low-dose aspirin on pregnancy outcomes \citep{schisterman2014preconception, mumford2016expanded}. However, to date this evidence has been limited to intention-to-treat analyses.

    The design and protocol used for the EAGeR study have been previously documented \citep{schisterman2013randomised}. Overall, 1,228 women were recruited into the study (615 aspirin, 613 placebo) and 11\% of participants chose to drop out of the study before completion. Roughly 43,000 person-weeks of information were available from daily diaries, as well as study questionnaires, and clinical and telephone evaluations collected at regular intervals over follow-up. The dataset is characterized by a substantial degree of non-compliance (more than 50\% at the end of the study), and thereby is susceptible to positivity violation.
	
	We used our incremental propensity score approach to evaluate the effect of aspirin on pregnancy outcomes in the EAGeR trial, accounting for time-varying exposure and dropout. {We let each variable become a constant equal to its final value after the time point at which no more actual data is collected on the subject, so we have balanced panel data as described in (\ref{setup:causal-process}).} Here, the study terminates at week 89 ($T=89$). We use 24 baseline covariates (e.g., age, race, income, education, etc.) and 5 time-dependent covariates (compliance, conception, vaginal bleeding, nausea and GI discomfort). $A_t$ is a binary treatment variable coded as $1$ if a woman took aspirin at time $t$ and $0$ otherwise. $R_t=1$ indicates that the woman is observed in the study at time $t$. Lastly, $Y_{t}$ is an indicator of having a pregnancy outcome of interest at time $t$. We are particularly interested in two types of pregnancy outcomes: \emph{live birth} and \emph{pregnancy loss} (fetal loss). We perform separate analysis for each of the two cases.
	
	For comparative purposes, we estimate the simple complete-case effect 
	\begin{align*}
		\widehat{\psi}_{CC} = \Pn(Y_T |\overline{A}_T=1, R_T = 1) - \Pn(Y_T | \overline{A}_T=0, R_T = 1).
	\end{align*}
	which relies on both non-compliance and drop-out being completely randomized. The value of $\widehat{\psi}_{CC}$ is 0.052 (5.2\%) for live birth and 0.012 (1.2\%) for pregnancy loss, both of which are close to the intention-to-treat estimates reported in \citet{schisterman2013randomised, schisterman2014preconception}.

	We give a brief discussion on why standard modeling approaches fail here. We found strong evidence of positivity violations in the EAGER dataset; as shown in Figure \ref{fig:ps-always-treated} in Appendix  \ref{appendix:application-other-approaches}, the average propensity score quickly drops to zero as $t$ grows. This suggests that very few patients follow the given protocol of taking aspirin late in the study. Thus, it is unrealistic to use an intervention where \emph{all} participants would take aspirin at every time, as required in many standard models including the popular marginal structural models (MSMs) \citep{robins2000marginal}. In fact, when we modeled the effect curve by $\E[Y^{\overline{a}_T}] = m(\overline{a}_T;\beta) = \beta_0 + \sum_{t=1}^T \beta_{1t}a_t$ so that the coefficient for exposure can vary with time, then the standard inverse-weighted MSM fits failed and no coefficient estimates were found even for moderate values of $T = \sim 10$ (see Figure \ref{fig:ps-always-treated}-(b) in Appendix \ref{appendix:application-other-approaches} for a closer look). This positivity violation precludes other standard approaches for time-varying treatments as well. 

	One quick remedy could be to move away from standard ATEs and instead only estimate the mean outcome if no one were treated, comparing to the observed outcome (\emph{not} {all} versus {none} as in the ATE). Then we can apply some other {nonparametric} approaches available in the literature for estimating this one-sided counterfactual. When we use the g-computation (plug-in) estimator \citep{robins1986}, the result seems to suggest that the mean outcome if no one received aspirin is worse than the observed (Figure \ref{fig:aspirin-other-g-computation} in Appendix \ref{appendix:application-other-approaches}). However, when we use the sequential doubly robust (SDR) estimator \citep{luedtke2017sequential}, the huge overlap between 95\% CI intervals prevents us from drawing any firm conclusion (Figure \ref{fig:aspirin-other-SDR} in Appendix \ref{appendix:application-other-approaches}).

	Now, we estimate the incremental effect curve $\psi_T(\delta)$, which represents the probability of having live birth or pregnancy loss at the end of the study ($t=T$) if the odds of taking aspirin for all women were increased by a factor of $\delta$ at all timepoints, across different values of $\delta$. Again, we use the cross-validated superleaner algorithm \citep{van2007super} to combine support vector machine, random forest, and k-nearest neighbor regression to estimate a tuple of nuisance functions $(m_t, \omega_{t}, \pi_{t})$ at each $t \leq T$. We proceed with sample splitting with $K=2$ splits, and use 10,000 bootstrap replications to compute pointwise and uniform confidence intervals. Results are shown in Figure \ref{fig:appl-aspirin}.
	\begin{figure}[!htb] 
		\begin{minipage}{0.49\textwidth}
			\centering
			\includegraphics[width=.95\linewidth]{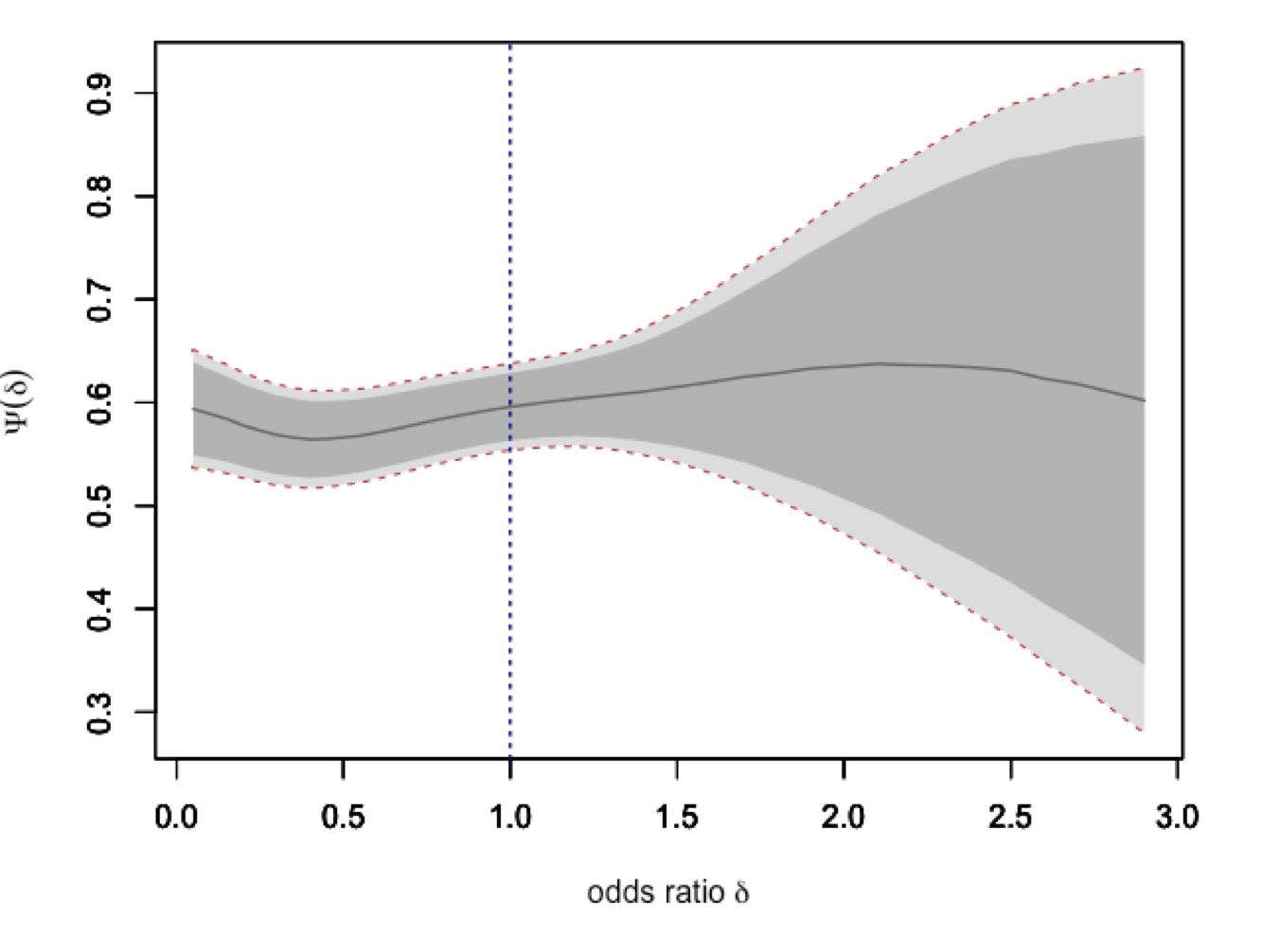}
		\end{minipage}\hfill
		\begin{minipage}{0.49\textwidth}
			\centering
			\includegraphics[width=.95\linewidth]{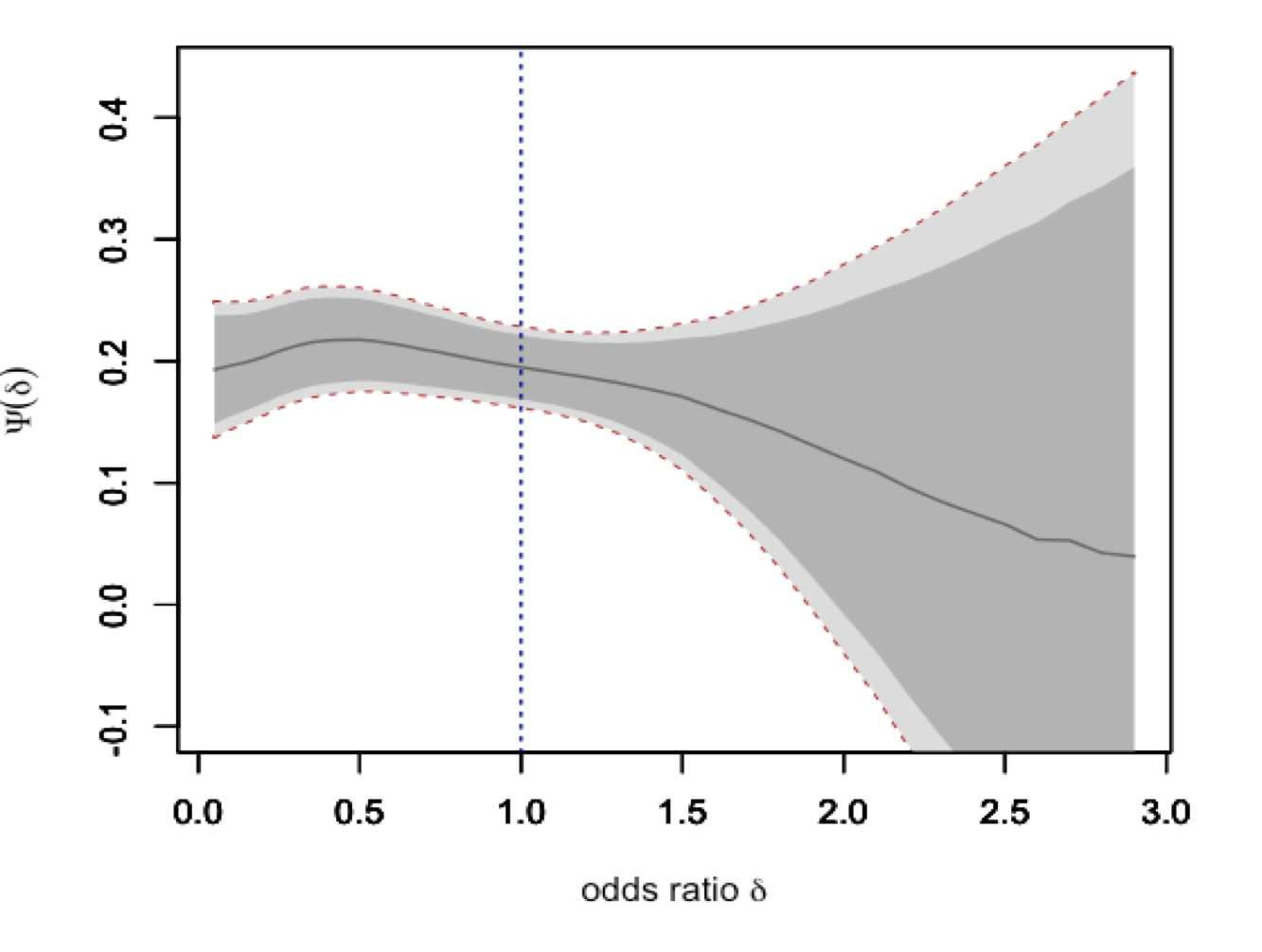}
		\end{minipage}
		\caption{Estimated incremental effect curves which represent the probability of having a live birth (Left) and a pregnancy loss (Right). Lighter grey area with red dotted line represents 95\% uniform bands and darker grey area represents 95\% pointwise bands.}
		\label{fig:appl-aspirin}
	\end{figure}
	
	{
	The estimated curve in Figure \ref{fig:appl-aspirin} appears to be almost flat for live birth, and have a slightly negative gradient with respect to $\delta$ (odds ratio) for pregnancy loss. However, at level $\alpha=.05$ we fail to reject the null of no incremental intervention effects for both cases (as both confidence bands contain a horizontal line). This is mainly due to the noncompliance of aspirin takers that makes the bands too wide for large $\delta$ regimes. Thus, our analysis yielded a similar result to the previous findings of \citet{schisterman2014preconception}, indicating that use of low dose aspirin was not significantly associated with live birth or pregnancy loss. Nonetheless, the estimated incremental intervention effects provide more detailed information with greater nuance, requiring none of the parametric and positivity assumptions.
	}
	
	{
	\begin{remark}
	 In this analysis we have looked into the intervention that does tell us about the effect of overall increase or decrease in treatment at each $t$, but not the optimal timing of treatment (e.g., when aspirin should be prescribed since conception). As pointed out by \citet{Kennedy17}, one could address such timing issues by considering $\delta$ depending on time and covariate history, which will bring added complexity. We leave this to our future work.
	\end{remark}
	}

	\section{Discussion}

	Incremental interventions are a novel class of stochastic dynamic intervention where positivity assumptions can be completely avoided. However, they had not been extended to repeated outcomes, and without further assumptions do not give identifiability under dropout,  both of which are very common in practice. In this paper we solved this problem by showing how incremental intervention effects are identified and can be estimated when dropout occurs (conditionally) at random. Even in the case of many dropouts, our proposed method efficiently uses all the data without sacrificing robustness. We gave an identifying expression for incremental intervention effects under monotone dropout, without requiring any positivity assumptions. We established general efficiency theory and constructed the efficient influence function, and presented nonparametric estimators which converge at fast rates and yield uniform inferential guarantees, even when all the nuisance functions are estimated with flexible machine learning tools at slower rates. Furthermore, we analyzed the variance ratio of incremental intervention effects to conventional deterministic dynamic intervention effects in a novel infinite time horizon setting in which the number of timepoints can possibly grow with sample size, and showed that incremental intervention effects can yield near-exponential gains in statistical precision. Finally, we showed that the proposed methods can effectively mitigate the bias caused by subject dropout via the simulation study, and applied the methods in study of the effect of low-dose aspirin on pregnancy outcomes.

    There are a number of avenues for future work. The first is application to other  substantive problems in medicine and the social sciences. For example, in a forthcoming paper we analyze the effect of aspirin on pregnancy outcomes with more extensive data. It will also be important to consider other types of non-monotone missingness where the standard time-varying MAR assumption \ref{assumption:A2-M} may not be appropriate \citep{sun2014inverse, tchetgen2016discrete}. We expect that our approach can be extended to other important problems in causal inference; for example, one could develop incremental intervention effects for continuous treatments and instruments \citep{kennedy2017non,kennedy2019robust}, or for mediation in the same spirit as \citep{diaz2019causal}, but generalized to the longitudinal case with dropout. Developing incremental-based sensitivity analyses for the longitudinal MAR assumption would also be an interesting extension. 
	
	\subsection*{Acknowledgement}
	Edward Kennedy and Ashley Naimi gratefully acknowledge financial support from the NSF (Grant \# DMS1810979) and NIH (Grant \# R01HD093602) for this research, respectively. This work was also supported by the Intramural Research Program of the Eunice Kennedy Shriver National Institutes of Child Health and Human Development, National Institutes of Health, Bethesda, Maryland, contract numbers HHSN267200603423, HHSN267200603424, and HHSN267200603426. We are also grateful for useful comments from two anonymous referees. This work was completed while Kwangho Kim was a PhD student at Carnegie Mellon University.
	

\bibliographystyle{unsrtnat}
\bibliography{reference}

\newpage

\appendix
\input{appendix}

\end{document}

%% file: appendix.tex


    \begin{center}
    \LARGE \textsc{Appendix}
    \end{center}

	\section{Algorithm} \label{sec:algorithm}
	
	An algorithm detailing how to compute the proposed estimator \eqref{eqn::proposed-estimator} at $t \leq T$ is given in Algorithm \ref{algorithm-1} as below.
	
	\begin{algorithm}
	\caption{Implementation of the proposed estimator (\ref{eqn::proposed-estimator})} 
	\label{algorithm-1}
	
	Let $\delta$ be fixed. For each $k \in \{1,...,K\}$, let $D_0 =\{Z_i : S_i \neq k\}$ and $D_1 =\{Z_i : S_i = k\}$ denote the training and test sets, respectively, and let $D = D_0 \bigcup D_1$.
	
	\begin{enumerate}
		\item For each time $s=1,...,t$ regress $A_s$ on $H_s$ using observable samples at time $s$ (i.e., only if $R_{s}=1$) in $D_0$, then obtain predicted values $\widehat{\pi}_s(H_s)$ only for units with $R_s=1$ in $D$.
		
		\item For each time $s=1,...,t$ regress $R_{s+1}$ on $(H_s,A_s)$ using observable samples at time $s$ in $D_0$, then obtain predicted values $\widehat{\omega}_s(H_s,A_s)$ only for units with $R_s=1$ in $D$.
		
		\item For each time $s=1,...,t$, letting $W_k = \frac{\delta A_k + 1-A_k}{\delta\hat{\pi}_k(H_k) + 1-\hat{\pi}_k(H_k)}\cdot\frac{1}{\hat{\omega}_k( H_k, A_k)}$ and construct following cumulative product weights 
		\begin{itemize}
			\item[$\cdot$] $\widetilde{W}_{s} = \prod_{k=1}^s W_k$ for $1 \leq s < t$
		\end{itemize} 
		for units with $R_{s}=1$ in $D$.
		
		\item Let $M_{t+1} = Y_{t}$. Then for $s=t,t-1,...,1$,
		\begin{itemize}
			\item[a.] Regress $M_{s+1}$ on $(H_s, A_s)$ using observable samples at time $s+1$ (i.e., only if $R_{s+1}=1$) in $D_0$, then obtain predictions $\widehat{m}_s(H_s, 1)$ and $\widehat{m}_s(H_s, 0)$ for units with $R_s=1$ in $D$.
			\item[b.] Compute \begin{align*}
			\scriptsize
			M_{s} &= \left( \frac{1}{\delta A_s + 1-A_s}  \right) \Bigg[ \frac{\left\{m_{s}(H_{s}, 1) -m_{s}(H_{s}, 0)\right\}\delta(A_s- \widehat{\pi}_{s}(H_{s}))\widehat{\omega}_s( H_s, A_s)}{\delta\widehat{\pi}_{s}(H_{s}) + 1 - \widehat{\pi}_{s}(H_{s})} \\ 
				& \quad + \begin{pmatrix}
                \delta m_{s}(H_{s}, 1)\left\{\widehat{\pi}_{s}(H_{s})\widehat{\omega}_s( H_s, A_s)-A_sR_{s+1}\right\}  \\
                + m_{s}(H_{s}, 0)\left\{(1-\widehat{\pi}_{s}(H_{s}))\widehat{\omega}_s( H_s, A_s) - (1-A_s)R_{s+1} \right\}
                \end{pmatrix} \Bigg]
            \normalsize
            \end{align*}
            in $D$.
		\end{itemize}    
		
		
		\item Compute $\sum_{s=1}^t M_s\widetilde{W}_{s} + \widetilde{W}_{t}Y_{t}R_{t+1}$ for units in $D_1$ and define $\widehat{\psi}^{(k)}_t(\delta)$ to be its average.
	\end{enumerate}
	
	\textbf{Output} : $\widehat{\psi}_t(\delta) = \frac{1}{K}\sum_{k=1}^{K} \widehat{\psi}^{(k)}_t(\delta)$
	\end{algorithm}

	\section{Auxiliary figures for the simulation study}
	\label{appendix:aux-figures-simulation}
	We provide some auxiliary figures to help readers better understand the simulation setup and result presented in Section \ref{sec:sim-study} using a random example. Figures \ref{fig:dist-x3} and \ref{fig:dist-y} illustrate how the dropout process may induce a large upward bias in estimation of incremental effects as shown in Figure \ref{fig:example-bias}. Figure \ref{fig:example-bias} also shows our methods can successfully adjust for dropout. All the results in this example are measured at $t=4$ with the dropout rate of 52.5\%.
	
	\begin{figure}[t!]
	\centering
    \includegraphics[width=0.8\textwidth]{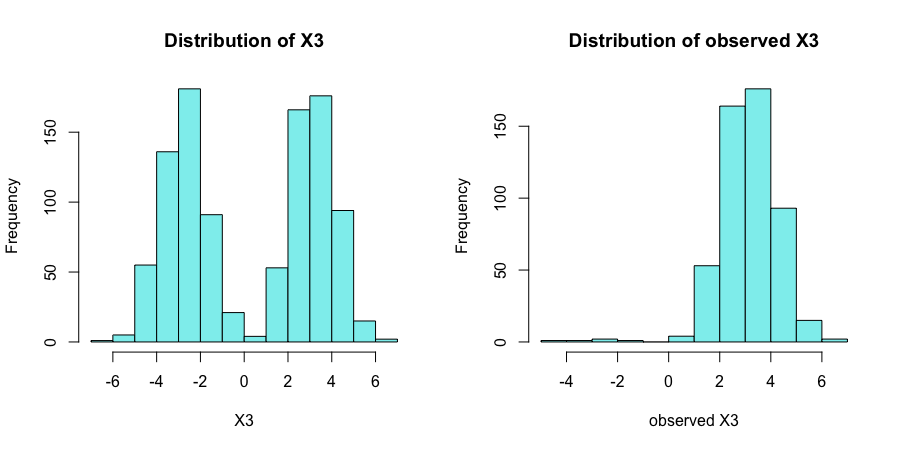}
    \caption{$X_3$ is symmetric about zero. However, sampling distribution of the observed $X_{3}$ is skewed to the left as samples with negative $X_3$ values have dropped out.} \label{fig:dist-x3}
    \end{figure}
	
	\begin{figure}[t!]
	\centering
    \includegraphics[width=0.8\textwidth]{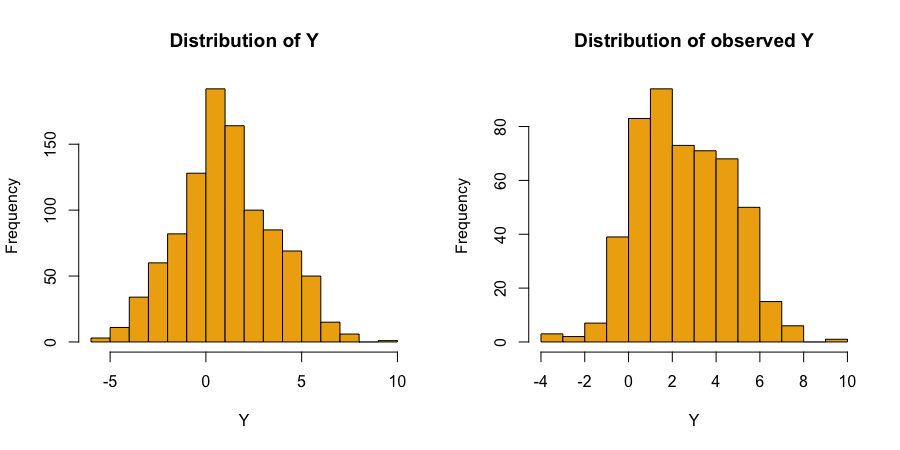}
    \caption{The observed distribution of $X_3$ causes the distribution of observed $Y$ to shift to the right, which results in an upward bias.} \label{fig:dist-y}
    \end{figure}
	
	\begin{figure}[t!]
	\centering
    \includegraphics[width=0.8\textwidth]{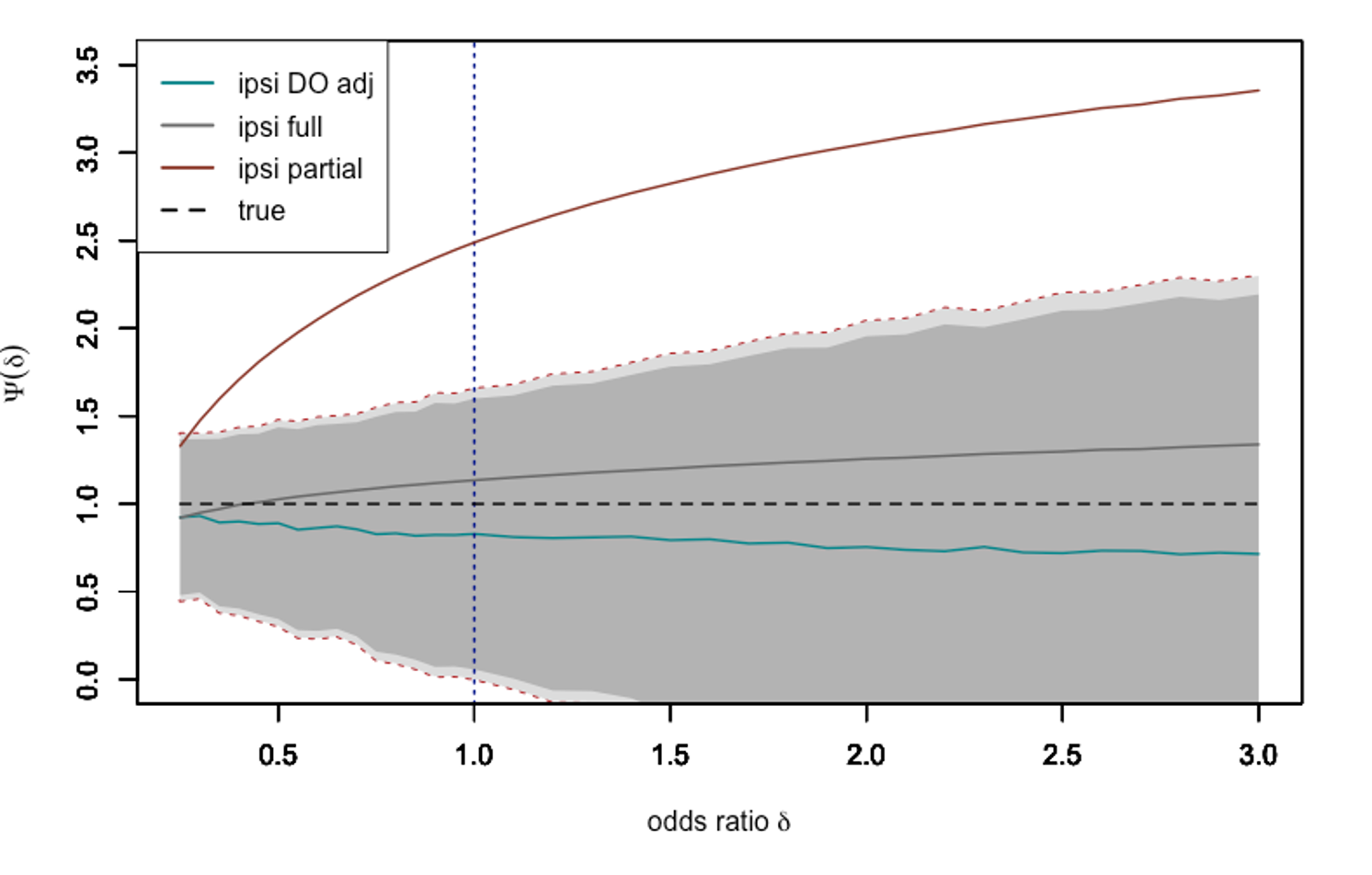}
    \caption{Estimates of the incremental effect using observed data (ipsi partial) are largely deviated upward from the true values, and the ones using full data (ipsi full), whereas our proposed method (ipsi DO adj) successfully adjusts for dropout.} \label{fig:example-bias}
    \end{figure}
	
	\section{Alternative approaches for the EAGeR data analysis} 
	\label{appendix:application-other-approaches}
	
	Here, we discuss why standard approaches fail for our analysis of the EAGER dataset in Section \ref{sec:application} of the main text. For comparative purposes, we alter our target effect and then apply some other nonparametric approaches available in the literature. Then we compare the result with the one we obtained in Section \ref{sec:application}. 
	
	\subsection{Why standard model fails: positivity violation}
	
	All the standard models dealing with time-varying treatments, except on very rare occasions, require treatment positivity. However, as will be elaborated below, positivity is likely violated in the EAGER dataset. Many individuals turned out not to follow the given protocol of taking aspirin and this non-compliance only exacerbates over time. To illustrate this, we present the average propensity score over time in Figure \ref{fig:ps-always-treated}-(a). As shown in Figure \ref{fig:ps-always-treated}-(a), the average propensity score quickly drops to zero as $t$ grows. In other words, Figure \ref{fig:ps-always-treated}-(a) implies that it would be hard to imagine having all of the study participants take aspirin at each time.
	\begin{figure}[!htb] 
		\begin{minipage}{0.44\textwidth}
			\centering
			\includegraphics[width=.9\linewidth]{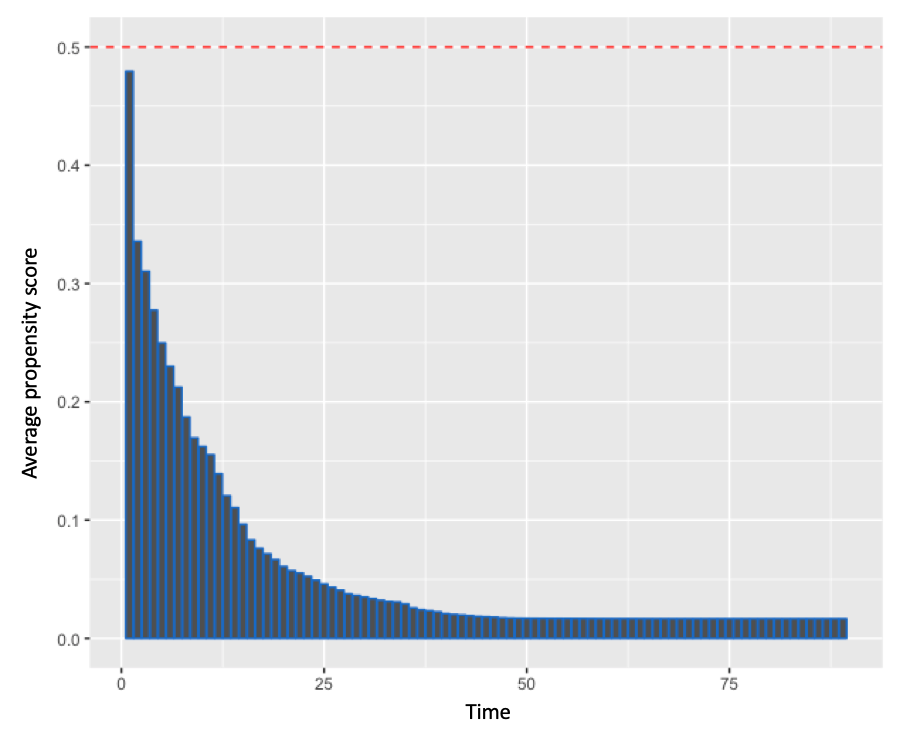}
			\caption*{(a)}
		\end{minipage}\hfill
		\begin{minipage}{0.46\textwidth}
			\centering
			\includegraphics[width=.9\linewidth]{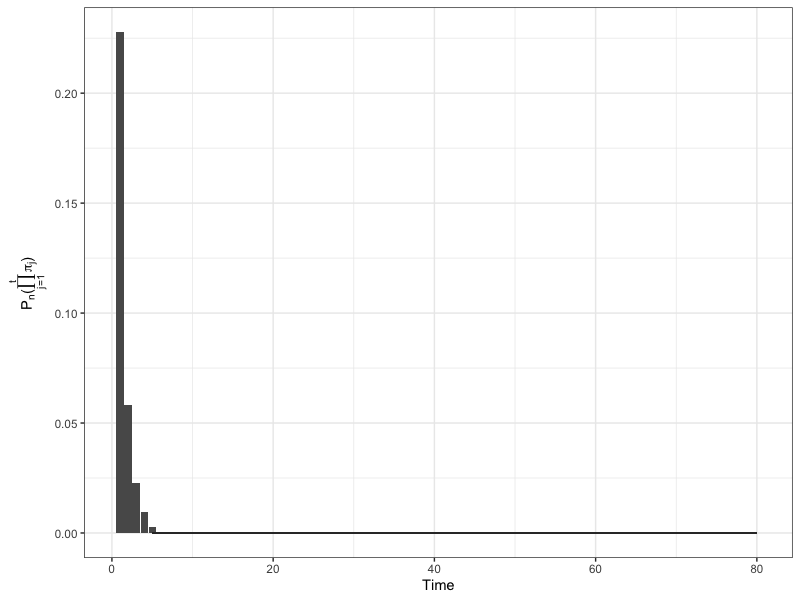}
			\caption*{(b)}
		\end{minipage}
		\caption{(a) The average propensity score over the course of follow-up. We observe that due to the non-complinace, the average propensity score sharply decreases over time, which strongly hints at positivity violation in the EAGeR dataset. (b) $\Pn(\prod_{j=1}^t \widehat{\pi}_j)$ over the course of follow-up. When $t \geq 5$, $\Pn(\prod_{j=1}^t \widehat{\pi}_j)$ becomes less than $5\times10^{-4}$, which makes an IPW estimation in MSMs infeasible. We used Random Forests (via the \texttt{ranger} package in R) to estimate $\pi_t$.}
		\label{fig:ps-always-treated}
	\end{figure}
	
	Even if positivity is only nearly violated, it can pose a serious problem in attempting to estimate our target causal effect. One of the most widely-used approaches to handle time-varying treatments is marginal structural models (MSMs) \citep{robins2000marginal}. In practice, MSMs are often estimated via inverse probability weighting (IPW). The following quantity appears in the IPW (also in the doubly robust) moment condition
	\[
	 h(\overline{A}_T)\left\{ \frac{Y - m(\overline{A}_T;\beta)}{\prod_{t=1}^T \widehat{\pi}_t} \right\},
	\]
    for any choice of $h$ (with matching dimensions) where $\widehat{\pi}_t(a_t)=\widehat{\Pb}(A_t=a_t\mid H_t)$. However, Figure \ref{fig:ps-always-treated}-(b) indicates that on average a cumulative product of propensity score sharply drops to zero even with moderate $t$. This would make standard estimation techniques such as IPW to fail as $\Pn(\prod_{j=1}^T \widehat{\pi}_j)$ easily blows up. 
    
    More specifically, when we parametrically model the effect curve by $\E[Y^{\overline{a}_T}] = m(\overline{a}_T;\beta) = \beta_0 + \sum_{t=1}^T \beta_{1t}a_t$ so that the coefficients for exposure can vary with time, an inverse-weighted MSM estimator that is the solution to 
	\[
	 \Pn \left[h(\overline{A}_T)\left\{ \frac{Y - m(\overline{A}_T;\beta)}{\prod_{t=1}^T \widehat{\pi}_t} \right\} \right] = 0
	\]
	indeed fails since no coefficient estimates can be found in the above equation even for moderate values of $T$, e.g., $T = \sim 10$. Thus, it appears that positivity violation in our dataset precludes the standard MSM-based approach. We remark that these limitations are not at all unique to the analysis of our EAGeR dataset, but also common to many observational studies based on the MSM or other approaches \citep[e.g.,][]{luedtke2017sequential}.

	\subsection{Alternative approach}
	
    Due to the positivity violation, the estimation results, if any, through standard approaches will remain dubious at best. Therefore, we alter our target contrast from the standard ATE to the mean outcome we would have observed in a population if ``observed" versus none (not all versus none) were treated, which is defined by
    \begin{equation} \label{def:ATE-observed-vs-treated}
    \begin{aligned}
    \tau_{\text{obs}}(T) \equiv \mathbb{E}\big[ Y^{\bar{A}_T = \overline{a^{\text{obs}}}, \bar{R}_T=\bar{\mathbf{1}}} \big] -\mathbb{E}\big[ Y^{\bar{A}_T=\bar{\mathbf{0}},  \bar{R}_T=\bar{\mathbf{1}}}\big],
    \end{aligned}
    \end{equation}
    where $\overline{a^{\text{obs}}}$ denotes an observed history of aspirin consumption. This new estimand would tell us how the mean outcome would have changed if no one in the population had taken aspirin throughout the study. In this way, we can avoid estimating the problematic counterfactual $\mathbb{E}\big[ Y^{\bar{A}_T=\bar{\mathbf{1}},  \bar{R}_T=\bar{\mathbf{1}}}\big]$. However, by construction this solution entails the fundamental limitation because we have sacrificed the causal effect of original interest.
    
    In order to estimate our new causal parameter \eqref{def:ATE-observed-vs-treated}, here we use the g-computation \footnote{We also tried a weighting estimator but omitted the result here, since it gives almost the same result as the g-computation, only with wider confidence bands.} (plug-in) estimator \citep{robins1986} and the sequential doubly robust (SDR) estimator proposed by \citet{luedtke2017sequential} which allows right-censored data structures.

    \subsection{Estimation and inference}
    
    \textbf{Estimation.}
    First for the g-computation estimator, we estimate the following g-formula
    \begin{equation*}
        \begin{aligned}
        \mathbb{E}\big[Y^{\bar{A}_T=\bar{a}_T,  \overline{R}_T= \overline{\mathbf{0}} } \big] &= \int \cdots \int \mathbb{E}\big[Y|\overline{X}_T,  \overline{A}_T = \overline{a}_T, \bar{R}_T= \bar{\mathbf{1}}_T \big] \prod_{t=2}^T d\Pb(X_t |\overline{X}_{t-1},  \overline{A}_{t-1} = \overline{a}_{t-1}, \overline{R}_{t-1}= \overline{\mathbf{1}}_{t-1} )\\
        & \qquad \qquad \quad \times  d\Pb({X}_1, {A}_1 = {a}_1, {R}_1 = 1) 
        \end{aligned}
    \end{equation*}
    via plug-in estimators of the pseudo-outcome regression function each time step. Next, for the SDR estimator, we tailor Algorithm 2 of \citet{luedtke2017sequential} for our right-censored data structures (everything remains the same except that we add the condition $\overline{R}_{t-1}= \overline{\mathbf{1}}_{t-1}$ on each pseudo-outcome regression function). For both methods, we use the same nonparametric ensemble we used in Section \ref{sec:application}. 
    
    \textbf{Inference.}
    Confidence intervals are estimated by bootstrapping at 95\% level for both of the estimators. Note that for the SDR estimator, we are guaranteed to consistently estimate standard errors (pointwisely) by bootstrapping due to the following asymptotically property,
    \begin{equation*}
    \begin{aligned}
    \sqrt{n} (\widehat{\tau}_{\text{obs}}(t)- \tau_{\text{obs}}(t)) \leadsto \mathcal{N}(0, Var(\phi_{\tau}(t)))
    \end{aligned}
    \end{equation*}
    for all $t\leq T$, where $\phi_{\tau}(t)$ is the influence function of $\widehat{\tau}_{\text{obs}}(t)$. However, this is no longer guaranteed for the g-computation estimator.

    \subsection{Result}
    
    For the sake of completeness, we estimate each $\tau_{\text{obs}}(t)$ for all $t=2 \sim 89$ and present the cumulative effects over time $t$. The results for the g-computation and the SDR estimators are presented in Figures \ref{fig:aspirin-other-g-computation}, \ref{fig:aspirin-other-SDR}, respectively.
    
    \begin{figure}[!htb]
	\minipage{0.45\textwidth}
	\includegraphics[width=\linewidth]{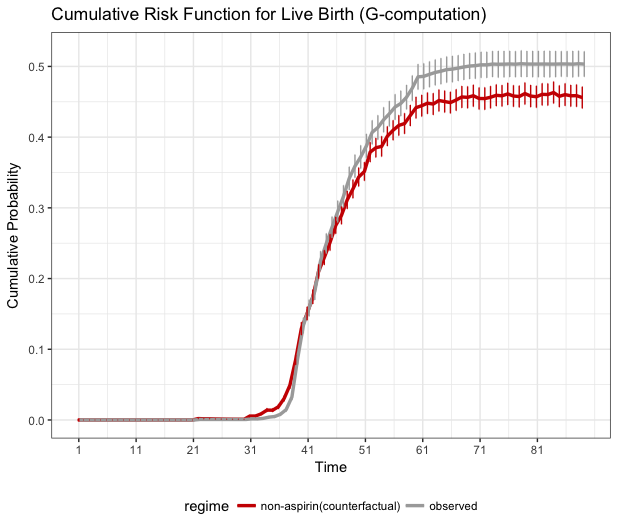}
	\caption*{Live birth} 
	\endminipage\hfill
	\minipage{0.45\textwidth}
	\includegraphics[width=\linewidth]{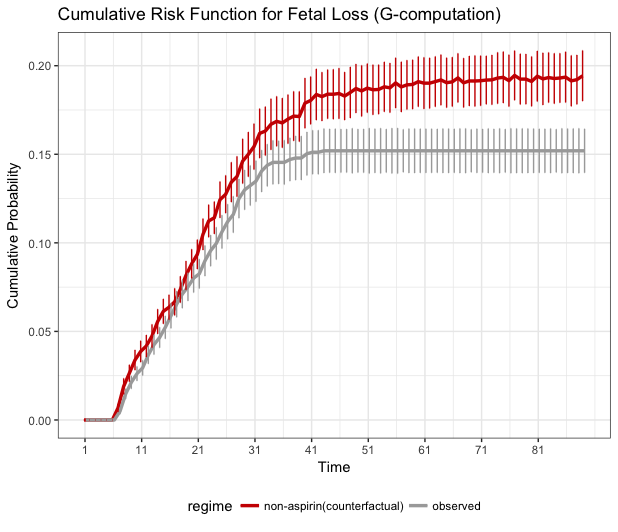}
	\caption*{Pregnancy loss} 
	\endminipage
	\caption{Cumulative risk curve for live birth and pregnancy loss via the regression based g-computation estimator. Pointwise 95\% confidence interval is estimated by bootstrapping with 1000 resampling.} \label{fig:aspirin-other-g-computation}
    \end{figure}

	\begin{figure}[!htb]
	\minipage{0.45\textwidth}
	\includegraphics[width=\linewidth]{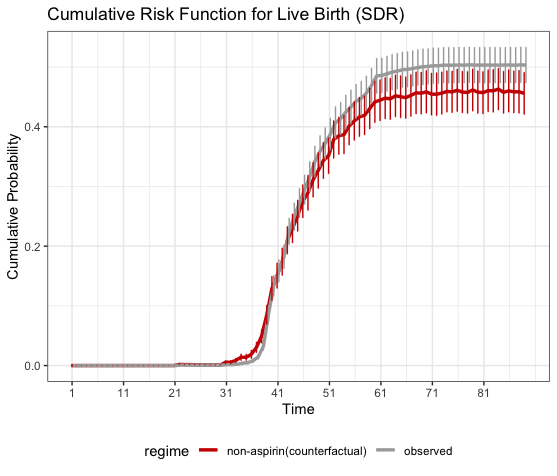}
	\caption*{Live birth}
	\endminipage\hfill
	\minipage{0.45\textwidth}
	\includegraphics[width=\linewidth]{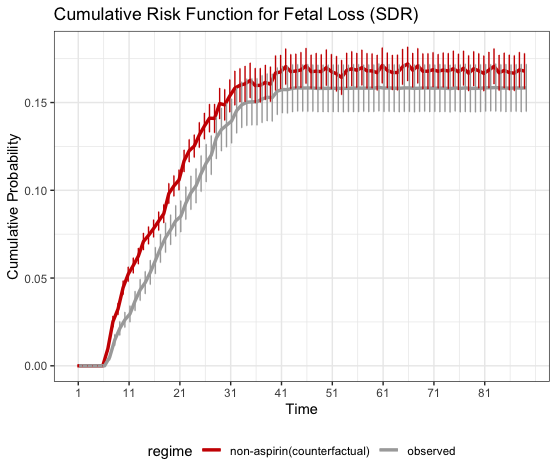}
	\caption*{Pregnancy loss}
	\endminipage
	\caption{Cumulative risk curve for live birth and pregnancy loss via the sequential doubly robust (SDR) estimator. Pointwise 95\% confidence interval is estimated by bootstrapping with 1000 resampling.} \label{fig:aspirin-other-SDR}
    \end{figure}
	
	The result based on the g-computation estimator in Figure \ref{fig:aspirin-other-g-computation} shows that the counterfactual mean outcomes for never-takers (individuals who have never taken aspirin throughout the study) are worse-off than the observed. Specifically, for the never-takers the probability of having live birth has been decreased and the probability of having fetal loss has increased. The result seems to be statistically significant at $T=89$.
	
	On the other hand, the result based on the SDR estimator in Figure \ref{fig:aspirin-other-SDR} indicates that although the mean effects for the never-takers still appear to be worse off than the observed, they look no longer statistically significant. Hence in this case we cannot draw any firm conclusion about the effect of aspirin on pregnancy outcome.
	
	It might be tempting to take the results from Figure \ref{fig:aspirin-other-g-computation} as it seems to deliver more clear messages. However, we do not know if our variance estimates there are correct in the first place. Accuracy of our estimate is further afflicted by the moderate sample size ($n$=1024) due to the slow convergence rates of the plug-in regression. These issues can be mitigated in the SDR estimator. Thus, we should rather resort to the results presented in Figure \ref{fig:aspirin-other-SDR}, which basically tells us that the effect of low-dose aspirin is insignificant and remains dubious, based on the causal effect defined in \eqref{def:ATE-observed-vs-treated}.
	
	After all, it should be noted that due to the positivity violation we end up limiting ourselves to the more narrow notion of causal effects (i.e. observed versus none) which is different from the ATE type estimands that are typically of utmost interest for policy makers. The causal effect in \eqref{def:ATE-observed-vs-treated} might not be practically meaningful as to aspirin prescription for pregnant since we are in general much more interested in the always-taker group than the never-taker group.

    \section{More details on influence functions and efficiency bound}
	\label{sec:ifdetails}
	
	Here, we shall introduce the {influence function}, which is a foundational object of statistical theory that allows us to characterize a wide range of estimators with favorable theoretical properties. There are two notions of the influence function: one for estimators and the other for parameters. To distinguish these two cases we will call the latter, which corresponds to parameters, \emph{influence curves} as in for example, \citet{boos2013essential, Kennedy16} \footnote{However, the terms `influence curve' and `influence function' are used interchangeably in many cases.}. Before we go on, we declare that the primary sources of this section are \citet{kennedy2014semiparametric, Kennedy16, Edward20Slide} from which all the terms, definitions and results are directly borrowed.

    First, we give a definition of influence curves. It was first introduced by \citet{Hampel74} and studied to provide a general solution to find approximation-by-averages representation for a functional statistic. We only consider nonparametric models here.
    
    Suppose that we are given a target functional $\psi$. For a nonparametric model $\Pb$, let $\{ \Pb_\epsilon$, $\epsilon \in \mathbb{R} \}$ denote a smooth parametric submodel for $\Pb$ with $\Pb_{\epsilon=0} = \Pb$. A  typical example of this parametric submodel can be given by $\{\Pb_\epsilon: p_\epsilon(z) = p(z)(1 + \epsilon s(z)) \}$ for some mean-zero, uniformly bounded function $s$. Then the \textit{influence curve} for parameter $\psi(\Pb)$ is defined by any mean-zero, finite-variance function $\phi(\Pb)$ that satisfies the following \textit{pathwise differentiability},
    \begin{equation} \label{def:influence-curve}
    \frac{\partial}{\partial \epsilon}\psi(\Pb_\epsilon) \Bigg\vert_{\epsilon=0} = \int \phi(\Pb) \left(\frac{\partial}{\partial\epsilon}\log d\Pb_\epsilon \right) \Bigg\vert_{\epsilon=0} d\Pb.
    \end{equation}
    
    The above pathwise differentiability implies that our target parameter $\psi$ is smooth enough to admit a von Mises expansion: for two distribution $\Pb, \Qb$
    \begin{equation} \label{def:von-Mises-expansion}
    \psi(\mathbb{Q}) - \psi(\mathbb{P}) = \int \phi (\mathbb{Q}) d(\Qb - \Pb) + R_2(\Qb, \Pb)
    \end{equation}
    where $R_2$ is a second-order remainder. Therefore, the influence curve corresponds to the functional derivative in a Von Mises expansion of $\psi$.
    
    One can obtain the classical Cram\'{e}r-Rao lower bound for each parametric submodel $\Pb_\epsilon$; the Cram\'{e}r-Rao lower bound for $\Pb_\epsilon$ is $\psi^\prime(\Pb_\epsilon)^2 / \E(s_\epsilon^2)$ where $\psi^\prime(\Pb_\epsilon) = \frac{\partial}{\partial \epsilon}\psi(\Pb_\epsilon) \big\vert_{\epsilon=0}$ and $s_\epsilon = s_\epsilon(z) = \frac{\partial}{\partial\epsilon}\log d\Pb_\epsilon \big\vert_{\epsilon=0}$. The asymptotic variance of any nonparametric estimator is no smaller than the supremum of the Cram\'{e}r-Rao lower bounds for all parametric submodel, and it is known that under the above pathwise differentiability condition the greatest such lower bound is given by
    \[
    	\underset{\Pb_\epsilon}{\sup} \frac{\psi^\prime(\Pb_\epsilon)^2}{\E(s_\epsilon^2)} \leq \E(\phi^2).
    \]
    Hence, $\E(\phi^2) = \var(\phi)$ is the nonparametric analog of the Cram\'{e}r-Rao lower bound, and we call the influence curve that attains the above bound the \textit{efficient influence curve}. The efficient influence curve gives the efficiency bound for estimating $\psi$. In parametric models, more than one influence curves may exist. On the other hand in nonparametric model, the influence curve is unique. However, the efficient influence curve is always unique in any cases.
    
    Once the efficient influence curve is known, no estimator can be more efficient than $\hat{\psi}(\Pb)$ such that
    \begin{equation} \label{eqn:EIF-converging-distribution}
    	\sqrt{n}(\hat{\psi} - \psi) \rightsquigarrow N(0, var(\phi))
    \end{equation}
    as $\var(\phi)$ serves to be our nonparametric efficiency bound. In \eqref{eqn:EIF-converging-distribution}, we call $\phi$ the (efficient) \textit{influence function} for the estimator $\hat{\psi}$ \footnote{In fact, influence curves themselves are the putative influence functions.}. For each nonparametric estimator, the efficient influence function, if exists, is almost surely unique, so in this sense the influence function contains all information about an estimator’s asymptotic behavior. In other words, if we know the influence function for an estimator, we know its asymptotic distribution and can easily construct confidence intervals and hypothesis tests.
    
    Characterizing the influence curves is crucial not only to give the efficiency bound for estimating $\psi$, thus providing
    a benchmark against which estimators can be compared, but probably more importantly, to construct estimators with very favorable properties, such as double robustness or general second-order bias. One may can find an (asymptotically linear) estimator that satisfies \eqref{eqn:EIF-converging-distribution} by solving appropriate estimating equations using the influence curves. Section \ref{proof:thm-eif} of the appendix contains an example of developing an efficient, model-free estimator based on the efficient influence curve of the target parameter.
    
    Finally we remark that for complicated functionals, pretending discrete space on $Z$ can facilitate our procedure to characterize influence curves. For example, assuming that our unit space is discrete,  the \textit{influence curve} $\phi(\Pb)$ for the functional $\psi(\Pb)$ can be defined by
    \begin{equation} \label{def:influence-curve-Gateaux-deriv}
    \phi(\Pb) = \frac{\partial}{\partial \epsilon} \psi \left( (1-\epsilon)\Pb + \epsilon\delta_z  \right)  \Big\vert_{\epsilon=0^+} = \underset{\epsilon \rightarrow 0^+}{\lim} \frac{\psi\left( (1-\epsilon)\Pb + \epsilon\delta_z  \right) - \psi(\Pb)}{\epsilon}
    \end{equation}
    where we let $\delta_z$ be the Dirac measure at $Z = z$. This definition is equivalent to the \textit{Gateaux derivative} of $\psi$ at  $\Pb$ in direction of point mass $(\delta_z - \Pb)$ (see, for example, Chapter 5 in \citet{boos2013essential}).
    
    For more details for nonparametric efficiency theory and influence functions,  we refer to \citet{kennedy2014semiparametric, Kennedy16, Edward20Slide, van2003unified, Tsiatis06}.

	\section{Additional Technical Results}
	
	\subsection{Sequential regression formulation} 
	\label{sequential-regression-formulation}
	
	The efficient influence function derived in the previous subsection involves pseudo-regression functions $m_s$. To avoid complicated conditional density estimation, as suggested by \citet{Kennedy17}, one may formulate a series of sequential regressions for $m_s$, as described in the subsequent remark.
	
	\begin{remark} \label{rmk:m_s}
		From the definition of $m_s$, it immediately follows that
		$$m_s = \int_{\mathcal{X}_s \times \mathcal{A}_s} m_{s+1}dQ_{s+1}(a_{s+1} \mid h_{s+1}, R_{s+1}=1) d\Pb(x_{s+1}|h_s,a_s, R_{s+1}=1).$$ 
		Hence, we can find equivalent form of the functions $m_s(\cdot)$ in Theorem \ref{thm:eif} as the following recursive regression:
		\begin{align*}
			& m_s(H_s, A_s, R_{s+1}=1)\\
			& = \E \left[ \frac{m_{s+1}(H_{s+1}, 1, 1) \delta\pi_{s+1}(H_{s+1}) + m_{s+1}(H_{s+1}, 0, 1) \{ 1 - \pi_{s+1}(H_{s+1}) \} }{\delta\pi_{s+1}(H_{s+1}) + 1 - \pi_{s+1}(H_{s+1})} \Bigg\vert H_s, A_s, R_{s+1}=1 \right]
		\end{align*}
		for $s= 1, ... , t-1$, where we use shorthand notation $m_{s+1}(H_{s+1}, a_{s+1}, 1) = m_{s+1}(H_{s+1}, A_{s+1}=a_{s+1}, R_{t+2}=1)$ and $m_s(H_s,A_s,1)=\mu(H_s,A_s, R_{s+1}=1)$.
	\end{remark}
	
	Above  sequential  regression  form  is  practically useful since it allows us to bypass all the conditional density estimations and instead use regression methods that are more readily available in statistical software.
	
	\subsection{EIF for $T=1$}
	\label{eif-for-T=1}
	
	In the next corollary we provide the efficient influence function for the incremental effect for a single timepoint study ($T=1$) whose identifying expression is given in Corollary \ref{cor:ident-exp-pt-trt}. 
	
	\begin{cor} \label{cor:eif-single-exp}
		When $T=1$, the efficient influence function for $\psi(\delta)$ in Corollary \ref{cor:ident-exp-pt-trt} is given by 
		\begin{align*}
			\mathbbm{1}\left(R=1\right)\left[\frac{\delta\pi(1\vert X)\phi_{1,R=1}(Z)+\pi(0\vert X)\phi_{0,R=1}(Z)}{\delta\pi(1\vert X) + \pi(0\vert X)} + \frac{\delta\{\mu(X,1,1) - \mu(X,0,1) \}\left(A-\pi(1\vert X)\right)}{\left\{\delta\pi(1\vert X)+\pi(0\vert X)\right\}^2}\right]
		\end{align*}
		where
		$$
		\mu(x,a, 1) = \E(Y \mid X = x, A = a, R=1),
		$$
		$$
		\pi(a\vert X) = d\Pb(A=a \mid X = x), 
		$$
		$$
		\omega(X,a) = d\Pb(R=1 \mid X = x, A=a), 
		$$
		and 
		$$
		\phi_{a,R=1}(Z) = \frac{\mathbbm{1}\left(A=a\right)\mathbbm{1}\left(R=1\right)}{\pi(a\vert X)\omega(X,a)}\left\{Y- \mu(X,a,1)\right\} + \mu(X,a,1)
		$$
		which is the uncentered efficient influence function for $\E[\mu(X,a, 1)]$.
	\end{cor}
	
	The efficient influence function for the point exposure case has a simpler and more intuitive form. As stated in Corollary \ref{cor:eif-single-exp}, it is a weighted average of the two efficient influence functions $\phi_{0,R=1}, \phi_{1,R=1}$, plus a contribution term due to unknown propensity scores. 
	
	\section{Proofs}
	
	\subsection{Lemma for the identifying expression in Theorem \ref{thm:ident-exp}}
    \label{proof:lem-ident-assumption}
	
	Without assumptions (A2-M) and (A3), our target parameter $\psi_t(\delta) = \E\left(Y_t^{\overline{Q}_t(\delta) } \right)$ would not be identified. The following lemma extends Theorem 1 in \citet{Kennedy17} to our setting.

    \begin{lemma} \label{lem:identification}
    	Under (A2-M) and (A3), and for all $t \leq T$, we have following identities:
    	\begin{itemize}
    		\setlength\itemsep{0.1em}
    		\item[a.] $d\Pb(A_t|H_t) = d\Pb(A_t|H_t, R_t=1)$
    		\item[b.] $d\Pb(Y_{t-1},X_t|A_{t-1}, H_{t-1}) = d\Pb(Y_{t-1},X_t|A_{t-1}, H_{t-1}, R_t=1)$
    		\item[c.] $\E[Y_{t}|{H}_t, {A}_t] =  \E[Y_{t}|{H}_t, {A}_t, R_{t+1}=1]$ 
    	\end{itemize}
    \end{lemma}
    
    
    \begin{proof}
    \hfill
    \begin{enumerate}[label=\alph*.,leftmargin=.15in]
    	\setlength\itemsep{1em}
    	\item $\bm{d\Pb(A_t|H_t) = d\Pb(A_t|H_t, R_t=1)}$ 
    	By abuse of notation, for $s < t$, here we let $\underline{X}_s$, $\underline{A}_s$ represent $(X_s,...,X_t)$, $(A_s,...,A_t)$ respectively, and $\underline{Y}_{s-1}$ represent $(Y_{s-1},...,Y_{t-1})$.
    	First note that
    	\begin{align*}
    	d\Pb(A_t,H_t) &= d\Pb(\overline{X}_t, \overline{A}_t, \overline{Y}_{t-1}) = d\Pb(\underline{X}_2, \underline{A}_2, \underline{Y}_1 \mid X_1, A_1)d\Pb(X_1,A_1) \\
    	&= d\Pb(\underline{X}_2, \underline{A}_2, \underline{Y}_1 \mid X_1, A_1, R_2 = 1)d\Pb(X_1,A_1, R_1 = 1) \\
    	&= d\Pb(\underline{X}_3, \underline{A}_3, \underline{Y}_2 \mid \overline{X}_2, \overline{A}_2, \overline{Y}_1, R_2 = 1) \frac{d\Pb(X_1,A_1, R_1 = 1)}{d\Pb(X_1,A_1,R_2 = 1)} d\Pb(\overline{X}_2, \overline{A}_2,\overline{Y}_1, R_2 = 1) \\
    	&= d\Pb(\underline{X}_3, \underline{A}_3, \underline{Y}_2 \mid \overline{X}_2, \overline{A}_2, \overline{Y}_1, R_3 = 1) \frac{d\Pb(X_1,A_1, R_1 = 1)}{d\Pb(X_1,A_1, R_2 = 1)} d\Pb(\overline{X}_2, \overline{A}_2, \overline{Y}_1, R_2 = 1) \\
    	& = d\Pb({X}_t, {A}_t, {Y}_{t-1} \mid \overline{X}_{t-1}, \overline{A}_{t-1}, \overline{Y}_{t-2}, R_{t} = 1) \\
    	& \quad \times \prod_{s=1}^{t-2}  \frac{d\Pb(\overline{X}_s, \overline{A}_s, \overline{Y}_{s-1}, R_s = 1)}{d\Pb(\overline{X}_s, \overline{A}_s, \overline{Y}_{s-1}, R_{s+1} = 1)}  d\Pb(\overline{X}_{t-1}, \overline{A}_{t-1}, \overline{Y}_{t-2}, R_{t-1} = 1) \\
    	&= d\Pb(\overline{X}_t, \overline{A}_t, \overline{Y}_{t-1}, R_{t} = 1) \prod_{s=1}^{t-1}  \frac{d\Pb(\overline{X}_s, \overline{A}_s, \overline{Y}_{s-1}, R_s = 1)}{d\Pb(\overline{X}_s, \overline{A}_s, \overline{Y}_{s-1}, R_{s+1} = 1)}  
    	\end{align*}
    	where the first equality follows by definition, the second by definition of conditional probability, the third by assumption (A2-M), the fourth again by definition of conditional probability, the fifth by assumption (A2-M), and the sixth by repeating the same step $t-1$ times. The last expression is obtained by simply rearranging terms using the definition of conditional probability.
    	
    	Now we let 
    	$$
    	\mathlarger{\mathlarger{\bm{\Pi}}}_{\Pb}(t-1) \equiv \prod_{s=1}^{t-1}  \frac{d\Pb(\overline{X}_s, \overline{A}_s, \overline{Y}_{s-1}, R_s = 1)}{d\Pb(\overline{X}_s, \overline{A}_s, \overline{Y}_{s-1}, R_{s+1} = 1)}
    	$$
    	so we can write $d\Pb(A_t,H_t) = d\Pb(\overline{X}_t, \overline{A}_t, \overline{Y}_{t-1}, R_{t} = 1) \mathlarger{\mathlarger{\bm{\Pi}}}_{\Pb}(t-1)$.
    	
    	Then, similarly we have
    	\begin{align*}
    	d\Pb(H_t) &= d\Pb(\overline{X}_t, \overline{A}_{t-1}, \overline{Y}_{t-1}) = d\Pb(\overline{X}_t, \overline{A}_{t-1}, \overline{Y}_{t-1}, R_t=1) \mathlarger{\mathlarger{\bm{\Pi}}}_{\Pb}(t-1).
    	\end{align*}
    	Hence, finally we obtain
    	\begin{align*}
    	d\Pb(A_t \mid H_t) &= \frac{d\Pb(A_t,H_t)}{d\Pb(H_t)} = \frac{d\Pb(\overline{X}_t, \overline{A}_t,\overline{Y}_{t-1},  R_{t} = 1)}{d\Pb(\overline{X}_t, \overline{A}_{t-1}, \overline{Y}_{t-1}, R_t=1)} \\
    	&= \frac{d\Pb({A}_t, {H}_t, R_{t} = 1)}{d\Pb({H}_t, R_t=1)} \\
    	& = d\Pb(A_t|H_t, R_t=1)
    	\end{align*}
    	where the second equality comes from the above results. The proof naturally leads to  $\bm{dQ_t(A_t|H_t) = dQ_t(A_t|H_t, R_t=1)}$.
    	
    	\item $\bm{d\Pb(Y_{t-1}, X_t|A_{t-1}, H_{t-1}) = d\Pb(Y_{t-1}, X_t|A_{t-1}, H_{t-1}, R_t=1)}$ 
    	
    	By definition $d\Pb(Y_{t-1}, X_t|A_{t-1}, H_{t-1})= d\Pb(H_t)/d\Pb(A_{t-1}, H_{t-1})$, and from the part a) it immediately follows
    	\begin{align*}
    	& d\Pb(H_t) = d\Pb(\overline{X}_t, \overline{A}_{t-1}, \overline{Y}_{t-1}, R_t=1) \mathlarger{\mathlarger{\bm{\Pi}}}_{\Pb}(t-1), \\
    	& d\Pb(A_{t-1}, H_{t-1}) =  d\Pb(\overline{X}_{t-1}, \overline{A}_{t-1}, \overline{Y}_{t-2}, R_{t-1} = 1) \mathlarger{\mathlarger{\bm{\Pi}}}_{\Pb}(t-2) .
    	\end{align*}
    	Hence, we have
    	\begin{align*}
    	\frac{d\Pb(H_t)}{d\Pb(A_{t-1}, H_{t-1})} &= \frac{d\Pb(\overline{X}_t, \overline{A}_{t-1}, \overline{Y}_{t-1}, R_t=1)}{d\Pb(\overline{X}_{t-1}, \overline{A}_{t-1}, \overline{Y}_{t-2}, R_{t} = 1)} \\
    	&= d\Pb(Y_{t-1}, X_t \mid \overline{H}_{t-1}, {A}_{t-1}, R_{t} = 1)
    	\end{align*}
    	which yields the desired result.
    	
    	\item $\bm{\E[Y_{t}|{H}_t, {A}_t] =  \E[Y_{t}|{H}_t, {A}_t, R_{t+1}=1]}$ 
    	
    	By definition
    	$
    	\E[Y_{t}|{H}_t, {A}_t] = \int y d\Pb(Y_{t} = y\vert {H}_t, {A}_t),
    	$
    	and thereby it suffices to show that $d\Pb(Y_{t} \vert {H}_t, {A}_t) = d\Pb(Y_{t} \vert {H}_t, {A}_t, R_{t+1})$. 
    	
    	By the same logic we used for the first proof, we have
    	\begin{align*}
    	d\Pb(Y_{t}, {H}_t, {A}_t) = d\Pb(Y_t, {H}_t, {A}_t, R_t=1) \mathlarger{\mathlarger{\bm{\Pi}}}_{\Pb}(t-1)
    	\end{align*}
    	and also
    	$$
    	d\Pb({H}_t, {A}_t) =  d\Pb({H}_t, {A}_t, R_{t} = 1) \mathlarger{\mathlarger{\bm{\Pi}}}_{\Pb}(t-1).
    	$$
    	Hence, by Assumption (A2-M) we have that
    	$$
    	d\Pb(Y_{t} \mid {H}_t, {A}_t) = d\Pb(Y_{t} \mid {H}_t, {A}_t, R_t = 1) = d\Pb(Y_{t} \mid {H}_t, {A}_t, R_{t+1} = 1).
    	$$
    \end{enumerate}
    
    \end{proof}
	
	Following the exact same logic used in the proof of \citet[][Theorem 1]{Kennedy17}, under Assumptions \ref{assumption:A1} and \ref{assumption:A2-E}, for all $s < t$ we have the recursion formula
	\[
	\E\{{Y_t}^{(\overline{a}_{s-1},\underline{Q}_s)} \mid H_{s-1}, A_{s-1}\} = \int_{\mathcal{X}_s\times\mathcal{A}_s} \E\{{Y_t}^{(\overline{a}_{s},\underline{Q}_{s+1})} \mid H_{s} = h_s, A_{s}=a_s\} dQ_s(a_s \mid h_s) d\Pb (y_{s-1}, x_s \mid h_{s-1}, a_{s-1} ).
	\]
	Applying the above $t$ times leads to
	\begin{align*}
	    \E\{{Y_t}^{\overline{Q}_t}\} &= \int_{\overline{\mathcal{X}}_t\times\overline{\mathcal{A}}_t} \E\{{Y_t}^{\overline{a}_{t}} \mid H_{t} = h_t, A_{t}=a_t\} \prod_{s=1}^{t} dQ_s(a_s \mid h_s) d\Pb (y_{s-1}, x_s \mid h_{s-1}, a_{s-1}).
	\end{align*}
	
	Finally, Assumption \ref{assumption:A1} and Lemma \ref{lem:identification} give
	\begin{align*}
	    \E\{{Y_t}^{\overline{Q}_t}\} &= \int_{\overline{\mathcal{X}}_t\times\overline{\mathcal{A}}_t} \E\{{Y_t} \mid H_{t} = h_t, A_{t}=a_t, R_{t+1}=1\} \prod_{s=1}^{t} dQ_s(a_s \mid h_s, R_s=1) d\Pb (y_{s-1}, x_s \mid h_{s-1}, a_{s-1}, R_s=1).
	\end{align*}
	
	\subsection{Proof of Theorem \ref{thm:eif}}
	\label{proof:thm-eif}
	
	\subsubsection{Identifying expression for the efficient influence function}
	
	In the next lemma, we provide an identifying expression for the efficient influence function for our incremental effect $\psi_t(\delta)$ under a nonparametric model, which allows the data-generating process $\Pb$ to be infinite-dimensional.
	
	\begin{lemma} \label{lem:eif}
		Define
		\begin{align*}
			m_s&(h_s,a_s, R_{s+1}=1) \\
			& = \int_{ \mathcal{R}_s} \mu(h_{t},a_{t}, R_{t+1}=1)  \prod_{k=s+1}^{{t}} dQ_k(a_k \mid h_k, R_k=1) d\Pb(y_{k-1},x_k|h_{k-1},a_{k-1}, R_k=1)
		\end{align*} 
		for $s=0,...,{t}-1$, $\forall t \leq T$,  where we write $\mathcal{R}_s = (\overline{\mathcal{X}}_{t}\times \overline{\mathcal{A}}_{t}) \setminus (\overline{\mathcal{X}}_{s}\times \overline{\mathcal{A}}_{s})$ and $\mu(h_{t},a_{t}, R_{t+1}=1) = \E(Y_t \mid H_{t} = h_{t}, A_{t} = a_{t}, R_{t+1}=1)$. For $s=t$ and $s=t+1$, we set $m_{s}(\cdot)=\mu(h_{t},a_{t}, R_{t+1}=1)$ and $m_{t+1}(\cdot)=Y$. Moreover, let $\frac{\mathbbm{1}(H_s=h_s, R_s=1)}{d\Pb(h_s,R_s=1)} \phi_s(H_s, A_s, R_s=1;a_s)$ denote the efficient influence function for $dQ_s(a_s|h_s,R_s=1)$. 
		
		Then, the efficient influence function for $m_0=\psi_t (\delta)$ is given by
		\begin{equation*} 
			\begin{aligned}
				& \sum_{s=0}^{t} \left\{  \int_{ \mathcal{A}_{s+1}} m_{s+1}(H_{s+1}, A_{s+1}, R_{s+2}=1)dQ_{s+1}(a_{s+1}|H_{s+1},R_{s+1}=1) - m_s(H_{s}, A_{s}, R_{s+1}=1) \right\}  \\
				& \qquad \times \mathbbm{1}\left(R_{s+1}=1\right)  \left(  \prod_{k=0}^{s} \frac{dQ_k(A_k \mid H_k, R_k=1)}{d\Pb(A_k \mid H_k, R_k=1)}   \frac{1}{d\Pb(R_{k+1}=1 \mid H_k, A_k,R_k=1)} \right) \\
				& + \sum_{s=1}^{t} \mathbbm{1}(R_s=1) \left(  \prod_{k=0}^{s-1} \frac{dQ_k(A_k \mid H_k, R_k=1)}{d\Pb(A_k \mid H_k, R_k=1)}  \frac{1}{d\Pb(R_{k+1}=1 \mid H_k, A_k,R_k=1)}  \right) \\
				& \qquad \quad  \times \int_{ \mathcal{A}_{s}}  m_s(H_s, a_s,R_{s+1}=1) \phi_s(H_s, A_s, R_s=1;a_s) d\nu(a_s) 
			\end{aligned}
		\end{equation*} 
		where we define $dQ_{t+1} = 1$, $m_{t+1}(\cdot)=Y$, and $dQ_0(a_0|h_0)/d\Pb(a_0|h_0) = 1$, and $\nu$ is a dominating measure for the distribution of $A_s$.
	\end{lemma}
	
	The proof of Lemma \ref{lem:eif} involves derivation of efficient influence function for more general stochastic interventions that depend on the both observational propensity scores and right-censoring process. We begin by presenting the following three additional lemmas.
	
	\begin{lemma}[\cite{Kennedy17}] \label{lem:eif_dQ}
		For $\forall t$, the efficient influence function for 
		\begin{equation*} 
			\begin{aligned}
				dQ_t(a_t \mid h_t, R_t=1) = \frac{a_t\delta\pi_{t}(h_t) + (1-a_t) \{ 1 - \pi_t(h_t) \} }{\delta\pi_{t}(h_t) + 1 - \pi_{t}(h_t)} 
			\end{aligned}
		\end{equation*} 
		which is defined in (\ref{eqn:incr-intv-ps}) is given by $\frac{\mathbbm{1}(H_t=h_t, R_t=1)}{d\Pb(h_t,R_t=1)} \phi_t(H_t, A_t, R_t=1;a_t)$, where $\phi_t(H_t, A_t, R_t=1;a_t)$ equals
		\begin{equation*} 
			\begin{aligned}
				\frac{(2a_t-1)\delta\{A_t-\pi_t(H_t)\}}{\left( \delta\pi_t(H_t) + 1 -  \pi_t(H_t) \right)^2}
			\end{aligned}
		\end{equation*} 	
		where $\pi_t(h_t) = \Pb(A_t=1 \mid H_t=h_t, R_t=1)$.
	\end{lemma}
	
	\begin{lemma} \label{lem:eif_1}
		Suppose $\overline{Q}_T$ is not depending on $\Pb$. Recall that for $\forall t \leq T$,
		\begin{align*} 
			m_s&(h_s,a_s, R_{s+1}=1) 
			& = \int_{ \mathcal{R}_s} \mu(h_{t},a_{t}, R_{t+1}=1)  \prod_{k=s+1}^{{t}} dQ_k(a_k \mid h_k, R_k=1) d\Pb(y_{k-1},x_k|h_{k-1},a_{k-1}, R_k=1)
		\end{align*} 
		for $s=0,...,{t}-1$, where we write $\mathcal{R}_s = (\overline{\mathcal{X}}_{t}\times \overline{\mathcal{A}}_{t}) \setminus (\overline{\mathcal{X}}_{s}\times \overline{\mathcal{A}}_{s})$ and $\mu(h_{t},a_{t}, R_{t+1}=1) = \E(Y_t \mid H_{t} = h_{t}, A_{t} = a_{t}, R_{t+1}=1)$. Note that from definition of $m_s$ it immeidately follows $m_s = \int_{\mathcal{X}_{s} \times \mathcal{A}_{s}} m_{s+1}dQ_{s+1}(a_{s+1} \mid h_{s+1}, R_{s+1}=1) d\Pb(x_{s+1}|h_{s},a_{s}, R_{s+1}=1) $. 
		
		Now the efficient influence function for $\psi^*(\overline{Q}_t)=m_0$ is
		\begin{equation*} 
			\begin{aligned}
				& \sum_{s=0}^{t} \left\{  \int_{ \mathcal{A}_{s+1}} m_{s+1}(H_{s+1}, A_{s+1}, R_{s+2}=1)dQ_{s+1}(a_{s+1}|H_{s+1},R_{s+1}=1) - m_s(H_{s}, A_{s}, R_{s+1}=1) \right\}  \\
				& \qquad \times  \left(  \prod_{k=0}^{s} \frac{dQ_k(A_k \mid H_k, R_k=1)}{d\Pb(A_k \mid H_k, R_k=1)}   \frac{\mathbbm{1}\left(R_{k+1}=1\right)}{d\Pb(R_{k+1}=1 \mid H_k, A_k,R_k=1)} \right) 
			\end{aligned}
		\end{equation*} 
		where we define $dQ_{t+1} = 1$, $m_{t+1}(\cdot)=Y_t$, and $dQ_0(a_0|h_0)/d\Pb(a_0|h_0) = 1$.
	\end{lemma}

	\begin{lemma} \label{lem:eif_2}
		Suppose $\overline{Q}_T$ depends on $\Pb$ and let $\frac{\mathbbm{1}(H_t=h_t, R_t=1)}{d\Pb(h_t,R_t=1)} \phi_t(H_t, A_t, R_t=1;a_t)$ denote the efficient influence function for $dQ_t(a_t|h_t,R_t=1)$ defined in Lemma \ref{lem:eif_dQ} for all $t$. Then the efficient influence function for $\psi_t(\delta)$ is given as
		\begin{equation*} 
			\begin{aligned}
				&\varphi^*(\overline{Q}_t) \\
				& + \sum_{s=1}^{t}  \left(  \prod_{k=0}^{s-1} \frac{dQ_k(A_k \mid H_k, R_k=1)}{d\Pb(A_k \mid H_k, R_k=1)}  \frac{\mathbbm{1}(R_{k+1}=1)}{d\Pb(R_{k+1}=1 \mid H_k, A_k,R_k=1)}  \right) \\
				& \qquad \quad  \times \int_{ \mathcal{A}_{s}}  m_s(H_s, a_s,R_{s+1}=1) \phi_s(H_s, A_s, R_s=1;a_s) d\nu(a_s) 
			\end{aligned}
		\end{equation*} 
		where $\varphi^*(\overline{Q}_t)$ is the efficient influence function from Lemma \ref{lem:eif_1} and $\nu$ is a dominating measure for the distribution of $A_s$.
	\end{lemma}
	
	The proof of Lemma \ref{lem:eif_dQ}, \ref{lem:eif_1} and \ref{lem:eif_2} is basically a series of chain rules, after specifying efficient influence functions for terms that repeatedly appear. We provide a brief sketch for the proof of Lemma \ref{lem:eif_1} and \ref{lem:eif_2} below, which can be easily extendable to  the full proof. This also could be useful to develop other results for more general stochastic interventions.

	\subsubsection*{Proof of Lemma \ref{lem:eif_1} and Lemma \ref{lem:eif_2}}
	Let $\mathcal{IF}: \psi \rightarrow \phi$ denote a map to the efficient influence function $\phi$ for a functional $\psi$. First, without proof, we specify efficient influence functions for mean and conditional mean which serve two basic ingredients for our proof. For mean value of a random variable $Z$, we have
	$$
	\mathcal{IF}\big(\E[Z]\big) = Z - \E[Z],
	$$
	and for conditional mean with a pair of random variables $(X,Y) \sim \Pb$ when $X$ is discrete, we have
	$$
	\mathcal{IF}\big(\E[Y\vert X=x]\big) = \frac{\mathbbm{1}(X=x)}{\Pb(X=x)}\Big\{ Y - \E[Y \mid X=x] \Big\}.
	$$
	These results can be obtained by applying (\ref{def:influence-curve}) or (\ref{def:influence-curve-Gateaux-deriv}).

	\begin{proof} 
		It is sufficient to prove for $t=2$ since it is straightforward to extend the proof for arbitrary $t \leq T$ by induction. For $t=2$, it is enough to compute the following four terms. 
		\begin{itemize} 
			\item[A)]
			$
			\begin{aligned}[t]
			& \int_{ \mathcal{H}_2\times \mathcal{A}_2} \mathcal{IF}  \Big( \mu(h_2,a_2, R_3=1) \Big)  \prod_{s=1}^{2} dQ_s(a_s \mid h_s, R_s=1) d\Pb(y_{s-1},x_s|h_{s-1},a_{s-1}, R_s=1) \\
			&= \int_{ \mathcal{H}_2\times \mathcal{A}_2} \frac{\mathbbm{1}\{ (H_2,A_2,R_3)=(h_2,a_2,1) \} }{d\Pb(h_2,a_2,R_3=1)}\Big\{Y -  \mu(h_2,a_2, R_3=1)  \Big\}  \\
			& \qquad \qquad \ \times \prod_{s=1}^{2} dQ_s(a_s \mid h_s, R_s=1) d\Pb(y_{s-1},x_s|h_{s-1},a_{s-1}, R_s=1)\\
			&= \int_{ \mathcal{H}_2\times \mathcal{A}_2}  \mathbbm{1}\big\{ (H_2,A_2,R_3)=(h_2,a_2,1) \big\} \big\{Y -  \mu(h_2,a_2, R_3=1)  \big\} \\
			&  \qquad \qquad \ \times \prod_{s=1}^{2} \frac{dQ_s(a_s \mid h_s, R_s=1)}{d\Pb(a_s \mid h_s, R_s=1)} \frac{1}{d\Pb(R_{s+1}=1 \mid h_s, a_s,R_s=1)} \\
			&= \{ Y - \mu(H_2,A_2, R_3=1) \} \mathbbm{1}(R_3=1) \prod_{s=1}^{2} \frac{dQ_t(A_s \mid H_s, R_s=1)}{d\Pb(A_s \mid H_s, R_s=1)} \frac{1}{d\Pb(R_{s+1}=1 \mid H_s, A_s,R_s=1)} 
			\end{aligned}
			$
			\item[B)] 
			$
			\begin{aligned}[t]
			& \int_{ \mathcal{H}_2\times \mathcal{A}_2} \mu(h_2,a_2, R_3=1)    \mathcal{IF}  \Big( d\Pb(y_1,x_2 | h_1, a_1, R_2=1) \Big) d\Pb(h_1)  \prod_{s=1}^{2} dQ_s(a_s \mid h_s, R_s=1) \\
			& = \int_{ \mathcal{H}_2\times \mathcal{A}_2} \mu(h_2,a_2, R_3=1)   \frac{\mathbbm{1}\big\{ (H_1,A_1,R_2)=(h_1,a_1,1) \big\}}{d\Pb(h_1, a_1, R_2=1)}  \\
			& \qquad \qquad \ \times \Big\{\mathbbm{1}(Y_1=y_1,X_2=x_2) - d\Pb(y_1,x_2 | h_1, a_1, R_2=1) \Big\} d\Pb(h_1)  \prod_{s=1}^{2} dQ_s(a_s \mid h_s, R_s=1) \\
			& = \int_{ \mathcal{H}_2\times \mathcal{A}_2} \mu(h_2,a_2, R_3=1)   \\
			& \qquad \qquad \ \times \frac{\mathbbm{1}\big\{ (H_1,A_1,R_2)=(h_1,a_1,1) \big\} \big\{\mathbbm{1}(Y_1=y_1,X_2=x_2) - d\Pb(y_1,x_2 | h_1, a_1, R_2=1) \big\} }{d\Pb(R_2=1|h_1,a_1)d\Pb(a_1|h_1)d\Pb(h_1)}  \\		
			& \qquad \qquad \ \times d\Pb(h_1)  \prod_{s=1}^{2} dQ_s(a_s \mid h_s, R_s=1) \\
			&=  \int_{ \mathcal{H}_2\times \mathcal{A}_2} \mu(h_2,a_2, R_3=1)  dQ_2(a_2 \mid h_2, R_2=1) \mathbbm{1}\big\{ (H_1,A_1,R_2)=(h_1,a_1,1) \big \} \\
			& \qquad \qquad \ \times \big\{\mathbbm{1}(Y_1=y_1,X_2=x_2) - d\Pb(y_1,x_2 | h_1, a_1, R_2=1) \big\}  \frac{dQ_1(A_1 \mid H_1)}{d\Pb(A_1 \mid H_1)} \frac{1}{d\Pb(R_{2}=1 \mid H_1, A_1)} \\
			&= \Bigg\{  \int_{ \mathcal{H}_2\times \mathcal{A}_2 \setminus \mathcal{H}_2 } \mu(H_2,a_2, R_3=1)  dQ_2(a_2 \mid H_2, R_2=1) \\
			& \qquad	-  \int_{ \mathcal{H}_2\times \mathcal{A}_2 \setminus \mathcal{H}_1 \times \mathcal{A}_1 } \mu(h_2,a_2, R_3=1)  dQ_2(a_2 \mid h_2, R_2=1) d\Pb(y_1,x_2 | h_1, a_1, R_2=1) \Bigg\}   \\
			& \qquad \times  \mathbbm{1}(R_2=1)  \frac{dQ_1(A_1 \mid H_1)}{d\Pb(A_1 \mid H_1)} \frac{1}{d\Pb(R_{2}=1 \mid H_1, A_1)} \\
			&= \Bigg\{  \int_{ \mathcal{A}_2 } \mu(H_2,a_2, R_3=1)  dQ_2(a_2 \mid H_2, R_2=1) 	-  m_1(h_1,a_1,R_2=1) \Bigg\}   \\
			& \qquad \times  \mathbbm{1}(R_2=1)  \frac{dQ_1(A_1 \mid H_1)}{d\Pb(A_1 \mid H_1)} \frac{1}{d\Pb(R_{2}=1 \mid H_1, A_1)} \\
			\end{aligned}
			$
			\item[C)]
			$
			\begin{aligned}[t]
			& \int_{ \mathcal{H}_2\times \mathcal{A}_2} \mu(h_2,a_2, R_3=1)    d\Pb(y_1,x_2 | h_1, a_1, R_2=1)  \mathcal{IF}  \Big(d\Pb(h_1) \Big)  \prod_{s=1}^{2} dQ_s(a_s \mid h_s, R_s=1) \\
			&= \int_{ \mathcal{H}_2\times \mathcal{A}_2} \mu(h_2,a_2, R_3=1)    d\Pb(y_1,x_2 | h_1, a_1, R_2=1)  \big\{\mathbbm{1}(X_1=x_1) - d\Pb(x_1) \big\}   \prod_{s=1}^{2} dQ_s(a_s \mid h_s, R_s=1) \\
			&=  \int_{ \mathcal{H}_2\times \mathcal{A}_2 \setminus \mathcal{H}_1 } \mu(h_2,a_2, R_3=1)  dQ_2(a_2 \mid h_2, R_2=1) d\Pb(y_1,x_2 | h_1, a_1, R_2=1) dQ_1(a_1|h_1) - m_0 \\
			& =\int_{ \mathcal{A}_1} m_1(h_1,a_1,R_2=1) dQ_1(a_1|h_1) - m_0 \\ 
			\end{aligned}
			$
			\item[D)] Let $\frac{\mathbbm{1}(H_t=h_t, R_t=1)}{d\Pb(h_t,R_t=1)} \phi_t(H_t, A_t, R_t=1;a_t)$ denote the efficient influence function for $dQ_t(a_t|h_t,R_t=1)$ as given in Lemma \ref{lem:eif_dQ}. Then we have \\
			$
			\begin{aligned}
			& \int_{ \mathcal{H}_2\times \mathcal{A}_2} \mu(h_2,a_2, R_3=1)  d\Pb(h_1)  d\Pb(y_1,x_2 | h_1, a_1, R_2=1)  \mathcal{IF}  \Big( dQ_1(a_1|h_1) dQ_2(a_2 \mid h_2, R_2=1) \Big)   \\
			&= \int_{ \mathcal{H}_2\times \mathcal{A}_2} \mu(h_2,a_2, R_3=1)  d\Pb(h_1)  d\Pb(y_1,x_2 | h_1, a_1, R_2=1)    \frac{\mathbbm{1}\big\{ (H_2,R_2)=(h_2,1) \big\}}{d\Pb(h_2, R_2=1)} \phi_2  dQ_1(a_1|h_1)  \\
			& \quad + \int_{ \mathcal{H}_2\times \mathcal{A}_2} \mu(h_2,a_2, R_3=1)  d\Pb(h_1)  d\Pb(y_1,x_2 | h_1, a_1, R_2=1)    \frac{\mathbbm{1}\big\{ (H_1=h_1) \big\}}{d\Pb(h_1)} \phi_1 dQ_2(a_2 \mid h_2, R_2=1)\\
			&= \int_{ \mathcal{H}_2\times \mathcal{A}_2} \mu(h_2,a_2, R_3=1)   \frac{\mathbbm{1}\big\{ (H_2,R_2)=(h_2,1) \big\}d\Pb(h_1)  d\Pb(y_1,x_2 | h_1, a_1, R_2=1) dQ_1(a_1|h_1) }{d\Pb(y_1,x_2|h_1,a_1,R_2=1)d\Pb(R_2=1|h_1,a_1)d\Pb(a_1|h_1)d\Pb(h_1)} \phi_2    \\
			& \quad + \int_{ \mathcal{H}_2\times \mathcal{A}_2} \mu(h_2,a_2, R_3=1)    d\Pb(y_1,x_2 | h_1, a_1, R_2=1)    \mathbbm{1}\big\{ (H_1=h_1) \big\} \phi_1 dQ_2(a_2 \mid h_2, R_2=1)\\
			&= \int_{ \mathcal{H}_2\times \mathcal{A}_2 \setminus \mathcal{H}_2} \mu(H_2,a_2, R_3=1)   \mathbbm{1}(R_2=1) \phi_2   \frac{dQ_1(A_1 \mid H_1)}{d\Pb(A_1 \mid H_1)} \frac{1}{d\Pb(R_{2}=1 \mid H_1, A_1)}  \\
			& \quad + \int_{ \mathcal{H}_2\times \mathcal{A}_2 \setminus \mathcal{H}_1} \mu(h_2,a_2, R_3=1)  dQ_2(a_2 \mid h_2, R_2=1)  d\Pb(y_1,x_2 | h_1, a_1, R_2=1)  \phi_1 \\
			&= \left\{  \frac{dQ_1(A_1 \mid H_1)}{d\Pb(A_1 \mid H_1)} \frac{1}{d\Pb(R_{2}=1 \mid H_1, A_1)} \right\} \int_{ \mathcal{A}_2 } \mu(H_2,a_2, R_3=1) \phi_2  d\nu(a_2)  \mathbbm{1}(R_2=1)     \\
			& \quad + \int_{ \mathcal{A}_1 } m_1(h_1,a_1,R_2=1) \phi_1 d\nu(a_1)
			\end{aligned}			
			$
		\end{itemize}	
		Note that we have set $dQ_0(a_0|h_0)/d\Pb(a_0|h_0) = 1$, and that we have $d\Pb(R_1=1)=1$ and $ \mathbbm{1}(R_1=1)=1 $ by construction. Hence, putting part A), B), and C) together proves Lemma \ref{lem:eif_1} and part D) proves Lemma \ref{lem:eif_2}. 
		
		Note that to formally verify that the given expressions in Lemmas \ref{lem:eif_1} and \ref{lem:eif_2} are the efficient influence functions, we would need to check if the pathwise differentiability formula \eqref{def:influence-curve} holds. This essentially follows if the remainder terms are second-order, which will be verified in Lemmas \eqref{lem:remainder_1} and \eqref{lem:remainder_2} later.
	\end{proof}

	Finally, we are ready to give a proof of Theorem \ref{thm:eif}. In fact, it is nothing but rearranging terms in the given efficient influence function.
	
	\vspace*{0.15in}
	\subsubsection{Proof of Theorem \ref{thm:eif}}
	\begin{proof}
		First, we define following shorthand notations for the proof: for $\forall s \leq t$
		$$
		dQ_{s}(A_{s}) \equiv dQ_{s}(A_{s}|H_{s},R_{s}=1),  \qquad d\Pb_{s}(A_{s}) \equiv d\Pb(A_{s} \mid H_{s}, R_{s}=1),
		$$
		$$
		d\omega_{s} \equiv \omega_{s}( H_{s}, A_{s}) \equiv d\Pb(R_{{s}+1}=1 \mid H_{s}, A_{s},R_{s}=1),    
		$$	
		$$
		m_{s}(H_{s}, a_{s}) \equiv m_{s}(H_{s}, a_{s}, R_{{s}+1}=1)
		$$
		With these notations we can rewrite the result of Lemma \ref{lem:eif_1} as below.
		\begin{align*}
			& \sum_{s=0}^{t} \left\{  \int_{ \mathcal{A}_{{s}+1}} m_{{s}+1}(H_{{s}+1}, a_{{s}+1})dQ_{{s}+1}(a_{{s}+1}) - m_{s}(H_{s}, A_{s}) \right\} \mathbbm{1}\left(R_{{s}+1}=1\right)  \left(  \prod_{k=0}^{s} \frac{dQ_k(A_k )}{d\Pb_k(A_k)}   \frac{1}{d\omega_k} \right) \\
			& = \sum_{{s}=1}^{t} \left\{  \int_{ \mathcal{A}_{s}} m_{s}(H_{s}, a_{s})dQ_{s}(a_{s}) - m_{s}(H_{s}, A_{s}) \left[ \mathbbm{1}\left(R_{{s}+1}=1\right) \frac{dQ_{s}(A_{s})}{d\Pb_{s}(A_{s})}  \frac{1}{d\omega_{s}} \right] \right\}  \\
			& \qquad \times \mathbbm{1}\left(R_{{s}}=1\right)  \left(  \prod_{k=0}^{s-1} \frac{dQ_k(A_k )}{d\Pb_k(A_k)}   \frac{1}{d\omega_k} \right) + \mathbbm{1}\left(R_{{t}+1}=1\right)  \left(  \prod_{s=1}^{t} \frac{dQ_s(A_s )}{d\Pb_s(A_s)}   \frac{1}{d\omega_s} \right)Y_t - m_0.
		\end{align*}
		Now, by the result of Lemma \ref{lem:eif_1} and \ref{lem:eif_2}, we can represent the efficient influence function for $\psi_t(\delta)$ as
		\begin{equation*} 
			\begin{aligned}	
				& \sum_{{s}=1}^{t} \Bigg\{  \int_{ \mathcal{A}_{s}} m_{s}(H_{s}, a_{s})dQ_{s}(a_{s}) - m_{s}(H_{s}, A_{s}) \left[ \mathbbm{1}\left(R_{{s}+1}=1\right) \frac{dQ_{s}(A_{s})}{d\Pb_{s}(A_{s})}  \frac{1}{d\omega_{s}} \right] \\
				& \ \ \qquad +  \int_{ \mathcal{A}_{s}}  m_{s}(H_{s}, a_{s}) \phi_{s}(H_{s}, A_{s}, R_{s}=1;a_{s}) d\nu(a_{s})  \Bigg\}  
				\mathbbm{1}\left(R_{s}=1\right)  \left(  \prod_{k=0}^{s-1} \frac{dQ_k(A_k )}{d\Pb_k(A_k)}   \frac{1}{d\omega_k} \right) \\
				& \qquad + \mathbbm{1}\left(R_{{t}+1}=1\right)  \left(  \prod_{s=1}^{t} \frac{dQ_s(A_s )}{d\Pb_s(A_s)}   \frac{1}{d\omega_s} \right)Y_t - m_0.
			\end{aligned}
		\end{equation*} 
		On the other hand, we have
		\begin{align*}
			\int_{ \mathcal{A}_{s}} m_{s}(H_{s}, a_{s})dQ_{s}(a_{s}) &= \frac{m_{s}(H_{s},1)\delta\pi_{s}(H_{s})+m_{s}(H_{s},0)\{ 1-\pi_{s}(H_{s}) \}}{\delta\pi_{s}(H_{s}) + 1-\pi_{s}(H_{s})},
		\end{align*}
		\begin{align*}
			\frac{dQ_s(A_s )}{d\Pb_s(A_s )} &= \frac{\delta A_s + 1-A_s}{\delta\pi_s(H_s) + 1-\pi_s(H_s)},
		\end{align*}
		\begin{align*}
			m_{s}(H_{s}, A_{s})\frac{dQ_{s}(A_{s})}{d\Pb_{s}(A_{s})} &= \frac{m_{s}(H_{s}, 1, R_{{s}+1}=1)\delta A_{s} + m_{s}(H_{s}, 0, R_{{s}+1}=1)(1-A_{s})}{\delta\pi_{s}(H_{s}) + 1-\pi_{s}(H_{s})},
		\end{align*}	
		\begin{align*}
			\int_{ \mathcal{A}_{s}}   m_{s}(H_{s}, a_{s})&  \phi_{s}(H_{s},A_{s},R_{s}=1;a_{s}) d\nu(a_{s})
			= \frac{\{ m_{s}(H_{s}, 1) - m_{s}(H_{s}, 0)\}\delta(A_{s} - \pi_{s}(H_{s})) }{\left(\delta\pi_{s}(H_{s}) + 1 - \pi_{s}(H_{s}) \right)^2 },
		\end{align*}
		
		which leads to
		\begin{align*}
			& \int_{ \mathcal{A}_{s}} m_{s}(H_{s}, a_{s})dQ_{s}(a_{s}) - m_{s}(H_{s}, A_{s}) \left[ \mathbbm{1}\left(R_{{s}+1}=1\right) \frac{dQ_{s}(A_{s})}{d\Pb_{s}(A_{s})}  \frac{1}{d\omega_{s}} \right] \\
			& \quad +  \int_{ \mathcal{A}_{s}}  \mathbbm{1}\left(R_{{s}+1}=1\right) m_{s}(H_{s}, a_{s}) \phi_{s}(H_{s}, A_{s}, R_{s}^\prime=1;a_{s}) d\nu(a_{s})  \\
			&= \frac{1}{\left\{ \delta\pi_{s}(H_{s}) + 1 - \pi_{s}(H_{s}) \right\} \omega_s( H_s, A_s)} \\
			& \quad \times \Bigg[ \left\{m_{s}(H_{s}, 1) -m_{s}(H_{s}, 0)\right\}\delta(A- \pi_{s}(H_{s})) \frac{\omega_s( H_s, A_s)}{\delta\pi_{s}(H_{s}) + 1 - \pi_{s}(H_{s}) } \\
			& \quad \qquad + \delta m_{s}(H_{s}, 1)\left\{\pi_{s}(H_{s})\omega_s( H_s, A_s)-A_sR_{s+1}\right\} \\
			& \quad \qquad + m_{s}(H_{s}, 0)\left\{(1-\pi_{s}(H_{s}))\omega_s( H_s, A_s) - (1-A_s)R_{s+1} \right\} \Bigg]\\
		\end{align*}
		After some rearrangement, we finally obtain an equivalent form of the efficient influence function for $\psi_t(\delta)$ by
		\begin{equation*} 
			\begin{aligned}	
				& \sum_{s=1}^{t} \left\{ \frac{1}{\delta A_s + 1-A_s}  \right\} \Bigg[ \frac{\left\{m_{s}(H_{s}, 1) -m_{s}(H_{s}, 0)\right\}\delta(A- \pi_{s}(H_{s}))\omega_s( H_s, A_s)}{\delta\pi_{s}(H_{s}) + 1 - \pi_{s}(H_{s})} \\ 
				& \qquad + \begin{pmatrix}
                \delta m_{s}(H_{s}, 1)\left\{\pi_{s}(H_{s})\omega_s( H_s, A_s)-A_sR_{s+1}\right\}  \\
                + m_{s}(H_{s}, 0)\left\{(1-\pi_{s}(H_{s}))\omega_s( H_s, A_s) - (1-A_s)R_{s+1} \right\}
                \end{pmatrix} \Bigg]  \prod_{k=1}^{s} \left\{ \frac{\delta A_k + 1-A_k}{\delta\pi_k(H_k) + 1-\pi_k(H_k)} \cdot\frac{R_{k}}{\omega_{k}( H_{k}, A_{k})} \right\} \\
				& + \prod_{s=1}^{t} \left\{ \frac{\delta A_s + 1-A_s}{\delta\pi_s(H_s) + 1-\pi_s(H_s)}\cdot\frac{R_{s+1}}{\omega_s( H_s, A_s)}Y_t \right\} - \psi_t(\delta).
			\end{aligned}
		\end{equation*} 
		Note that we use convention that $dQ_0=d\Pb_0=d\omega_0=1$ and $R_1=1$.
	\end{proof}

	\subsection{Proof of Theorem \ref{thm:inf-time-horizon}}
	\label{proof:thm-inf-time}	
	
	Let $\widehat{\psi}_{c.ipw}(\overline{a^\prime}_{T})$ denote the standard IPW estimator of a  classical deterministic intervention effect $\E \left[ Y^{\overline{a^\prime}_{T}} \right]$ under $i.i.d$ assumption, i.e.
	$$
	\widehat{\psi}_{c.ipw}(\overline{a^\prime}_{T}) = \prod_{t=1}^{T} \left( \frac{\mathbbm{1}\left({A}_{t}={a^\prime}_{t}\right)}{ \pi_t(a^\prime_t|H_t)} \right)Y.
	$$
	Hence $\widehat{\psi}_{c.ipw}(\overline{\bm{1}})$ is equivalent to $\widehat{\psi}_{at}$ in the main text. Now by definition we have
	\begin{align*}
		Var\left(\widehat{\psi}_{c.ipw}(\overline{a^\prime}_{T})\right) &= \E\left\{ \left( \prod_{t=1}^{T} \frac{\mathbbm{1}\left({A}_{t}={a^\prime}_{t}\right)}{ \pi_t(a^\prime_t|H_t)^2} \right)Y^2 \right\} -  \left\{ \E \left[ \prod_{t=1}^{T} \frac{\mathbbm{1}\left({A}_{t}={a^\prime}_{t}\right)}{ \pi_t(a^\prime_t|H_t)} Y  \right]   \right\}^2 \\
		& \equiv \mathbb{V}_{c.ipw.1}(\overline{a^\prime}_{T}) - \mathbb{V}_{c.ipw.2}(\overline{a^\prime}_{T})
	\end{align*}
	where $\mathbb{V}_{c.ipw.1}(\overline{a^\prime}_{T})$ and $\mathbb{V}_{c.ipw.2}(\overline{a^\prime}_{T})$ are simply the first and second term in the first line of the expansion respectively.
	
	By the same procedure to derive g-formula \citep{robins1986} it is easy to see
	\begin{align*}
		\mathbb{V}_{c.ipw.1}(\overline{a^\prime}_{T}) &= \E\left\{ \prod_{t=1}^{T} \left( \frac{\mathbbm{1}\left({A}_{t}={a^\prime}_{t}\right)}{ \pi_t(a^\prime_t|H_t)^2} \right)Y^2 \right\} \\
		&= \int_{\mathcal{X}} \E\left[Y^2 \mid \overline{X}_{t}, \overline{A}_{t}=\overline{a^\prime}_{t} \right] \prod_{t=1}^{T} \frac{d\Pb(X_t \mid \overline{X}_{t-1}, \overline{A}_{t-1}=\overline{a^\prime}_{t-1})}{\pi_t(a^\prime_t|H_t)} 
	\end{align*}
	where $\mathcal{X}=\mathcal{X}_1 \times \cdots \times \mathcal{X}_{T}$. Above result simply follows by iterative expectation conditioning on $\overline{X}_{t}$ and then another iterative expectation conditioning on $H_t$ followed by the fact that $\E \left[ \frac{\mathbbm{1}\left({A}_{t}={a^\prime}_{t}\right)}{\pi_t(a^\prime_t|H_t)} \big\vert H_t \right]=1$ for all $t$. We repeat this process $T$ times, starting from $t=T$ all the way through $t=1$. 
	
	Likewise, for $\widehat{\psi}_{inc}$ we have 
	\begin{align*}
		Var(\widehat{\psi}_{inc}) &= \E\left\{ \prod_{t=1}^{T}  \left( \frac{\delta A_t + 1-A_t }{\delta{\pi}_t(H_t) + 1-{\pi}_t(H_t)} \right)^2Y^2  \right\} - \left\{  \E  \left[ \prod_{t=1}^{T}  \left( \frac{\delta A_t + 1-A_t }{\delta{\pi}_t(H_t) + 1-{\pi}_t(H_t)} \right)Y \right]  \right\}^2 \\
		& \equiv \mathbb{V}_{inc.1} - \mathbb{V}_{inc.2}
	\end{align*}
	For the first term $\mathbb{V}_{inc.1}$, observe that
	\begin{align*}
		&	\E\left\{ \prod_{t=1}^{T}  \left( \frac{\delta A_t + 1-A_t }{\delta{\pi}_t(H_t) + 1-{\pi}_t(H_t)} \right)^2Y^2  \right\} \\
		&= \E\left\{ \prod_{t=1}^{T-1}  \left( \frac{\delta A_t + 1-A_t }{\delta{\pi}_t(H_t) + 1-{\pi}_t(H_t)} \right)^2  \E\left[ \left( \frac{\delta A_{T} + 1-A_{T} }{\delta{\pi}_{T}(H_{T}) + 1-{\pi}_{T}(H_{T})} \right)^2Y^2 \Bigg\vert H_{T} \right] \right\} \\
		&= \E\left\{ \prod_{t=1}^{T-1}  \left( \frac{\delta A_t + 1-A_t }{\delta{\pi}_t(H_t) + 1-{\pi}_t(H_t)} \right)^2  \E\left[ \frac{\delta^2Y^2 }{(\delta{\pi}_{T}(H_{T}) + 1-{\pi}_{T}(H_{T}))^2} \Bigg\vert H_{T}, A_{T}=1 \right]{\pi}_{T}(H_{T}) \right\} \\
		& \quad \ + \E\left\{ \prod_{t=1}^{T-1}  \left( \frac{\delta A_t + 1-A_t }{\delta{\pi}_t(H_t) + 1-{\pi}_t(H_t)} \right)^2  \E\left[ \frac{Y^2}{(\delta{\pi}_{T} + 1-{\pi}_{T})^2}  \Bigg\vert H_{T}, A_{T}=0 \right]\left(1-{\pi}_{T}(H_{T})\right) \right\}
	\end{align*}
	where we apply the law of total expectation in the first equality and the law of total probability in the second.
	
	After repeating the same process for $T-1$ times, for $t=T-1, ... ,1$, we obtain $2^{T}$ terms in the end where each of which corresponds to the distinct treatment sequences $\overline{A}_{T}=\overline{a}_{T}$. Hence, we eventually have
	\begin{align*}
		\mathbb{V}_{inc.1}  
		&= \sum_{\overline{a}_{T} \in \overline{\mathcal{A}}_{T}} \int_{\mathcal{X}} \E\left[Y^2 \mid H_{T}, A_{T}=a_{T} \right] \prod_{t=1}^{T} \frac{\mathbbm{1}\left({a}_{t}=1\right)\delta^2{\pi}_t(H_t) + \mathbbm{1}\left({a}_{t}=0\right)\{1-{\pi}_t(H_t)\}}{(\delta{\pi}_{t}(H_t) + 1-{\pi}_{t}(H_t))^2} \\
		& \qquad \qquad \quad \times d\Pb(X_t \mid \overline{X}_{t-1}, \overline{A}_{t-1}=\overline{a}_{t-1}). 
	\end{align*}
	
	Recall that we assume $\pi_t(H_t)=p$ for all $t$ as stated in Theorem \ref{thm:inf-time-horizon}. Hence we can write $\pi_t(a_t \mid H_t)$ as $\pi_t(a_t)=  \mathbbm{1}\left({a}_{t}=1\right)p + \mathbbm{1}\left({a}_{t}=0\right)\{1-p\}$.

	We want to find an upper bound of the variance ratio $\text{VR}(\widehat{\psi}_{c.ipw}(\overline{a}_{T}), \widehat{\psi}_{inc}) \coloneqq \frac{ \mathbb{V}_{inc.1} - \mathbb{V}_{inc.2}}{\mathbb{V}_{c.ipw.1}(\overline{a}_{T}) - \mathbb{V}_{c.ipw.2}(\overline{a}_{T})} $ for always-treated unit (i.e., $\overline{a}_{T}=\overline{\bm{1}}$). This can be done by computing the quantity
	$$
	\frac{ \mathbb{V}_{inc.1}}{\mathbb{V}_{c.ipw.1}(\overline{\bm{1}}) - \mathbb{V}_{c.ipw.2}(\overline{\bm{1}})} 
	$$
	since $0 < \mathbb{V}_{inc.2} < \mathbb{V}_{inc.1}$ by Jensen's inequality. 
	
	Note that we have
	\begin{align*}
		\mathbb{V}_{c.ipw.1}(\overline{\bm{1}}) - \mathbb{V}_{c.ipw.2}(\overline{\bm{1}})&= \int_{\mathcal{X}} \E\left[Y^2 \mid \overline{X}_{T}, \overline{A}_{T}=\overline{a^\prime}_{T} \right] \prod_{t=1}^{T} \frac{d\Pb(X_t \mid \overline{X}_{t-1}, \overline{A}_{t-1}=\overline{a^\prime}_{t-1})}{p} - \left(\E[Y^{\overline{\bm{1}}}]\right)^2 \\
		& =  \left( \frac{1}{p} \right)^{T} \E\left[\left(Y^{\overline{\bm{1}}} \right)^2 \right]    - \left(\E\left[Y^{\overline{\bm{1}}}\right]\right)^2
	\end{align*}
	, and under the given boundedness assumption we see the ratio of the second term to the first term becomes quickly (at least exponentially) negligible as $t$ increases. Hence we can write
	\begin{align*}
		\frac{1}{\mathbb{V}_{c.ipw.1}(\overline{\bm{1}}) - \mathbb{V}_{c.ipw.2}(\overline{\bm{1}})} \leq \frac{1}{\mathbb{V}_{c.ipw.1}(\overline{\bm{1}})}\left( 1+ \frac{c \left(\E\left[Y^{\overline{\bm{1}}}\right]\right)^2}{\left( 1/p \right)^{T} \E\left[\left(Y^{\overline{\bm{1}}} \right)^2 \right]} \right)
	\end{align*}
	for some constant $c$ such that $\frac{1}{1-\mathbb{V}_{c.ipw.2}(\overline{\bm{1}})/\mathbb{V}_{c.ipw.1}(\overline{\bm{1}})}=\frac{1}{1-p^T{\left(\E\left[Y^{\overline{\bm{1}}}\right]\right)^2}\big/{\E\left[\left(Y^{\overline{\bm{1}}} \right)^2 \right]}} \leq {c}$. Note that in our setting in which we have an infinitely large value of $T$, $c$ can be almost any constant greater than one.
	
	Putting above ingredients together, for sufficiently large $t$ it follows that
	\begin{align*}
		\text{VR}(\widehat{\psi}_{c.ipw}(\overline{\bm{1}}), \widehat{\psi}_{inc})  & \leq \frac{\mathbb{V}_{inc.1}}{\mathbb{V}_{c.ipw.1}(\overline{\bm{1}})} \left( 1+ \frac{c \left(\E\left[Y^{\overline{\bm{1}}}\right]\right)^2}{\left( 1/p \right)^{T} \E\left[\left(Y^{\overline{\bm{1}}} \right)^2 \right]} \right),
	\end{align*}
	where we have
	\begin{align*}
		\frac{\mathbb{V}_{inc.1}}{\mathbb{V}_{c.ipw.1}(\overline{\bm{1}})} &= \frac{ w(\overline{\bm{1}}) \mathbb{V}_{c.ipw.1}(\overline{\bm{1}} ) + \sum_{\overline{a}_{T} \neq \overline{\bm{1}}} w(\overline{a}_{T}; \delta, p) \mathbb{V}_{c.ipw.1}(\overline{a}_{T} )}{\mathbb{V}_{c.ipw.1}(\overline{\bm{1}})} \\
		&= w(\overline{\bm{1}})  + \sum_{\overline{a}_{T} \neq \overline{\bm{1}}}  w(\overline{a}_{T}; \delta, p) \prod_{t=1}^{T} \left(\frac{p}{\pi_t(a_t)} \frac{\E\left[\left(Y^2 \right)^{\overline{a}_{T}} \right]}{\E\left[\left(Y^{\overline{\bm{1}}} \right)^2 \right]} \right) \\
		&\leq \frac{b_u^2}{\E\left[\left(Y^{\overline{\bm{1}}} \right)^2 \right]} \left\{ w(\overline{\bm{1}})  +  \sum_{\overline{a}_{T} \neq \overline{\bm{1}}} \left[ \prod_{t=1}^{T} \frac{\mathbbm{1}\left({a}_{t}=1\right)\delta^2 p^2 + \mathbbm{1}\left({a}_{t}=0\right)(1- p)p}{(\delta p + 1-p)^2} \right]  \right\} \\
		&= \frac{b_u^2}{\E\left[\left(Y^{\overline{\bm{1}}} \right)^2 \right]}\left\{ \frac{\delta^2p^2 + p(1-p)}{(\delta p + 1 - p)^2} \right\}^{T}
	\end{align*}
	where the first equality follows by the fact that $ \mathbb{V}_{inc.1} = \sum_{\overline{a}_{T} \in \overline{\mathcal{A}}_{T}} w(\overline{a}_{T}; \delta, p) \mathbb{V}_{c.ipw.1}(\overline{a}_{T} )$ derived in the proof of the first part, the second equality by the fact that $\mathbb{V}_{c.ipw.1}(\overline{a}_{T} ) = \prod_{t=1}^{T} \frac{1}{\pi_t(a_t)}  \E\left[\left(Y^2 \right)^{\overline{a}_{T}} \right] $, the first inequality by definition of $w(\overline{a}_{T}; \delta, p)$ and the given boundedness assumption, and the last equality by binomial theorem. Therefore we obtain the upper bound as
	$$
	\text{VR}(\widehat{\psi}_{c.ipw}(\overline{\bm{1}}), \widehat{\psi}_{inc})  \leq \frac{b_u^2}{\E\left[\left(Y^{\overline{\bm{1}}} \right)^2 \right]}\left\{ \frac{\delta^2p^2 + p(1-p)}{(\delta p + 1 - p)^2} \right\}^{T} \left( 1+ \frac{c \left(\E\left[Y^{\overline{\bm{1}}}\right]\right)^2}{\left( 1/p \right)^{T} \E\left[\left(Y^{\overline{\bm{1}}} \right)^2 \right]} \right).
	$$
	
	Next for the lower bound, first we note that
	\begin{align*}
		\mathbb{V}_{inc.2} &= \left\{  \E  \left[ \prod_{t=1}^{T}  \left( \frac{\delta A_t + 1-A_t }{\delta p + 1-p} \right)Y \right]  \right\}^2  \\
		&= \Bigg\{ \sum_{\overline{a}_{T} \in \overline{\mathcal{A}}_{T}} \int_{\mathcal{X}} \E\left[Y \mid H_{T}, A_{T}=a_{T} \right] \left(\prod_{t=1}^{T} \frac{\mathbbm{1}\left({a}_{t}=1\right)\delta p + \mathbbm{1}\left({a}_{t}=0\right)(1-p)}{\delta p + 1-p}\right)  \\
		& \qquad \qquad \qquad \times d\Pb(X_t \mid \overline{X}_{t-1}, \overline{A}_{t-1}=\overline{a}_{t-1}) \Bigg\}^2 \\
		& \leq b_u^2 \left[ \sum_{\overline{a}_{T} \in \overline{\mathcal{A}}_{T}}\prod_{t=1}^{T} \left( \frac{\mathbbm{1}\left({a}_{t}=1\right)\delta p + \mathbbm{1}\left({a}_{t}=0\right)(1-p)}{\delta p + 1-p}\right)  \right]^2 \\
		&= b_u^2 \left( \frac{\delta p + 1-p}{\delta p + 1-p} \right)^{2T} = b_u^2
	\end{align*}
	where the first equality follows by definition, the second equality by exactly same process used to find the expression for $\mathbb{V}_{inc.1}$, the first inequality by the boundedness assumption, and the third equality by binomial theorem.
	
	However, we already know that
	\begin{align*}
		\mathbb{V}_{c.ipw.1}(\overline{\bm{1}}) - \mathbb{V}_{c.ipw.2}(\overline{\bm{1}})\leq  \mathbb{V}_{c.ipw.1}(\overline{\bm{1}}) = \left( \frac{1}{p} \right)^{T} \E\left[\left(Y^{\overline{\bm{1}}} \right)^2 \right].
	\end{align*}
	
	Hence putting these together we conclude
	\begin{align*}
		\text{VR}(\widehat{\psi}_{c.ipw}(\overline{\bm{1}}), \widehat{\psi}_{inc}) &= \frac{ \mathbb{V}_{inc.1} - \mathbb{V}_{inc.2}}{\mathbb{V}_{c.ipw.1}(\overline{\bm{1}}) - \mathbb{V}_{c.ipw.2}(\overline{\bm{1}})} \\
		& \geq \frac{ \mathbb{V}_{inc.1} - b_u^2}{\mathbb{V}_{c.ipw.1}(\overline{\bm{1}})} \\
		&=  \frac{b_u^2}{\E\left[\left(Y^{\overline{\bm{1}}} \right)^2 \right]}\left\{ \frac{\delta^2p^2 + p(1-p)}{(\delta p + 1 - p)^2} \right\}^{T} - \frac{b_u^2}{\E\left[\left(Y^{\overline{\bm{1}}} \right)^2 \right]} p^{T}.
	\end{align*}
	
	At this point, we obtain upper and lower bound for $\text{VR}(\widehat{\psi}_{c.ipw}(\overline{\bm{1}}), \widehat{\psi}_{inc})$, which yields the result of part $ii)$ having $C_{T} = \frac{b_u^2}{\E\left[\left(Y^{\overline{\bm{1}}} \right)^2 \right]}$.
	
	Proof for the case of $\overline{a^\prime}_{T}=\overline{\bm{0}}$ (\textit{never-treated unit}) is based on the almost same steps as the case of $\overline{a^\prime}_{T}=\overline{\bm{1}}$ except for the rearragement of terms due to replacing $\left(\frac{1}{p}\right)^{T}$ by $\left(\frac{1}{1-p}\right)^{T}$ and so on. In fact, due to the generality of our proof structure, the exact same logic used for $\widehat{\psi}_{c.ipw}(\overline{\bm{1}})$ also applies to $\widehat{\psi}_{c.ipw}(\overline{\bm{0}})$ (and $\widehat{\psi}_{c.ipw}(\overline{a^\prime}_{T})$ for $\forall \overline{a^\prime}_{T} \in \overline{\mathcal{A}_{T}}$). We present the result without the proof as below.
	
	\begin{align*}
		C_{T}^\prime\left[ \left\{ \frac{\delta^2p(1-p) + (1-p)^2}{(\delta p + 1 - p)^2} \right\}^{T} - (1-p)^{T} \right] & \leq \text{VR}(\widehat{\psi}_{c.ipw}(\overline{\bm{0}}), \widehat{\psi}_{inc}) \\
		& \leq C_{T}^\prime \zeta^\prime(T;p) \left\{ \frac{\delta^2p(1-p) + (1-p)^2}{(\delta p + 1 - p)^2} \right\}^{T}
	\end{align*}
	where we define $C_{T}^\prime = \frac{b_u^2}{\E\left[\left(Y^2 \right)^{\overline{\bm{0}}} \right]}$ and $\zeta^\prime(T;p) = \left( 1+ \frac{c \left(\E\left[Y^{\overline{\bm{1}}}\right]\right)^2}{\left( 1/(1-p) \right)^{T} \E\left[\left(Y^{\overline{\bm{1}}} \right)^2 \right]} \right)$.
	\\

	\subsection{Proof of Corollary \ref{cor:inf-time-horizon}}
	\label{proof:cor-inf-time}
		
	Now we provide following Lemma \ref{lem:inf-time-decomp} which becomes a key to prove Corollary \ref{cor:inf-time-horizon}.
	
	\begin{lemma} \label{lem:inf-time-decomp}
		Assume that $\pi_t(H_t)=p$ for all $1\leq t\leq T$ for $0<p<1$. Then we have following variance decomposition :
		\begin{equation*}
		\begin{aligned}
		& \var(\widehat{\psi}_{inc})  
		=  \var \left( \sum_{  \overline{a}_{T} \in \overline{\mathcal{A}}_{T} } \sqrt{w(\overline{a}_{T}; \delta, p)} \widehat{\psi}_{c.ipw}(\overline{a}_{T}) \right)
		\end{aligned}
		\end{equation*}
		where for $\forall \overline{a}_{T} \in \overline{\mathcal{A}}_{T}$ the weight $w$ is defined by
		$$ 
		w(\overline{a}_{T}; \delta, p) = \prod_{t=1}^{T} \frac{\pi_t(a_t) \left\{ \mathbbm{1}\left({a}_{t}=1\right)\delta^2p + \mathbbm{1}\left({a}_{t}=0\right)(1-p) \right\} }{(\delta{\pi}_{t}(H_t) + 1-{\pi}_{t}(H_t))^2}.
		$$
	\end{lemma}
	
	\begin{proof}
		From the last display for $\mathbb{V}_{inc.1}$, we have that
		\begin{align*}
		& \mathbb{V}_{inc.1}  \\
		& = \sum_{\overline{a}_{T} \in \overline{\mathcal{A}}_{T}} \int_{\mathcal{X}} \E\left[Y^2 \mid H_{T}, A_{T}=a_{T} \right] \prod_{t=1}^{T} \frac{\pi_t(a_t) \left( \mathbbm{1}\left({a}_{t}=1\right)\delta^2p + \mathbbm{1}\left({a}_{t}=0\right)\{1-p\} \right) }{(\delta p + 1-p)^2}  \\
		& \qquad \qquad \quad \times \prod_{t=1}^{T} \frac{d\Pb(X_t \mid \overline{X}_{t-1}, \overline{A}_{t-1}=\overline{a}_{t-1})}{\pi_t(a_t)} \\
		&=  \sum_{\overline{a}_{T} \in \overline{\mathcal{A}}_{T}} w(\overline{a}_{T}; \delta, p)  \int_{\mathcal{X}} \E\left[Y^2 \mid H_{T}, A_{T}=a_{T} \right] 
		\prod_{t=1}^{T} \frac{d\Pb(X_t \mid \overline{X}_{t-1}, \overline{A}_{t-1}=\overline{a}_{t-1})}{\pi_t(a_t)} \\
		&= \sum_{\overline{a}_{T} \in \overline{\mathcal{A}}_{T}} w(\overline{a}_{T}; \delta, p) \mathbb{V}_{c.ipw.1}(\overline{a}_{T} )
		\end{align*}
		where we let weight $w(\overline{a}_{T}; \delta, p)$ denote the product term 
		$\prod_{t=1}^{T} \frac{\pi_t(a_t) \left( \mathbbm{1}\left({a}_{t}=1\right)\delta^2p + \mathbbm{1}\left({a}_{t}=0\right)\{1-p\} \right) }{(\delta{\pi}_{t}(H_t) + 1-{\pi}_{t}(H_t))^2}$.
		
		Next, we observe that
		\begin{align*}
		\mathbb{V}_{inc.2} &= \left\{  \E  \left[ \prod_{t=1}^{T}  \left( \frac{\delta A_t + 1-A_t }{\delta p + 1-p} \right)Y \right]  \right\}^2  \\
		&= \left\{  \E  \left[ \prod_{t=1}^{T}  \left( \frac{\delta \mathbbm{1}\left({A}_{t}=1\right)  }{\delta p + 1-p} \right)Y + \ \cdots \ + \prod_{t=1}^{T}  \left( \frac{ \mathbbm{1}\left({A}_{t}=0 \right)  }{\delta p + 1-p} \right)Y \right]  \right\}^2 \\
		&= \sum_{\overline{a}_{T} \in \overline{\mathcal{A}}_{T}} v^2_{inc.2}(\overline{A}_{T}; \overline{a}_{T}) + \sum_{\overline{a^\prime}_{T} \neq \overline{a}_{T} } v_{inc.2}(\overline{A}_{T}; \overline{a}_{T}) v_{inc.2}(\overline{A}_{T}; \overline{a^\prime}_{T})
		\end{align*}
		where we have decomposed $\mathbb{V}_{inc.2}$ into $2^{T} \times 2^{T}$ terms by defining $v_{inc.2}(\overline{A}_{T}; \overline{a}_{T})$ by 
		$$
		v_{inc.2}(\overline{A}_{T}; \overline{a}_{T}) \equiv  \E  \left[ \prod_{t=1}^{T}  \left( \frac{\delta \mathbbm{1}(a_t=1) + \mathbbm{1}(a_t=0) }{\delta p + 1-p} \right)\mathbbm{1}(A_t=a_t) \cdot Y \right]. 
		$$ 
		Then for fixed $\overline{a}_{T}$ it is straightforward to see that 
		\begin{align*}
		\frac{v^2_{inc.2}(\overline{A}_{T}; \overline{a}_{T})}{w(\overline{a}_{T}; \delta, p)}  &= \left\{  \E  \left[ \prod_{t=1}^{T}  \left( \frac{ \left\{ \delta \mathbbm{1}(a_t=1) + \mathbbm{1}(a_t=0)  \right\} \mathbbm{1}(A_t=a_t) }{ \sqrt{\pi(a_t) \left(  \mathbbm{1}\left({a}_{t}=1\right)\delta^2p + \mathbbm{1}\left({a}_{t}=0\right)\{1-p\} \right) } } \right)Y \right]  \right\}^2  \\
		&= \left\{  \E  \left[ \prod_{t=1}^{T}  \left( \frac{  \mathbbm{1}(A_t=a_t) }{ \pi(a_t)  } \right)Y \right]  \right\}^2 =  \mathbb{V}_{c.ipw.2}(\overline{a}_{T})
		\end{align*}
		
		Now putting this together, we obtain
		\begin{align*}
		& \mathbb{V}_{inc.1} - \mathbb{V}_{inc.2} \\
		&= \sum_{\overline{a}_{T} \in \overline{\mathcal{A}}_{T}} w(\overline{a}_{T}; \delta, p) \left\{  \mathbb{V}_{c.ipw.1}(\overline{a}_{T} ) -  \mathbb{V}_{c.ipw.2}(\overline{a}_{T})  \right\} - \sum_{\overline{a^\prime}_{T} \neq \overline{a}_{T} } v_{inc.2}(\overline{A}_{T}; \overline{a}_{T}) v_{inc.2}(\overline{A}_{T}; \overline{a^\prime}_{T}) \\
		&= \sum_{\overline{a}_{T} \in \overline{\mathcal{A}}_{T}} w(\overline{a}_{T}; \delta, p)	 Var\left(\widehat{\psi}_{c.ipw}(\overline{a}_{T})\right)  - \sum_{\overline{a^\prime}_{T} \neq \overline{a}_{T} } v_{inc.2}(\overline{A}_{T}; \overline{a}_{T}) v_{inc.2}(\overline{A}_{T}; \overline{a^\prime}_{T}).
		\end{align*}
		
		However, from the second term in the last display one could notice that
		\begin{align*}
		\frac{v_{inc.2}(\overline{A}_{T}; \overline{a}_{T}) v_{inc.2}(\overline{A}_{T}; \overline{a^\prime}_{T})}{\sqrt{w(\overline{a}_{T}; \delta, p)w(\overline{a^\prime}_{T}; \delta, p)}} &=  \E  \left[ \prod_{t=1}^{T}  \left( \frac{  \mathbbm{1}(A_t=a_t) }{ \pi(a_t)  } \right)Y \right] \E  \left[ \prod_{t=1}^{T}  \left( \frac{  \mathbbm{1}(A_t=a^\prime_t) }{ \pi(a^\prime_t)  } \right)Y \right]  \\
		&= -Cov(\widehat{\psi}_{c.ipw}(\overline{a}_{T}), \widehat{\psi}_{c.ipw}(\overline{a^\prime}_{T}))
		\end{align*}
		where the last equality follows by the fact that
		$$
		\E\left\{ \prod_{t=1}^{T}  \left( \frac{  \mathbbm{1}(A_t=a_t) }{ \pi(a_t)  } \right) \prod_{t=1}^{T}  \left( \frac{  \mathbbm{1}(A_t=a^\prime_t) }{ \pi(a^\prime_t)  } \right) Y^2  \right\} = 0 \quad \text{for } \ \forall \overline{a^\prime}_{T} \neq \overline{a}_{T}.
		$$
		
		Hence finally we conclude that
		\begin{align*}
		& \var(\widehat{\psi}_{inc}) = \mathbb{V}_{inc.1} - \mathbb{V}_{inc.2} \\
		&= \sum_{\overline{a}_{T} \in \overline{\mathcal{A}}_{T}} w(\overline{a}_{T}; \delta, p)	 \var\left(\widehat{\psi}_{c.ipw}(\overline{a}_{T})\right) + \sum_{\substack{ \overline{a}_{T}, \overline{a^\prime}_{T} \in \overline{\mathcal{A}}_{T} \\ \overline{a^\prime}_{T} \neq \overline{a}_{T} } } \sqrt{w(\overline{a}_{T}; \delta, p)w(\overline{a^\prime}_{T}; \delta, p)}Cov(\widehat{\psi}_{c.ipw}(\overline{a}_{T}), \widehat{\psi}_{c.ipw}(\overline{a^\prime}_{T})) \\
		&= \sum_{  \overline{a}_{T}, \overline{a^\prime}_{T} \in \overline{\mathcal{A}}_{T} } \sqrt{w(\overline{a}_{T}; \delta, p)w(\overline{a^\prime}_{T}; \delta, p)}Cov(\widehat{\psi}_{c.ipw}(\overline{a}_{T}), \widehat{\psi}_{c.ipw}(\overline{a^\prime}_{T})).
		\end{align*}
		
	\end{proof}
	
	In Lemma \ref{lem:inf-time-decomp} it should be noticed that the weight $w(\overline{a}_{T}; \delta, p)$ exponentially and monotonically decays to zero for $\forall \overline{a}_{T} \in \overline{\mathcal{A}}_{T}$.
	
	Now we show that there always exists $T_{min}$ such that $\var(\widehat{\psi}_{inc}) < \var(\widehat{\psi}_{c.ipw}(\overline{\bm{1}}))$ for all $T \geq T_{min}$.  Let $\overline{\bm{1}} = [1,...,1]$. From Lemma \ref{lem:inf-time-decomp} it follows that
	\begin{align*}
	&   \var(\widehat{\psi}_{inc}) - \var(\widehat{\psi}_{c.ipw}(\overline{\bm{1}})) \\
	& = \sum_{\overline{a}_{T} \in \overline{\mathcal{A}}_{T}} w(\overline{a}_{T}; \delta, p)	 \var\left(\widehat{\psi}_{c.ipw}(\overline{a}_{T})\right) - \var(\widehat{\psi}_{c.ipw}(\overline{\bm{1}}))  \\
	& \quad + \sum_{\substack{ \overline{a}_{T}, \overline{a^\prime}_{T} \in \overline{\mathcal{A}}_{T} \\ \overline{a^\prime}_{T} \neq \overline{a}_{T} } } \sqrt{w(\overline{a}_{T}; \delta, p)w(\overline{a^\prime}_{T}; \delta, p)}Cov(\widehat{\psi}_{c.ipw}(\overline{a}_{T}), \widehat{\psi}_{c.ipw}(\overline{a^\prime}_{T})) \\
	& = \sum_{\overline{a}_{T} \in \overline{\mathcal{A}}_{T}} \prod_{t=1}^{T} \frac{\pi_t(a_t)\left\{ \mathbbm{1}\left({a}_{t}=1\right)\delta^2p + \mathbbm{1}\left({a}_{t}=0\right)(1-p)\right\}}{(\delta p + 1-p)^2}  \left( \prod_{t=1}^{T} \frac{1}{\pi_t(a_t)}  \E\left[\left(Y^2 \right)^{\overline{a}_{T}} \right]  - \left(\E\left[Y^{\overline{a}_{T}}\right]\right)^2 \right)  \\
	& \quad - \left( \frac{1}{p} \right)^{T} \E\left[\left(Y^{\overline{\bm{1}}} \right)^2 \right] + \left(\E\left[Y^{\overline{\bm{1}}}\right]\right)^2 - \sum_{\substack{ \overline{a}_{T}, \overline{a^\prime}_{T} \in \overline{\mathcal{A}}_{T} \\ \overline{a^\prime}_{T} \neq \overline{a}_{T} } } \sqrt{w(\overline{a}_{T}; \delta, p)w(\overline{a^\prime}_{T}; \delta, p)} \E\Big[Y^{\overline{a}_T}\Big] \E\Big[Y^{\overline{a^\prime}_T}\Big] \\
	& \leq b^2_u \sum_{\overline{a}_{T} \in \overline{\mathcal{A}}_{T}} \left(\prod_{t=1}^{T} \frac{ \mathbbm{1}\left({a}_{t}=1\right)\delta^2p + \mathbbm{1}\left({a}_{t}=0\right)(1-p)}{(\delta p + 1-p)^2} \right)  - \left( \frac{1}{p} \right)^{T} \E\left[\left(Y^{\overline{\bm{1}}} \right)^2 \right] + \left(\E\left[Y^{\overline{\bm{1}}}\right]\right)^2 \\
	& \quad - \sum_{\substack{ \overline{a}_{T}, \overline{a^\prime}_{T} \in \overline{\mathcal{A}}_{T} \\ \overline{a^\prime}_{T} \neq \overline{a}_{T} } } \sqrt{w(\overline{a}_{T}; \delta, p)w(\overline{a^\prime}_{T}; \delta, p)} \E\Big[Y^{\overline{a}_T}\Big] \E\Big[Y^{\overline{a^\prime}_T}\Big] + \sum_{\overline{a}_{T} \in \overline{\mathcal{A}}_{T}} w(\overline{a}_{T}; \delta, p)\left(\E\left[Y^{\overline{a}_{T}}\right]\right)^2 \\
	& = b^2_u \left\{ \left[\frac{\delta^2p+1-p}{(\delta p + 1 - p)^2}\right]^T - \left(\frac{c^{1/T}_{{\bm{1}}}}{p}\right)^T \right\} - \sum_{ \overline{a}_{T}, \overline{a^\prime}_{T} \in \overline{\mathcal{A}}_{T} } \sqrt{w(\overline{a}_{T}; \delta, p)w(\overline{a^\prime}_{T}; \delta, p)} \E\Big[Y^{\overline{a}_T}\Big] \E\Big[Y^{\overline{a^\prime}_T}\Big] + \left(\E\left[Y^{{\bm{1}}}\right]\right)^2 \\
	& = b^2_u \left\{\left[\frac{\delta^2p+1-p}{(\delta p + 1 - p)^2}\right]^T - \left(\frac{c^{1/T}_{{\bm{1}}}}{p}\right)^T - A(\delta,p) + B \right\}
	\end{align*}
	where $c_{\bm{1}}=\frac{\E\left[\left(Y^{\overline{\bm{1}}} \right)^2 \right]}{b^2_u}$, $A(\delta,p)=\sum_{ \overline{a}_{T}, \overline{a^\prime}_{T} \in \overline{\mathcal{A}}_{T} } \sqrt{w(\overline{a}_{T}; \delta, p)w(\overline{a^\prime}_{T}; \delta, p)} \frac{\E\left[Y^{\overline{a}_T}\right]}{b_u}\frac{\E\big[Y^{\overline{a^\prime}_T}\big]}{b_u}$, and $B=\frac{\left(\E\big[Y^{{\overline{\bm{1}}}}\big]\right)^2}{b^2_u}$. We note that $|A(\delta,p)|\leq1$, $0 \leq B \leq 1$, and  $c^{1/T}_{\bm{1}} \rightarrow 1$ as $T\rightarrow \infty$.

	For $\delta > 1$, $\frac{\delta^2p+1-p}{(\delta p + 1 - p)^2} < \frac{1}{p}$. Hence based on above observation, it follows that for sufficiently large $T$ the last display is strictly less than zero. Consequently we conclude $ \var(\widehat{\psi}_{inc}) - \var(\widehat{\psi}_{c.ipw}(\overline{\bm{1}})) < 0$ for all $T\geq T_{min}$, which is the result of part $i)$. Likewise, we have the same conclusion for $\overline{\bm{0}}_{T} = [0,...,0]$ such that $ \var(\widehat{\psi}_{inc}) - \var(\widehat{\psi}_{c.ipw}(\overline{\bm{0}}_{T})) < 0$.
	
	The value of $T_{min}$ is determined by $\delta, p$, and distribution of counterfactual outcome $Y^{\overline{a}_T}$. One rough upper bound of such $T_{min}$ is 
	$$
	\min \left\{T: \left[\frac{\delta^2p+1-p}{(\delta p + 1 - p)^2}\right]^T - \frac{c_{\bm{1}}}{p^T} + 2 < 0\right\}
	$$
	which could be obtained by the last display above and is always finite due to the fact $c_{\bm{1}}>0$ by given assumption in the theorem. $T_{min}$ should not be very large for moderately large value of $\delta$ unless $c_{\bm{1}}$ is unreasonably small since the difference $\frac{1}{p^T} - \left[\frac{\delta^2p+1-p}{(\delta p + 1 - p)^2}\right]^T$ also grows exponentially.  \\

	\subsection{Proof of Theorem \ref{thm:convergence}}
	\label{proof:thm-6-1}
	
	First we need to define the following notations:
	$$
	\Vert f \mid_{\mathcal{D},\mathcal{T}} \equiv \sup_{\delta\in\mathcal{D},t\in \mathcal{T}}\vert f(\delta,t) \mid
	$$
	$$
	\widehat{\Psi}_n(\delta, t) \equiv \sqrt[]{n}\{\widehat{\psi}_t (\delta) - \psi_t (\delta) \} / \widehat{\sigma}(\delta, t)
	$$
	$$
	\widetilde{\Psi}_n(\delta, t) \equiv \sqrt[]{n}\{\widehat{\psi}_t (\delta) - \psi_t (\delta) \} / \sigma(\delta, t)
	$$
	$$
	\Psi_n(\delta; t) \equiv \mathbb{G}_{n}\{\widetilde{\varphi}(Z;\bm{\eta},\delta, t) \}
	$$
	where we let $\mathcal{T} = \{1,...,T\}$, let $ \mathbb{G}_{n}$ denote the empirical process on the full sample as usual, and let $\widetilde{\varphi}(Z;\bm{\eta},\delta, t) = \{\varphi(Z;\bm{\eta},\delta, t) - \psi(t;\delta)\} / \sigma(\delta; t)$ and let $\mathbb{G}$ be a mean-zero Gaussian process with covariance $\E[\mathbb{G}(\delta_1; t_1)\mathbb{G}(\delta_2; t_2)]=\E\left[\widetilde{\varphi}(Z;\bm{\eta},\delta_1, t_1) \widetilde{\varphi}(Z;\bm{\eta},\delta_2, t_2)\right]$ as defined in Theorem \ref{thm:convergence} in the main text.
	
	The proof consists of two parts; in the first part we will show $\Psi_n(\cdot) \leadsto \mathbb{G}(\cdot)$ in $l^{\infty}(\mathcal{D},\mathcal{T})$ and in the second we will show $\Vert \widehat{\Psi}_n-\Psi_n \mid_{\mathcal{D},\mathcal{T}} = o_\Pb(1)$.
	
	\textbf{Part 1.} 
	A proof of the first statement immediately follows from the proof of Theorem 3 in \citet{Kennedy17} who showed the function class $\mathcal{F}_{\bar{\bm{\eta}}}=\{\varphi(\cdot; \bar{\bm{\eta}},\delta): \delta \in \mathcal{D} \}$ is Lipschitz and thus has a finite bracketing integral for any fixed set of nuisance functions. Then Theorem 2.5.6 in \citet{van1996weak} gives the result. In our case, the function class $\mathcal{F}_{\bar{\bm{\eta}}}=\{\varphi(\cdot; \bar{\bm{\eta}},\delta, t): \delta \in \mathcal{D},  t\leq T \}$ is still Lipschitz, since for $\forall t \in \{1,..., T\}$ we have
	$$
	\left| \frac{\partial }{\partial\delta}\left[ \frac{ \{a_t - \pi_t(h_t)\}(1-\delta)}{\delta a_t + 1-a_t}  \right] \right| \leq  \frac{1}{\delta_l} +  \frac{1}{4\delta_l^2}
	$$
	$$
	\left| \frac{\partial }{\partial\delta}\left[ \frac{m_t(h_t,1,1)\delta\pi_t(h_t)+m_t(h_t,0,1)\{ 1-\pi_t(h_t) \}}{\delta\pi_t(h_t) + 1 - \pi_t(h_t)} \cdot \omega_t(h_t, a_t) \right] \right| \leq  \frac{2C}{\delta_l^2}
	$$
	$$
	\frac{\partial }{\partial\delta}\left[ \frac{\delta a_t + 1-a_t}{\delta\pi_t(h_t) + 1-\pi_t(h_t)}\cdot\frac{1}{\omega_t( h_t, a_t)}  \right] \leq \frac{1}{\epsilon_\omega\delta_l^2 }
	$$
	where we use assumption 1) and 2) in the Theorem, and the identification assumption (\ref{assumption:A3}) that there exist a constant $\epsilon_\omega$ such that $0<\epsilon_\omega<\omega_t( h_t, a_t) \leq 1$ and thus $\frac{1}{\omega_t( h_t, a_t)} \leq \frac{1}{\epsilon_\omega}$ a.e. [$\Pb$]. Therefore, every $\varphi(\cdot; \bar{\bm{\eta}},\delta, t)$ is basically a finite sum of products of Lipschitz functions with bounded $\mathcal{D}$ and we thus conclude $\mathcal{F}_{\bar{\bm{\eta}}}$ is Lipschitz. Hence our function class still has a finite bracketing integral for fixed $\bar{\bm{\eta}}$ and $t$, which completes the first part of our proof.
	
	\textbf{Part 2.} 
	Let $N = n/K$ be the sample size in any group $k = 1, ..., K$, and denote the empirical process over group k units by $\mathbb{G}^k_n=\sqrt[]{N}(\Pb^k_n - \Pb)$. From the result of Part 1 and the proof of Theorem 3 in \citet{Kennedy17} we have
	\begin{align*}
		& \widetilde{\Psi}_n(\delta; t) - \Psi_n(\delta; t) \\
		& = \frac{\sqrt[]{n}}{K\sigma(\delta; t)}\sum_{k=1}^K \left[ \frac{1}{\sqrt[]{N}}\mathbb{G}^k_n\left\{\varphi(Z;\hat{\bm{\eta}}_{-k},\delta, t) - \varphi(Z;\bm{\eta},\delta, t) \right\} + \Pb\left\{ \varphi(Z;\hat{\bm{\eta}}_{-k},\delta, t) - \varphi(Z;\bm{\eta},\delta, t) \right\} \right] \\
		& \equiv B_{n,1}(\delta;t) + B_{n,2}(\delta;t).
	\end{align*}
	
	Now we analyze two pieces $B_{n,1}(\delta;t)$ and $B_{n,2}(\delta;t)$ in the last display.  $B_{n,1}(\delta;t)=o_\Pb(1)$ follows by the exact same steps done by \citet{Kennedy17}. However, analysis on $B_{n,2}(\delta;t)$ requires extra work.
	
	To analyze $B_{n,2}(\delta;t)$, we use the same notation used in \citet{Kennedy17}. First let $\psi(\Pb;Q)$ denote the mean outcome under intervention $Q$ for a population corresponding to observed data distribution $\Pb$. Next, let  $\varphi^*(z;{\bm{\eta}, t})$ denote its \textit{centered} efficient influence function when $Q$ does not depend on $\Pb$, as given in Lemma \ref{lem:eif_1} and let  $\zeta^*(z;{\bm{\eta}}, t)$ denote the contribution to the efficient influence function $\varphi^*(z;{\bm{\eta}, t})$ due to estimating $Q$ when it depends on $\Pb$, as given in Lemma \ref{lem:eif_2}. Now by definition,
	$$
	\varphi(Z;{\bm{\eta}, \delta, t}) = \varphi^*(Z;{\bm{\eta}, t}) + \psi(\Pb;Q) + \zeta^*(Z;{\bm{\eta}}, t),
	$$
	and after some rearrangement we obtain
	\begin{align*}
		\frac{1}{\sqrt[]{n}}B_{n,2}(\delta;t) & = \Pb\left\{ \varphi(Z;\overline{\bm{\eta}},\delta, t) - \varphi(Z;{\bm{\eta}},\delta, t)  \right\} \\
		& =\int \varphi^*(z;\overline{\bm{\eta}},t) d\Pb(z) + \psi(\overline{\Pb};\overline{Q}) -  \psi({\Pb};\overline{Q}) \\
		& \quad +  \int  \zeta^*(z;\overline{\bm{\eta}},t)  d\Pb(z) + \psi({\Pb};\overline{Q}) -  \psi({\Pb};Q).
	\end{align*}
	Although one can relate $\overline{\bm{\eta}}$ to $\widehat{\bm{\eta}}_{-k}$ in above equation, it can be anything associated with new $\overline{\Pb}$ and $\overline{Q}$.
	
	Hence, by analyzing the second order remainder terms of von Mises expansion for the efficient influence functions given in Lemma \ref{lem:eif_1} and \ref{lem:eif_2}, we can evaluate the convergence rate of $B_{n,2}(\delta;t)$.
	The following two lemmas analyze those second order remainder terms in the presence of dropout process.
	
	\begin{lemma} \label{lem:remainder_1}
		Let $\psi({\Pb};Q)$ be a mean outcome under intervention $Q$ for a for a population
		corresponding to observed data distribution $\Pb$, and let $\varphi^*(z;{\bm{\eta}},t)$ denote its efficient influence
		function when $Q$ does not depend on $\Pb$ for given $t$, as given in Lemma \ref{lem:eif_1}. For another data distribution $\overline{\Pb}$, let $\overline{\bm{\eta}}$ denote the corresponding nuisance functions. Then we have the 1st-order von Mises expansion
		\begin{equation*} 
			\begin{aligned}
				\psi(\overline{\Pb};Q) -  \psi({\Pb};Q) &= \int \varphi^*(z;\overline{\bm{\eta}},t) d\Pb(z) \\
				& +\sum_{s=1}^{t} \sum_{r=1}^{s} \int  \left(m^*_s- \overline{m}_s \right) \left( \prod_{k=1}^{s} dQ_k d\Pb_k \frac{d\omega_{k}}{d\overline{\omega}_k} \right) \left( \frac{d\pi_{r} - d\overline{\pi}_r}{d\overline{\pi}_r} \right)  \prod_{k=1}^{r-1} \left\{ \left(\frac{d\pi_k}{d\overline{\pi}_k}  \frac{d\omega_{k}}{d\overline{\omega}_k}  \right)  \right\}  \\
    			& + \sum_{s=1}^{t} \sum_{r=1}^{s} \int  \left(m^*_s- \overline{m}_s \right) \left( \prod_{k=1}^{s}dQ_kd\Pb_k \right) \left( \frac{d\omega_{r} - d\overline{\omega}_r}{d\overline{\omega}_r} \right)  \prod_{k=1}^{r-1} \left\{ \left(\frac{d\pi_k}{d\overline{\pi}_k}  \frac{d\omega_{k}}{d\overline{\omega}_k}  \right) \right\}
			\end{aligned}
		\end{equation*} 
		where we define 
		$$
		\overline{m}_s = \int \overline{m}_{s+1}dQ_{s+1}d\overline{\Pb}_{s+1}, \qquad {m}^*_s = \int \overline{m}_{s+1}dQ_{s+1}d{\Pb}_{s+1}, 
		$$
		$$
		dQ_t = dQ_t(A_t \mid H_t), \qquad d\pi_t = d\Pb(A_t \mid H_t), \qquad d\Pb_t = d\Pb(X_t \mid H_{t-1}, A_{t-1}),
		$$
		$$
		d\omega_s=d\Pb(R_{s+1}=1 \mid H_s, A_s,R_s=1), \qquad d\overline{\omega}_s=d\overline{\Pb}(R_{s+1}=1 \mid H_s, A_s,R_s=1).
		$$	
	\end{lemma}
	
	\begin{proof}
		From Lemma \ref{lem:eif_1}, we have
		\begin{align*}
			\E\{ \varphi^*(Z;\overline{\bm{\eta}}) \} &=  \sum_{s=0}^{t} \E \left\{ \left(\int \overline{m}_{s+1}dQ_{s+1} - \overline{m}_s \right) \mathbbm{1}(R_{s+1}=1)    \prod_{r=0}^s\left(\frac{dQ_r}{d\overline{\pi}_r}  \frac{1}{d\overline{\omega}_r}  \right) \right\} \\
			&=  \sum_{s=0}^{t} \E \left\{ \E \left[ \left(\int \overline{m}_{s+1}dQ_{s+1} - \overline{m}_s \right) \mathbbm{1}(R_{s+1}=1) \mathbbm{1}(R_{s}=1)  \prod_{r=0}^s\left(\frac{dQ_r}{d\overline{\pi}_r}  \frac{1}{d\overline{\omega}_r}  \right) \Bigg\vert H_s, A_s, R_s \right] \right\} \\
			&=  \sum_{s=0}^{t} \E \Bigg\{  \E \left[ \left( \int \overline{m}_{s+1}dQ_{s+1} - \overline{m}_s \right) \mathbbm{1}(R_{s}=1) \prod_{r=0}^s\left(\frac{dQ_r}{d\overline{\pi}_r}  \frac{1}{d\overline{\omega}_r}  \right) \Bigg\vert H_s, A_s, R_s = 1,  R_{s+1}=1 \right] \\
			& \qquad  \qquad \times   d\Pb(R_{s+1}=1 \mid H_s, A_s, R_s=1)  \Bigg\} \\
			&=  \sum_{s=0}^{t} \E \left\{   \left(\int \int \overline{m}_{s+1}dQ_{s+1}d\Pb_{s+1} - \overline{m}_s \right) \mathbbm{1}(R_{s}=1) d\omega_s   \prod_{r=0}^s\left(\frac{dQ_r}{d\overline{\pi}_r}  \frac{1}{d\overline{\omega}_r}  \right)  \right\} \\
			&= \sum_{s=0}^{t} \E \left\{ \left(m^*_s- \overline{m}_s \right) d\omega_s  \mathbbm{1}(R_{s}=1) \prod_{r=0}^s\left(\frac{dQ_r}{d\overline{\pi}_r}  \frac{1}{d\overline{\omega}_r}  \right)   \right\} \\
			&= \sum_{s=0}^{t} \int \left(m^*_s- \overline{m}_s \right)  \prod_{r=0}^s \left\{ \left(\frac{dQ_r}{d\overline{\pi}_r}  \frac{1}{d\overline{\omega}_r}  \right) d\pi_r d\Pb_r d\omega_{r} \right\}
		\end{align*}	
		where the first equality follows by the definition and linearity of expectation, the second by iterated expectation and the equivalence between $\mathbbm{1}(R_{s+1}=1)$ and $\mathbbm{1}(R_{s+1}=1, R_{s}=1)$ \footnote{For $\forall s \leq t$ the event $\{R_{s}=1\}$ implies $\{R_{s'}=1$ for all $s' \leq s\}$ by construction. }, the third by the law of total probability on conditional expectation \footnote{For random variables $X,Y,Z$, when $Z$ is discrete it follows $\E[X|Y] = \sum_z \E[X|Y, Z=z]\Pb(Z=z|Y)$.}, the fourth by the result of Lemma \ref{lem:identification} (i.e. $d\Pb_{s+1} = d\Pb(X_{s+1} \mid H_s, A_s, R_{s+1}=1)$). To obtain the last equality, we first apply iterated expectation conditioning on $(H_s, R_s)$, then do another iterated expectation conditioning on $(H_{s-1}, A_{s-1}, R_{s-1})$ followed by same steps from the second, the third and the fourth equalities, and repeat these processes for $s-2, ... , 1$. 
		
		From the last expression, now we have
		\begin{align*}
			\sum_{s=0}^{t} \int & \left(m^*_s- \overline{m}_s \right)  \prod_{r=0}^s \left\{ \left(\frac{dQ_r}{d\overline{\pi}_r}  \frac{d\omega_{r}}{d\overline{\omega}_r}  \right) d\pi_r d\Pb_r  \right\} \\
			&= \sum_{s=0}^{t} \int  \left(m^*_s- \overline{m}_s \right)  \frac{d\pi_{s}}{d\overline{\pi}_s} \frac{d\omega_{s}}{d\overline{\omega}_s} dQ_sd\Pb_s  \prod_{r=0}^{s-1} \left\{ \left(\frac{dQ_r}{d\overline{\pi}_r}  \frac{d\omega_{r}}{d\overline{\omega}_r}  \right) d\pi_r d\Pb_r  \right\} \\
			&= \sum_{s=0}^{t} \int  \left(m^*_s- \overline{m}_s \right) \left( \frac{d\pi_{s} - d\overline{\pi}_s}{d\overline{\pi}_s} \right) \frac{d\omega_{s}}{d\overline{\omega}_s} dQ_sd\Pb_s  \prod_{r=0}^{s-1} \left\{ \left(\frac{dQ_r}{d\overline{\pi}_r}  \frac{d\omega_{r}}{d\overline{\omega}_r}  \right) d\pi_r d\Pb_r  \right\}  \\
			& \quad + \sum_{s=0}^{t} \int  \left(m^*_s- \overline{m}_s \right)  \frac{d\omega_{s}}{d\overline{\omega}_s} dQ_sd\Pb_s  \prod_{r=0}^{s-1} \left\{ \left(\frac{dQ_r}{d\overline{\pi}_r}  \frac{d\omega_{r}}{d\overline{\omega}_r}  \right) d\pi_r d\Pb_r  \right\} \\
			& = \sum_{s=1}^{t} \int  \left(m^*_s- \overline{m}_s \right) \left( \frac{d\pi_{s} - d\overline{\pi}_s}{d\overline{\pi}_s} \right) \frac{d\omega_{s}}{d\overline{\omega}_s} dQ_sd\Pb_s  \prod_{r=0}^{s-1} \left\{ \left(\frac{dQ_r}{d\overline{\pi}_r}  \frac{d\omega_{r}}{d\overline{\omega}_r}  \right) d\pi_r d\Pb_r  \right\}  \\
			& \quad + \sum_{s=1}^{t} \int  \left(m^*_s- \overline{m}_s \right) \left( \frac{d\omega_{s} - d\overline{\omega}_s}{d\overline{\omega}_s} \right) dQ_sd\Pb_s  \prod_{r=0}^{s-1} \left\{ \left(\frac{dQ_r}{d\overline{\pi}_r}  \frac{d\omega_{r}}{d\overline{\omega}_r}  \right) d\pi_r d\Pb_r  \right\} \\
			& \quad + \sum_{s=1}^{t} \int  \left(m^*_s- \overline{m}_s \right)  dQ_sd\Pb_s  \prod_{r=0}^{s-1} \left\{ \left(\frac{dQ_r}{d\overline{\pi}_r}  \frac{d\omega_{r}}{d\overline{\omega}_r}  \right) d\pi_r d\Pb_r  \right\}  + \left(m^*_0- \overline{m}_0 \right).
		\end{align*}
		 Note that we use the convention from earlier lemmas that all the quantities with negative times (e.g., $dQ_{-1}$) are set to one. After repeating above process $t-1$ times to the second last term in the last display, we obtain that
		\begin{align*}
			\sum_{s=0}^{t} \int & \left(m^*_s- \overline{m}_s \right)  \prod_{r=0}^s \left\{ \left(\frac{dQ_r}{d\overline{\pi}_r}  \frac{d\omega_{r}}{d\overline{\omega}_r}  \right) d\pi_r d\Pb_r  \right\} \\
			& = \sum_{s=1}^{t} \sum_{r=1}^{s} \int  \left(m^*_s- \overline{m}_s \right) \left( \prod_{k=r}^{s} dQ_k d\Pb_k \frac{d\omega_{k}}{d\overline{\omega}_k} \right) \left( \frac{d\pi_{r} - d\overline{\pi}_r}{d\overline{\pi}_r} \right)  \prod_{k=1}^{r-1} \left\{ \left(\frac{dQ_k}{d\overline{\pi}_k}  \frac{d\omega_{k}}{d\overline{\omega}_k}  \right) d\pi_k d\Pb_k  \right\}  \\
			& \quad + \sum_{s=1}^{t} \sum_{r=1}^{s} \int  \left(m^*_s- \overline{m}_s \right) \left( \prod_{k=r}^{s}dQ_kd\Pb_k \right) \left( \frac{d\omega_{r} - d\overline{\omega}_r}{d\overline{\omega}_r} \right)  \prod_{k=1}^{r-1} \left\{ \left(\frac{dQ_k}{d\overline{\pi}_k}  \frac{d\omega_{k}}{d\overline{\omega}_k}  \right) d\pi_k d\Pb_k  \right\}  \\
			& \quad + \sum_{s=1}^{t} \int  \left(m^*_s- \overline{m}_s \right)   \left( \prod_{r=1}^{s} dQ_r d\Pb_r \right) + \left(m^*_0- \overline{m}_0 \right).
		\end{align*}
		By Lemma 5 in \citet{Kennedy17} it follows
		$$
		\sum_{s=1}^{t} \int  \left(m^*_s- \overline{m}_s \right)   \left( \prod_{r=1}^{s} dQ_r d\Pb_r \right) = m_0 - m^*_0.
		$$
		Putting all these together, we have
		\begin{align*}
			\E\{ \varphi^*(Z;\overline{\bm{\eta}}) \} &= m_0 -\overline{m}_0 \\
			& +\sum_{s=1}^{t} \sum_{r=1}^{s} \int  \left(m^*_s- \overline{m}_s \right) \left( \prod_{k=1}^{s} dQ_k d\Pb_k \frac{d\omega_{k}}{d\overline{\omega}_k} \right) \left( \frac{d\pi_{r} - d\overline{\pi}_r}{d\overline{\pi}_r} \right)  \prod_{k=1}^{r-1} \left\{ \left(\frac{d\pi_k}{d\overline{\pi}_k}  \frac{d\omega_{k}}{d\overline{\omega}_k}  \right)  \right\}  \\
			& \quad + \sum_{s=1}^{t} \sum_{r=1}^{s} \int  \left(m^*_s- \overline{m}_s \right) \left( \prod_{k=1}^{s}dQ_kd\Pb_k \right) \left( \frac{d\omega_{r} - d\overline{\omega}_r}{d\overline{\omega}_r} \right)  \prod_{k=1}^{r-1} \left\{ \left(\frac{d\pi_k}{d\overline{\pi}_k}  \frac{d\omega_{k}}{d\overline{\omega}_k}  \right) \right\}
		\end{align*}
		, which yields the formula we have in Lemma \ref{lem:remainder_1}.
	\end{proof}
	
	\begin{lemma} \label{lem:remainder_2}
		Let $\zeta^*(z;\overline{\bm{\eta}},t)$ denote the contribution to
		the efficient influence function $\varphi^*(z;{\bm{\eta}},t)$ due to dependence between $\Pb$ and $Q$ as given in Lemma \ref{lem:eif_2}. Then for two different intervention distributions $Q$ and $\overline{Q}$ whose corresponding densities are $dQ_t$ and $d\overline{Q}_t$ respectively with respect to some dominating measure for $t = 1, ..., t$, we have the 1st-order Von Mises expansion
		\begin{equation*} 
			\begin{aligned}
				\psi({\Pb};\overline{Q}) -  \psi & ({\Pb};Q)  = \int  \zeta^*(z;\overline{\bm{\eta}},t)  d\Pb(z) \\
				& + \sum_{s=1}^{t}  \int \overline{\phi}_s d\pi_s (m_s - \overline{m}_s) d\nu d\Pb_s   \prod_{r=0}^{s-1}  \left( \frac{d\overline{Q}_r}{d\overline{\pi}_r}  \frac{d{\omega}_r}{d\overline{\omega}_r} d{\pi}_r d{\Pb}_r  \right) \\
    			&  + \sum_{s=1}^{t} \sum_{r=1}^{s} \int \overline{\phi}_s d\pi_s m_s d\nu d\Pb_s  \left( \prod_{k=0}^{s-1} d\overline{Q}_kd\Pb_k\frac{d\omega_k}{d\overline{\omega}_k} \right) \left( \frac{d\pi_r - d\overline{\pi}_r}{d\overline{\pi}_r} \right) \prod_{k=1}^{s-1}  \left(\frac{d\pi_k}{d\overline{\pi}_k}  \frac{d\omega_{k}}{d\overline{\omega}_k}  \right) \\
    			&  + \sum_{s=1}^{t} \sum_{r=1}^{s} \int \overline{\phi}_s d\pi_s m_s d\nu d\Pb_s  \left( \prod_{r=0}^{t-1} d\overline{Q}_kd\Pb_k \right) \left( \frac{d\omega_r - d\overline{\omega}_r}{d\overline{\omega}_r} \right)  \prod_{k=1}^{s-1}  \left(\frac{d\pi_k}{d\overline{\pi}_k}  \frac{d\omega_{k}}{d\overline{\omega}_k}  \right) \\
    			& + \sum_{s=1}^{t} \int m_s\left( d\overline{Q}_s - d{Q}_s - \overline{\phi}_s d\pi_s d\nu \right)d\Pb_s  \left( \prod_{r=0}^{s-1} d\overline{Q}_r d\Pb_r \right)
			\end{aligned}
		\end{equation*} 
		where we define all the notation in the same way in Lemma \ref{lem:remainder_1}.
	\end{lemma}

	\begin{proof}
		From Lemma 6 in \citet{Kennedy17} and by Lemma \ref{lem:identification}, we have
		\begin{align*}
			\psi({\Pb};\overline{Q}) -  \psi({\Pb};Q) &= \int m_t \left( \prod_{s=1}^{t} d\overline{Q}_sd\Pb_s  - \prod_{s=1}^{t} dQ_sd\Pb_s  \right)	\\
			&=  \sum_{s=1}^{t} \int  m_t \left( d\overline{Q}_t - d{Q}_t \right) d\Pb_t  \prod_{r=0}^{s-1} d\overline{Q}_s d\Pb_s.
		\end{align*}
		Next, for the expected contribution to the influence function due to estimating $Q$ when it depends on $\Pb$, we have that
		\begin{align*}
			\E[\zeta^*(Z;\overline{\bm{\eta}})] &= \E \left[ \sum_{s=1}^{t} \int \overline{\phi}_s \overline{m}_s d\nu  \left( \prod_{r=0}^{s-1}  \frac{d\overline{Q}_r}{d\overline{\pi}_r} \frac{1}{d\overline{\omega}_r} \right) \mathbbm{1}(R_{s}=1)  \right] \\
			&=\sum_{s=1}^{t}   \E \left[ \int \overline{\phi}_s d\pi_s \overline{m}_s d\nu \left( \prod_{r=0}^{s-1} \frac{d\overline{Q}_r}{d\overline{\pi}_r} \frac{1}{d\overline{\omega}_r} \right) \mathbbm{1}(R_{s}=1) \mathbbm{1}(R_{s-1}=1)  \right] \\
			&=  \sum_{s=1}^{t} \E \left\{ \left[ \int \overline{\phi}_s d\pi_s \overline{m}_s d\nu d\Pb_s  \left( \prod_{r=0}^{s-1} \frac{d\overline{Q}_r}{d\overline{\pi}_r} \frac{1}{d\overline{\omega}_r} \right)  \mathbbm{1}(R_{s-1}=1) \right] d\Pb(R_s=1 \mid H_{s-1}, A_{s-1}, R_{s-1}=1) \right\} \\
			&=  \sum_{s=1}^{t} \E \left\{  \int \overline{\phi}_s d\pi_s \overline{m}_s d\nu d\Pb_s  \prod_{r=0}^{s-1} \left( \frac{d\overline{Q}_r}{d\overline{\pi}_r} \frac{1}{d\overline{\omega}_r} \right) d{\omega}_{s-1} \mathbbm{1}(R_{s-1}=1)  \right\} \\
			&= \sum_{s=1}^{t} \int \overline{\phi}_s d\pi_s \overline{m}_s d\nu d\Pb_s  \prod_{r=0}^{s-1}  \left( \frac{d\overline{Q}_r}{d\overline{\pi}_r}  \frac{1}{d\overline{\omega}_r} d{\pi}_r d{\Pb}_r d{\omega}_r \right)   
		\end{align*}
		where the first equality by definition, the second by iterated expectation conditioning on $(H_s, R_s)$ and averaging over $A_s$, the third by iterated expectation conditioning on $(H_{s-1}, A_{s-1}, R_{s-1})$ and law of total probability, and the fifth by repeating the process $t$ times. 
		
		Now, we further expand our last expression as
		\begin{align*}
			\sum_{s=1}^{t} & \int \overline{\phi}_s d\pi_s \overline{m}_s d\nu d\Pb_s   \prod_{r=0}^{s-1}  \left( \frac{d\overline{Q}_r}{d\overline{\pi}_r}  \frac{1}{d\overline{\omega}_r} d{\pi}_r d{\Pb}_r d{\omega}_r \right) \\
			&= \sum_{s=1}^{t}  \int \overline{\phi}_s d\pi_s (\overline{m}_s - m_s) d\nu d\Pb_s   \prod_{r=0}^{s-1}  \left( \frac{d\overline{Q}_r}{d\overline{\pi}_r}  \frac{1}{d\overline{\omega}_r} d{\pi}_r d{\Pb}_r d{\omega}_r \right) \\
			& \qquad + \sum_{s=1}^{t}  \int \overline{\phi}_s d\pi_s m_s d\nu d\Pb_s   \prod_{r=0}^{s-1}  \left( \frac{d\overline{Q}_r}{d\overline{\pi}_r}  \frac{1}{d\overline{\omega}_r} d{\pi}_r d{\Pb}_r d{\omega}_r \right)  \\
			&= \sum_{s=1}^{t}  \int \overline{\phi}_s d\pi_s (\overline{m}_s - m_s) d\nu d\Pb_s   \prod_{r=0}^{s-1}  \left( \frac{d\overline{Q}_r}{d\overline{\pi}_r}  \frac{1}{d\overline{\omega}_r} d{\pi}_r d{\Pb}_r d{\omega}_r \right) \\
			& \quad +\sum_{s=1}^{t} \sum_{r=1}^{s} \int \overline{\phi}_s d\pi_s m_s d\nu d\Pb_s  \left( \prod_{k=0}^{s-1} d\overline{Q}_kd\Pb_k\frac{d\omega_k}{d\overline{\omega}_k} \right) \left( \frac{d\pi_r - d\overline{\pi}_r}{d\overline{\pi}_r} \right) \prod_{k=1}^{s-1}  \left(\frac{d\pi_k}{d\overline{\pi}_k}  \frac{d\omega_{k}}{d\overline{\omega}_k}  \right) \\
			& \quad + \sum_{s=1}^{t} \sum_{r=1}^{s} \int \overline{\phi}_s d\pi_s m_s d\nu d\Pb_s  \left( \prod_{r=0}^{t-1} d\overline{Q}_kd\Pb_k \right) \left( \frac{d\omega_r - d\overline{\omega}_r}{d\overline{\omega}_r} \right)  \prod_{k=1}^{s-1}  \left(\frac{d\pi_k}{d\overline{\pi}_k}  \frac{d\omega_{k}}{d\overline{\omega}_k}  \right)  \\
			& \quad + \sum_{s=1}^{t} \int  \overline{\phi}_s d\pi_s m_s d\nu d\Pb_s   \left( \prod_{r=0}^{s-1} d\overline{Q}_s d\Pb_s \right)
		\end{align*}
		where the first equality follows by adding and subtracting the second term, an the second by the same steps used in Lemma \ref{lem:remainder_1}.
		
		With the last term in the last expression above, it follows
		\begin{align*}
			\psi({\Pb};\overline{Q}) -  \psi({\Pb};Q) & - \sum_{s=1}^{t} \int  \overline{\phi}_s d\pi_s m_s d\nu d\Pb_s   \left( \prod_{r=0}^{s-1} d\overline{Q}_r d\Pb_r \right) \\	
			&= \sum_{s=1}^{t} \int m_s\left( d\overline{Q}_s - d{Q}_s - \overline{\phi}_s d\pi_s d\nu \right)d\Pb_s  \left( \prod_{r=0}^{s-1} d\overline{Q}_r d\Pb_r \right).
		\end{align*}
		
		Putting these all together, finally we have
		\begin{align*}
			\Psi({\Pb};\overline{Q}) -  \Psi & ({\Pb};Q)  = \E[ \zeta^*(Z;\overline{\bm{\eta}})] \\
			& + \sum_{s=1}^{t}  \int \overline{\phi}_s d\pi_s (m_s - \overline{m}_s) d\nu d\Pb_s   \prod_{r=0}^{s-1}  \left( \frac{d\overline{Q}_r}{d\overline{\pi}_r}  \frac{d{\omega}_r}{d\overline{\omega}_r} d{\pi}_r d{\Pb}_r  \right) \\
			&  + \sum_{s=1}^{t} \sum_{r=1}^{s} \int \overline{\phi}_s d\pi_s m_s d\nu d\Pb_s  \left( \prod_{k=0}^{s-1} d\overline{Q}_kd\Pb_k\frac{d\omega_k}{d\overline{\omega}_k} \right) \left( \frac{d\pi_r - d\overline{\pi}_r}{d\overline{\pi}_r} \right) \prod_{k=1}^{s-1}  \left(\frac{d\pi_k}{d\overline{\pi}_k}  \frac{d\omega_{k}}{d\overline{\omega}_k}  \right) \\
			&  + \sum_{s=1}^{t} \sum_{r=1}^{s} \int \overline{\phi}_s d\pi_s m_s d\nu d\Pb_s  \left( \prod_{r=0}^{t-1} d\overline{Q}_kd\Pb_k \right) \left( \frac{d\omega_r - d\overline{\omega}_r}{d\overline{\omega}_r} \right)  \prod_{k=1}^{s-1}  \left(\frac{d\pi_k}{d\overline{\pi}_k}  \frac{d\omega_{k}}{d\overline{\omega}_k}  \right) \\
			& + \sum_{s=1}^{t} \int m_s\left( d\overline{Q}_s - d{Q}_s - \overline{\phi}_s d\pi_s d\nu \right)d\Pb_s  \left( \prod_{r=0}^{s-1} d\overline{Q}_r d\Pb_r \right).
		\end{align*}
	\end{proof}
	
	Finally, the next Lemma completes the proof of the Theorem \ref{thm:convergence}. 
	
	\begin{lemma} \label{lem:upperbound_remainder}
		Remainders of the von Mises expansion from Lemma \ref{lem:remainder_1} and \ref{lem:remainder_2} are both diminishing at rate of $n^{-\frac{1}{2}}$ uniformly in $\delta$, if 
		$$
		\left( \underset{\delta\in \mathcal{D}}{sup}\Vert m_{\delta,s} -  \widehat{m}_{\delta,s}  \Vert + \Vert \pi_s -  \widehat{\pi}_{s} \Vert \right) \Big(  \Vert \pi_r - \overline{\pi}_r\Vert + \Vert \omega_r - \overline{\omega}_r\Vert \Big)  = o_\Pb(\frac{1}{\sqrt{n}}),
		$$
		for $\forall r \leq s \leq t$.
		
	\end{lemma}
	\begin{proof}
		The remainder term of the Von Mises expansion from Lemma \ref{lem:remainder_1} equals
		\begin{align*}
			& \sum_{s=1}^{t} \sum_{r=1}^{s} \int  \left(m^*_s- \overline{m}_s \right) \left( \prod_{k=1}^{s} dQ_k d\Pb_k \frac{d\omega_{k}}{d\overline{\omega}_k} \right) \left( \frac{d\pi_{r} - d\overline{\pi}_r}{d\overline{\pi}_r} \right)  \prod_{k=1}^{r-1} \left\{ \left(\frac{d\pi_k}{d\overline{\pi}_k}  \frac{d\omega_{k}}{d\overline{\omega}_k}  \right)  \right\}  \\
			& \quad + \sum_{s=1}^{t} \sum_{r=1}^{s} \int  \left(m^*_s- \overline{m}_s \right) \left( \prod_{k=1}^{s}dQ_kd\Pb_k \right) \left( \frac{d\omega_{r} - d\overline{\omega}_r}{d\overline{\omega}_r} \right)  \prod_{k=1}^{r-1} \left\{ \left(\frac{d\pi_k}{d\overline{\pi}_k}  \frac{d\omega_{k}}{d\overline{\omega}_k}  \right) \right\} \\
			& = \sum_{s=1}^{t} \sum_{r=1}^{s} \int \Big\{ (\overline{m}_{s+1} - {m}_{s+1})dQ_{s+1}d\Pb_{s+1}  + (m_s - \overline{m}_s)\Big\} \left( \prod_{k=1}^{s} dQ_k d\Pb_k \frac{d\omega_{k}}{d\overline{\omega}_k} \right) \left( \frac{d\pi_r - d\overline{\pi}_r}{d\overline{\pi}_r} \right)  \prod_{k=1}^{r-1}  \left(\frac{d\pi_k}{d\overline{\pi}_k}  \frac{d\omega_{k}}{d\overline{\omega}_k}  \right) \\
			& \quad + \sum_{s=1}^{t} \sum_{r=1}^{s} \int \Big\{ (\overline{m}_{s+1} - {m}_{s+1})dQ_{s+1}d\Pb_{s+1}  + (m_s - \overline{m}_s)\Big\} \left( \prod_{k=1}^{s}dQ_kd\Pb_k \right) \left( \frac{d\omega_r - d\overline{\omega}_r}{d\overline{\omega}_r} \right)  \prod_{k=1}^{r-1}  \left(\frac{d\pi_k}{d\overline{\pi}_k}  \frac{d\omega_{k}}{d\overline{\omega}_k}  \right) \\
			& \lesssim \sum_{s=1}^{t} \Big(  \Vert \overline{m}_{s+1} - {m}_{s+1} \Vert +  \Vert m_s -  \overline{m}_{s}  \Vert \Big) \sum_{r=1}^{s} \Big(  \Vert \pi_r - \overline{\pi}_r \Vert + \Vert \omega_r - \overline{\omega}_r\Vert \Big)
		\end{align*}
		where we obtain the first inequality simply by adding and subtracting $m_s$.
		
		For the remainder term from Lemma \ref{lem:remainder_2}, first note that by Lemma \ref{lem:identification} and Lemma 6 of \citet{Kennedy17},
		$$
		\int \overline{\phi}_s d\pi_s  = \frac{\delta(2a_s-1)(\pi_s - \overline{\pi}_s  )}{(\delta\overline{\pi}_s+1-\overline{\pi}_s)^2},
		$$
		$$
		d\overline{Q}_s - d{Q}_s - \int \overline{\phi}_s d\pi_s  = \frac{\delta(\delta-1)(2a_s-1)(\overline{\pi}_s - \pi_s)^2}{(\delta\overline{\pi}_s+1-\overline{\pi}_s)^2(\delta{\pi}_s+1-{\pi}_s)}.
		$$
		
		Hence, it immediately follows that the remainder term in Lemma \ref{lem:remainder_2} can be bounded by
		\begin{align*}
			& \sum_{s=1}^{t}  \int \overline{\phi}_s d\pi_s (m_s - \overline{m}_s) d\nu d\Pb_s   \prod_{r=0}^{s-1}  \left( \frac{d\overline{Q}_r}{d\overline{\pi}_r}  \frac{d{\omega}_r}{d\overline{\omega}_r} d{\pi}_r d{\Pb}_r  \right) \\
			&  + \sum_{s=1}^{t} \sum_{r=1}^{s} \int \overline{\phi}_s d\pi_s m_s d\nu d\Pb_s  \left( \prod_{k=0}^{s-1} d\overline{Q}_kd\Pb_k\frac{d\omega_k}{d\overline{\omega}_k} \right) \left( \frac{d\pi_r - d\overline{\pi}_r}{d\overline{\pi}_r} \right) \prod_{k=1}^{s-1}  \left(\frac{d\pi_k}{d\overline{\pi}_k}  \frac{d\omega_{k}}{d\overline{\omega}_k}  \right) \\
			&  + \sum_{s=1}^{t} \sum_{r=1}^{s} \int \overline{\phi}_s d\pi_s m_s d\nu d\Pb_s  \left( \prod_{r=0}^{t-1} d\overline{Q}_kd\Pb_k \right) \left( \frac{d\omega_r - d\overline{\omega}_r}{d\overline{\omega}_r} \right)  \prod_{k=1}^{s-1}  \left(\frac{d\pi_k}{d\overline{\pi}_k}  \frac{d\omega_{k}}{d\overline{\omega}_k}  \right) \\
			& + \sum_{s=1}^{t} \int m_s\left( d\overline{Q}_s - d{Q}_s - \overline{\phi}_s d\pi_s d\nu \right)d\Pb_s  \left( \prod_{r=0}^{s-1} d\overline{Q}_r d\Pb_r \right) \\
			& \lesssim  \sum_{s=1}^{t} \Vert \pi_s -  \overline{\pi}_{s} \Vert \left\{ \Vert m_s -  \overline{m}_{s}  \Vert +  \sum_{r=1}^{s} \Big(  \Vert \pi_r - \overline{\pi}_r\Vert + \Vert \omega_r - \overline{\omega}_r\Vert \Big) + \Vert \pi_s -  \overline{\pi}_{s} \Vert\right\}.
		\end{align*}
		
		Therefore, if we have
		$$
		\Big( \Vert m_{s} -  \widehat{m}_{s}  \Vert + \Vert \pi_s -  \widehat{\pi}_{s} \Vert \Big) \Big(  \Vert \pi_r - \overline{\pi}_r\Vert + \Vert \omega_r - \overline{\omega}_r\Vert \Big)  = o_\Pb(\frac{1}{\sqrt{n}}), \quad \forall r \leq s \leq t,
		$$
		then both of the remainders from Lemma \ref{lem:remainder_1} and \ref{lem:remainder_2} are diminishing at rate of $n^{-\frac{1}{2}}$ uniformly in $\delta$.
	\end{proof}
